\documentclass[10pt]{article}
\usepackage[utf8]{inputenc}
\usepackage[a4paper, total={6.5in, 10in}]{geometry}
\usepackage{hyperref}
\usepackage{graphicx}
\usepackage{amsfonts}
\usepackage{mathtools}
\usepackage{amsmath}
\usepackage{amssymb}
\usepackage{amsthm}
\usepackage{csquotes}
\usepackage{authblk}
\usepackage{multicol}

\allowdisplaybreaks

\DeclareFontFamily{U}{matha}{\hyphenchar\font45}
\DeclareFontShape{U}{matha}{m}{n}{
      <5> <6> <7> <8> <9> <10> gen * matha
      <10.95> matha10 <12> <14.4> <17.28> <20.74> <24.88> matha12
      }{}
\DeclareSymbolFont{matha}{U}{matha}{m}{n}
\DeclareFontSubstitution{U}{matha}{m}{n}

\DeclareFontFamily{U}{mathb}{\hyphenchar\font45}
\DeclareFontShape{U}{mathb}{m}{n}{
      <5> <6> <7> <8> <9> <10> gen * mathb
      <10.95> mathb10 <12> <14.4> <17.28> <20.74> <24.88> mathb12
      }{}
\DeclareSymbolFont{mathb}{U}{mathb}{m}{n}
\DeclareFontSubstitution{U}{mathb}{m}{n}

\DeclareFontFamily{U}{mathx}{}
\DeclareFontShape{U}{mathx}{m}{n}{<-> mathx10}{}
\DeclareSymbolFont{mathx}{U}{mathx}{m}{n}

\DeclareMathDelimiter{(}{\mathopen}{matha}{"70}{mathx}{"00}
\DeclareMathDelimiter{)}{\mathclose}{matha}{"71}{mathx}{"08}
\DeclareMathDelimiter{[}{\mathopen}{matha}{"72}{mathx}{"10}
\DeclareMathDelimiter{]}{\mathclose}{matha}{"73}{mathx}{"18}
\DeclareMathDelimiter{\lcurl}{\mathopen}{matha}{"74}{mathx}{"20}
\DeclareMathDelimiter{\rcurl}{\mathclose}{matha}{"75}{mathx}{"28}
\DeclareMathDelimiter{\lbbrac}{\mathopen}{matha}{"76}{mathx}{"30}
\DeclareMathDelimiter{\rbbrac}{\mathopen}{matha}{"77}{mathx}{"38}
\DeclareMathDelimiter{\langle}{\mathopen}{matha}{"78}{mathx}{"40}
\DeclareMathDelimiter{\rangle}{\mathclose}{matha}{"79}{mathx}{"44}
\DeclareMathSymbol{\dagger}{2}{matha}{"3A}

\usepackage{xcolor}
\usepackage{hyperref}
\hypersetup{
    colorlinks,
    linkcolor={green!30!black},
    citecolor={green!30!black},
    urlcolor={green!30!black}
}

\usepackage[
backend=biber,
style=alphabetic,
sorting=ynt
]{biblatex}
\title{The subcritical finite-volume massive sine-Gordon model}
\author{Jaka Pelaič}
\affil{Mathematical Institute, University of Oxford}
\date{July 2025}

\newcommand{\dd}{\mathrm{d}}
\newcommand{\sqrtbeta}{\sqrt{\hspace{-1.16pt}\beta\hspace{1.16pt}}}
\newcommand{\sqrtbetae}{\sqrt{\hspace{-.5pt}\beta\hspace{.5pt}}}

\newtheorem{thm}{Theorem}[section]
\newtheorem{defn}[thm]{Definition}
\newtheorem{lem}[thm]{Lemma}
\newtheorem{cor}[thm]{Corollary}
\newtheorem{prop}[thm]{Proposition}

\begin{document}

\maketitle
\begin{abstract}
    We present a construction of the finite-volume massive sine-Gordon model in the UV subcritical regime using a renormalization group method. The resulting measure has Gaussian tails, respects toroidal symmetries and is reflection-positive.
\end{abstract}

\begin{frame}{}
\begin{multicols}{2}
  \tableofcontents
  \end{multicols}
\end{frame}

\section{Introduction}

\subsection{Initial remarks}

The sine-Gordon model is quite popular in the statistical mechanics and quantum field theory literature. For an early study see \cite{Coleman1975}. The first rigorous renormalization group study of the model began with the Roman school in the series of papers \cite{Benfatto1982, Nicol1983, Nicol1986, Renn1986}. Soon thereafter \cite{Brydges1987} studied the model by a direct analysis of the Polchinski equation, but could not cover the full subcritical regime. With \cite{BrydgesYau} the foundation of an alternative method was laid. \cite{Dimock1993} was the first attempt at controlling the full subcritical regime in this approach, but unfortunately an error spoiled the results. The later \cite{Dimock2000} also contains an error.\footnote{This was only spotted recently by the author. The issue is that (56) does not follow from Lemma 6 as claimed. It is probably possible to fix the proof by switching back to the more local norms used in \cite{Dimock1993}, but we are not yet in a position to adjudicate whether this would cause any other problems. The present work is an alternative solution which avoids the problematic cluster expansion step altogether. This does come at a price: without a cluster expansion, it is not clear how to implement the usual leveraging of the mass for good control of the infinite volume behavior.} The next substantial development was the injection into the RG framework of the finite-range decomposition \cite{Mitter2000, Brydges2003}. Since then there has been further investigation of this and closely related models, especially in the critical phase (cf. e.g. \cite{Falco1,Falco2,Bauerschmidt2024I,Bauerschmidt2024II}).

The present work is primarily a technical note: to the best of our knowledge, in the available literature there is currently no complete construction of the model with Brydges-Yau-type methods in the full UV-subcritical (or superrenormalizable) regime, even in finite volume (this might be considered a folklore result). We also take this opportunity to establish some additional properties of the limiting measures. It should be possible, perhaps with substantial modifications of the method, to extend (at least some of) our results to infinite volume, which question we hope to revisit in the future.

Our main references were \cite{Dimock2000, ParkCity, Falco1}. If anything is ever unclear, it is a good bet that it is covered in one of these works.

\subsection{Structural remarks}

This paper is designed to be read as follows. The main numbered sections constitute a minority of the text, since there are many technical notions and results that we have relegated to the appendices. When we think it is time for the reader to familiarize themselves with the contents of a given appendix, we make a note of this in the body of the text. From that point onward, it is assumed that the reader understands the given appendix and can independently refer to it when they feel the need to refresh their memory.

\subsection{The model}

We study the sine-Gordon measures. Let $\beta\in[4\pi,8\pi)$ and $\zeta$ be a bounded Borel function on the unit torus $\mathbb{T}^2$.\footnote{The restriction on $\beta$ away from $0$ is not necessary, we have made it to avoid constants depending on $\beta^{-1}$ which seems to only be a distraction. The interested reader can easily modify the results herein to cover this case, which can anyway be treated by other methods giving stronger results.} Furthermore pick $\varepsilon > 0$ and then let $\gamma[C_\varepsilon]$ be the law of the mollified unit mass\footnote{We put in a mass to keep things as simple as possible, primarily to avoid issues with the zero mode which arise in the massless model. Setting it equal to $1$ is convenient because it equals the width of the torus. Of course, this is not a \enquote{generic} choice, but the additional complications in case the mass is not fixed to the torus size did not seem illuminating enough to warrant investigation in this note.} GFF $\chi_\varepsilon * \Phi$ ($\chi_\varepsilon$ is a not-too-smooth radially symmetric compactly supported mollifier with integral $1$, see Appendix \ref{AppMolly}). Put
\begin{equation}
    \mu_{\mathrm{sG}}[\varepsilon,\zeta](\dd\phi) = \exp\!\bigg(\int_{\mathbb{T}^2}\zeta(x) \big\lbbrac\!\cos\!\big(\sqrtbeta\phi(x)\big)\hspace{-1pt}\big\rbbrac_{\varepsilon}\dd x + E[\varepsilon,\zeta]\bigg)\gamma[\mathrm{C}_\varepsilon](\dd\phi).
\end{equation}
$E[\varepsilon,\zeta]$ is chosen to make $\mu_{\mathrm{sG}}$ a probability measure and the double brackets denote Wick ordering:
\begin{equation}
\big\lbbrac\!\cos\!\big(\sqrtbeta\phi(x)\big)\hspace{-1pt}\big\rbbrac_{\varepsilon}\, = e^{\frac{1}{2}\beta \mathrm{C}_\varepsilon(0)}\cos\!\big(\sqrtbeta\phi(x)\big).
\end{equation}
$\mu_{\mathrm{sG}}$ does make sense because $\mathrm{C}_\varepsilon$ is a UV-mollified GFF covariance with cutoff $\varepsilon$ so that the Gaussian measure $\gamma[\mathrm{C}_\varepsilon]$ is at least $\mathcal{C}^0$-valued. If $\mathrm{T}$ is a symmetry of the torus,
\begin{equation}
    \mathrm{T}_*\mu_\mathrm{sG}[\varepsilon,\zeta] = \mu_\mathrm{sG}[\varepsilon,\mathrm{T}\zeta].\label{EuclideanCovariance}
\end{equation}

At this point we fix some constants. First, $\sigma \coloneqq 1-\frac{\beta}{8\pi} \in (0,\frac{1}{2}]$ is the distance to criticality. $\mathfrak{d}$ is an integer of order (say) $10$ and determines how many derivatives $\chi_{\varepsilon}$ has; equivalently, its decay in momentum space. $\delta$ is a small positive number which could be fixed at, say, $\delta = \frac{1}{100}$ throughout, $\ell \gtrsim_{\sigma^{-1},\delta^{-1}}1$ is an integer, $L =\ell^{\mathfrak{m}}$ where $\mathfrak{m} = \ell^2$, $A\gtrsim_L 1$, and $\zeta_*\lesssim_{A^{-1}}1$. We also introduce $\mathfrak{c}^{(1)},\mathfrak{c}^{(2)},\mathfrak{c}^{(3)},\mathfrak{c}^{(4)}$ with $\mathfrak{c}^{(2)}\lesssim_1 1\lesssim_1 \mathfrak{c}^{(1)}\wedge\mathfrak{c}^{(3)}$, $\mathfrak{c}^{(4)}\gtrsim_1\mathfrak{c}^{(3)}$, and $\mathfrak{h}_*\gtrsim_1 1$, $\kappa_* = \mathfrak{c}^{(4)} \log L$. $\zeta_\bullet$ is a completely arbitrary positive number and will represent the upper bound on the size of the coupling $\zeta$ as measured in $\mathrm{L}^\infty$. $j_\bullet$ is a negative integer defined as
\begin{equation}
    j_\bullet = \bigg\lfloor \sigma^{-1}\log_L\!\bigg(\frac{\zeta_*}{\zeta_\bullet}\bigg)\bigg\rfloor \wedge 0
\end{equation}
and will serve as the terminus for the UV RG iteration; it is chosen so that if $\Vert\zeta\Vert_{\mathrm{L}^\infty} \leqslant \zeta_\bullet$, the effective coupling at scale $j_\bullet$ is bounded by $\zeta_*$. Write $\mathrm{B}_\bullet=\zeta_\bullet\,\overline{\mathrm{Ball}}(\mathrm{L}^\infty)$ topologized as a subset of $\mathrm{L}^\wp$ where $\wp = 5\sigma^{-1}$ (anything strictly larger than $4\sigma^{-1}$ would work); it is a closed subset. Finally, we introduce a maximal mollification scale
\begin{equation}
    \varepsilon_\bullet = \frac{1}{6}L^{j_\bullet-1}.
\end{equation}

An important convention we will follow is that $\mathcal{O}_{x_1,\ldots,x_n}$ denotes some constant depending on the specified parameters and this can change from line to line, which spares us from naming the various transient quantities that appear in estimates and cannot be comfortably absorbed with the $\lesssim$ notation.

The Laplace transforms for $J\in\mathcal{C}^\infty(\mathbb{T}^2)$ are
\begin{equation}
    \mathrm{Z}[\varepsilon,\zeta](J) = \int e^{\langle J,\phi\rangle}\mu_{\mathrm{sG}}[\varepsilon,\phi](\dd\phi).
\end{equation}
These are Fr\'echet-analytic functionals on $\mathcal{C}^\infty(\mathbb{T}^2)$. To measure their size, we introduce the space $\mathfrak{C}$ of continuous functionals on $\mathcal{C}^\infty(\mathbb{T}^2)$ with finite norm
\begin{equation}
    \Vert F\Vert_{\mathfrak{C}} = \sup_{J\in\mathcal{C}^\infty(\mathbb{T}^2)} e^{-\gamma\Vert J\Vert^2_*}\vert F(J)\vert;
\end{equation}
here $\gamma$ is $\mathcal{O}_1$ and $\Vert\!\cdot\!\Vert_*$ is a Sobolev norm with sufficiently large regularity (e.g. $10$ derivatives will certainly suffice). We also define the subset $\mathfrak{L}\subset\mathfrak{C}$ of those functionals which are Laplace transforms of a probability measure. It turns out that $\mathfrak{L}$ is closed, convergence in $\mathfrak{L}$ implies weak convergence of the corresponding probability measures on $\mathcal{D}'(\mathbb{T}^2)$, and these measures all have Gaussian tails (see Appendix \ref{AppLaplace}). Now we have everything necessary to state the central result.

\begin{prop}\label{mainThrm}
    The map $\mathrm{Z}[\cdot,\cdot]: (0,\varepsilon_\bullet] \times \mathrm{B}_\bullet \longrightarrow \mathfrak{C}$ is Lipschitz.
\end{prop}

The point of this proposition is that $\mathrm{Z}[\cdot,\cdot]$ automatically extend to the completion $[0,\varepsilon_\bullet] \times \mathrm{B}_\bullet$, i.e. we get a UV limit. We also write $\mathrm{Z}[0,\zeta]$ for the corresponding limiting objects. Note that $\zeta_\bullet$ was arbitrary, so we in fact get this result for all couplings, but the Lipschitz constant is not uniform, so we have found this the most natural formulation of the result.

The analysis we do should also carry through to small complex coupling. This would show that $\mathrm{Z}[0,\zeta]$ have analytic extensions in $\zeta$ around $\zeta = 0$, which implies that the usual perturbation theory is convergent, a non-trivial result (though we emphasize that it is much less clear if in the complex case, we get corresponding complex-valued measures; this is because Levy's theorem for complex measures is more delicate). We have elected to restrict just to real coupling to avoid more complications regarding the finer points of infinite-dimensional holomorphy.

The \enquote{toroidal} OS axioms follow:

\begin{cor}
Let $\zeta$ be constant. The probability measure $\mu_\mathrm{sG}[0,\zeta]$ is reflection-positive and invariant under toroidal symmetries. If $0<\zeta \leqslant \zeta_*$, it is also non-Gaussian.
\end{cor}

\begin{proof}
    The covariance property (\ref{EuclideanCovariance}) transfers by weak convergence to the limiting measures and so
    \begin{equation}
        \mathrm{T}_*\mu_\mathrm{sG}[0,\zeta] = \mu_\mathrm{sG}[0,\mathrm{T}\zeta] = \mu_\mathrm{sG}[0,\zeta].
    \end{equation}
    For reflection positivity, view $\mathbb{T}$ as $(-1/2, 1/2]^2$ and define $A_n$ to be the subset $(-1/2,1/2]\times[n^{-1},1/2-n^{-1}]$. Let $\zeta_n = \zeta\,\mathbb{I}_{A_n \cup \vartheta A_n}$ be the coupling which equals $\zeta$ on $A_n \cup \vartheta A_n$ ($\vartheta$ being the reflection) and vanishes otherwise. As $\zeta_n\rightarrow\zeta$ in $\mathrm{L}^\wp$, it suffices to establish RP for the approximating measures $\mu_\mathrm{sG}[0,\zeta_n]$. Let $F(\phi) = \sum_i c_i e^{\langle J_i,\phi\rangle}$ be an exponential functional which depends only on the field values in $(-1/2,1/2]\times[\varepsilon',1-\varepsilon']$ (i.e. $\mathrm{supp}(J_i)\subset\mathbb{R}\times(-1/2,1/2]\times[\varepsilon',1-\varepsilon']$). Then if $\varepsilon<\varepsilon'$,
    \begin{align}
        \langle\vartheta F,F\rangle&_{\mathrm{L}^2(\mu_{\mathrm{sG}}[\varepsilon,\zeta_n])} = \mathbb{E}\bigg[\overline{F(\vartheta\Phi_{\hspace{-.5pt}\varepsilon})}F(\Phi_{\hspace{-.5pt}\varepsilon})\exp\!\bigg(\int_{\mathbb{T}^2}\zeta_n(x) 
        \big\lbbrac\!\cos\!\big(\sqrtbeta\Phi_{\hspace{-.5pt}\varepsilon}(x)\big)\hspace{-1pt}\big\rbbrac_{\varepsilon}\dd x + E[\varepsilon,\zeta]\bigg)\bigg]\nonumber\\
        &=e^{E[\varepsilon,\zeta_n]}\mathbb{E}\!\left[\overline{F(\vartheta\Phi_{\hspace{-.5pt}\varepsilon})\exp\!\bigg(\zeta\int_{A_n} \big\lbbrac\!\cos\!\big(\sqrtbeta\vartheta\Phi_{\hspace{-.5pt}\varepsilon}\big)\hspace{-1pt}\big\rbbrac_{\varepsilon}\!\bigg)}F(\Phi_{\hspace{-.5pt}\varepsilon})\exp\!\bigg(\zeta\int_{A_n} \big\lbbrac\!\cos\!\big(\sqrtbeta\Phi_{\hspace{-.5pt}\varepsilon}\big)\hspace{-1pt}\big\rbbrac_{\varepsilon}\!\bigg)\right]\nonumber\\
        &\geqslant 0
    \end{align}
    where the bound follows by the reflection positivity of $\mathrm{C}$. By taking $\varepsilon\rightarrow 0$, since
    \begin{equation}
        \langle\vartheta F,F\rangle_{\mathrm{L}^2(\mu_{\mathrm{sG}}[\varepsilon,\zeta_n])} = \sum_{i,j}c_i\overline{c_j}\,\mathrm{Z}[\varepsilon,\zeta_n](J_i+\vartheta J_j)
    \end{equation}
    and $\varepsilon'$ was arbitrary, we obtain $\langle\vartheta F,F\rangle_{\mathrm{L}^2(\mu_{\mathrm{sG}}[0,\zeta_n])}\geqslant 0$ for all $F$ which only depend on the field values in the upper half of $\mathbb{T}^2$ at some positive distance from the circles $(-1/2,1/2]\times \lcurl 0,1/2\rcurl$. But $\mathrm{Z}[0,\zeta_n]$ is continuous in $J$, so the positivity immediately transfers to all $F$ supported on the upper half of $\mathbb{T}^2$.

    Non-Gaussianity follows from the detailed structure of the Laplace transform, see Section \ref{RemSteps}.
\end{proof}

\section{Renormalization group flow}

\subsection{Progressive integration}

We explicate a method to iteratively compute the integral that defines $\mathrm{Z}[\varepsilon,\zeta]$, giving us the ability to analyze these Laplace transforms sufficiently accurately to obtain the stated results. In this section, we just sketch the basic ideas, deferring to later sections for precise justification.

$\mathrm{C}_\varepsilon$ admits a scale decomposition into $\Gamma_{\hspace{-1pt}\varepsilon;j}$ with corresponding fields $\Psi_{\hspace{-1pt}\varepsilon;j}$. Refer to Appendix \ref{AppScale}. Therefore by shifting the integration,
\begin{align}
    \mathrm{Z}&[\varepsilon,\zeta](J) \nonumber\\
    &= e^{\frac{1}{2}\langle\mathrm{C}_\varepsilon J,J\rangle+E[\varepsilon,\zeta]}\mathbb{E}\bigg[\exp\!\bigg(\int_{\mathbb{T}^2}\zeta
    \big\lbbrac\!\cos\!\big(\sqrtbeta(\mathrm{S}_{L^{-\mathsf{N}-1}}\Psi_{\hspace{-1pt}\varepsilon;-\mathsf{N}-1}+\ldots+\mathrm{S}_{L^{-1}}\Psi_{\hspace{-1pt}\varepsilon;-1}+\Psi_{\hspace{-1pt}\varepsilon;0}+\mathrm{C}_\varepsilon J)\big)\hspace{-1pt}\big\rbbrac_{\varepsilon}\!\bigg)\bigg].
\end{align}
This is legitimate because $\mathrm{C}_{\varepsilon}$ has polynomial decay in momentum space, hence any smooth $J$ is in its Cameron-Martin space. Now we define
\begin{equation}
    \zeta_{\varepsilon;j} = L^{2j}e^{\frac{1}{2}\beta\sum_{k=j}^0\Gamma_{\hspace{-1pt}\varepsilon;k}(0)}\mathrm{S}_{L^{-j}}\zeta,\quad\text{i.e.}\quad \zeta_{\varepsilon;j+1} = L^2e^{-\frac{1}{2}\beta\Gamma_{\hspace{-1pt}\varepsilon;j}(0)}\mathrm{S}_{L^{-1}}\zeta_{\varepsilon;j}.\label{zetajdef}
\end{equation}
The flow equation for $j\leqslant -1$ is
\begin{align}
    \mathcal{Z}_{-\mathsf{N}-1}[\varepsilon,\zeta](\phi) &= \exp\!\bigg(\int_{\mathbb{T}_{-\mathsf{N}-1}^2}\zeta_{\varepsilon;-\mathsf{N}-1}\cos\!\big(\sqrtbeta\phi\big)\bigg),\nonumber\\
    \mathcal{Z}_{j+1}[\varepsilon,\zeta](\phi) &= \mathbb{E}\big[\mathcal{Z}_{j}(\mathrm{S}_L\phi + \Psi_{\hspace{-1pt}\varepsilon;j})\big].
\end{align}
By induction, it is easily verified that
\begin{equation}
    \mathrm{Z}[\varepsilon,\zeta](J) = e^{\frac{1}{2}\langle\mathrm{C}_{\varepsilon}J,J\rangle+E[\varepsilon,\zeta]}\mathbb{E}\bigg[\mathcal{Z}_j[\varepsilon,\zeta]\bigg(\sum_{k=j}^0\mathrm{S}_{L^{k-j}}\Psi_{\hspace{-1pt}\varepsilon;k}+\mathrm{S}_{L^{-j}}\mathrm{C}_\varepsilon J\bigg)\bigg].
\end{equation}

\subsection{The ansatz}

To actually solve the iteration, we need a good ansatz, which is provided by polymer activities. Refer to Appendices \ref{AppPolyBasic}, \ref{AppCrReg}, \ref{AppPolymerAct}. At this point we do the algebraic manipulations which however are not yet justified since we do not know if the objects involved have enough integrability, so we simply assume that they do. This will be verified once we introduce appropriate norms. Anyway, with $K_j$ a polymer activity,
\begin{align}
    \mathcal{Z}_j[\varepsilon,\zeta](\phi)&=e^{\mathcal{E}_j[\varepsilon,\zeta]}\sum_{X\in\mathcal{P}_j}e^{V_j[\varepsilon,\zeta](X^c,\phi)}\prod_{Y\in\mathcal{C}(X)}K_j[\varepsilon,\zeta](Y,\phi),\nonumber\\
    V_j(X,\phi)&= \int_{\mathbb{T}_j^2}\zeta_{\varepsilon;j} \cos\!\big(\sqrtbeta\phi\big).
\end{align}
Now we compute the RG map. That is, we want to compute $\mathbb{E}\big[\mathcal{Z}_j[\varepsilon,\zeta](\mathrm{S}_L\phi + \Psi_{\hspace{-1pt}\varepsilon;j})\big]$. We wish to represent $\mathcal{Z}_{j+1}$ as
\begin{equation}
    \mathcal{Z}_{j+1}[\varepsilon,\zeta] = e^{\mathcal{E}_{j+1}[\varepsilon,\zeta]}\sum_{X\in\mathcal{P}_{j+1}}e^{V_{j+1}[\varepsilon,\zeta](X^c,\phi)}\prod_{Y\in\mathcal{C}(X)}K_{j+1}[\varepsilon,\zeta](Y,\phi).
\end{equation}
First put
\begin{equation}
    \tilde{K}_j[\varepsilon,\zeta](X,\phi,\psi) = \mathbb{I}_{X\in\mathcal
    {P}_c}\sum_{Y\in\mathcal{P}_j}^{\bar{Y}^L = X}e^{V_j[\varepsilon,\zeta](L X-Y,\mathrm{S}_L\phi+\psi)}\prod_{Z\in\mathcal{C}(Y)}K_{j+1}[\varepsilon,\zeta](Z,\mathrm{S}_L\phi+\psi).
\end{equation}
An important point is that there is some ambiguity in the representation of $\mathcal{Z}_j$: terms can be moved out of $K_j$ and into $\mathcal{E}_j$. We will need to make use of this ambiguity to track some \enquote{bad} terms by removing them from $K_j$ before evaluating the expectation, otherwise $K_j$ would grow uncontrollably along the flow. This step is usually called extraction. To perform it, we take an activity $Q_j[\varepsilon,\zeta](X)$ which is nonzero only on small $X$ and does not depend on the field.\footnote{Field dependence is the generic situation in other models; it poses no further issues by itself. The resulting fact that the counterterms in this model are field-independent is nice because it does not spoil the duality with a Yukawa or Coulomb gas.} We also write $\delta\mathcal{E}_j[\varepsilon,\zeta](X) = \sum_{B\in\mathcal{B}_j(X)}\delta\mathcal{E}_j[\varepsilon,\zeta](B)$ for an activity which will correct $\mathcal{E}_j$ in a way which compensates the subtraction of $Q_j$. Finally, as an intermediate quantity, define for any $Y\in\mathcal{S}_{j+1}$ and $B\in\mathcal{B}_{j+1}(Y)$
\begin{equation}
    J_j[\varepsilon,\zeta](Y,B) = \sum_{D\in\mathcal{B}_j(L B),X\in\,\mathcal{S}_j}^{D\subset X,\bar{X}^L = Y}\frac{1}{\vert X\vert}Q_j[\varepsilon,\zeta](X) - \mathbb{I}_{Y=B}\sum_{Y'\in\mathcal{S}}\sum_{D\in\mathcal{B}_{j}(L B),X\in\,\mathcal{S}_j}^{D\subset X,\bar{X}^L = Y'}\frac{1}{\vert X\vert}Q_j[\varepsilon,\zeta](X).
\end{equation}
Note that this vanishes unless $B\subset Y\in\mathcal{S}$. The second term is only there to ensure
\begin{equation}
    \sum_{Y\in\mathcal{S}}J_j[\varepsilon,\zeta](Y,B) = 0\label{Jcancellation}.
\end{equation}
Another identity that $J_j$ satisfies is
\begin{equation}
    \sum_{B\in\mathcal{B}_{j+1}(Y)}J_j[\varepsilon,\zeta](Y,B)=\sum_{X\in\mathcal{S}_j}^{\bar{X}^L=Y}Q_j[\varepsilon,\zeta](X)-\mathbb{I}_{\vert Y\vert = 1}\sum_{D\in\mathcal{B}_j(L Y),X\in\mathcal{S}_j}^{D\subset X}\frac{1}{\vert X\vert}Q_j[\varepsilon,\zeta](X).
\end{equation}
For notational convenience, if $(Z,\mathfrak{B})$ is a pair of a polymer and a choice $\mathfrak{B}=(\mathfrak{B}_Y)_{Y\in\mathcal{C}(Z)}$ of a square $\mathfrak{B}_Y\in\mathcal{B}_{j+1}(Y)$ in every connected component of $Z$, we write
\begin{equation}
    J_j^{\mathfrak{B}\rightarrow Z}[\varepsilon,\zeta]=\prod_{Y\in\mathcal{C}(Z)}J_j[\varepsilon,\zeta](Y,\mathfrak{B}_Y).
\end{equation}
Finally, for $X\in\mathcal{P}_{j+1}$ define
\begin{align}
    P_j^X[\varepsilon,\zeta](\phi,\psi) &= \prod_{B\in\mathcal{B}_{j+1}(X)}\big(e^{V_j[\varepsilon,\zeta](L B,\mathrm{S}_L\phi+\psi)}-e^{V_{j+1}[\varepsilon,\zeta](B,\phi)+\delta\mathcal{E}_j[\varepsilon,\zeta](L B)}\big),\\
    R_j^X[\varepsilon,\zeta](\phi,\psi)&=\prod_{Y\in\mathcal{C}(X)}\bigg(\tilde{K}_j[\varepsilon,\zeta](Y,\phi,\psi)-\sum_{B\in\mathcal{B}_{j+1}(Y)}J_j[\varepsilon,\zeta](Y,B)\bigg).
\end{align}
From now we omit the arguments $\varepsilon,\zeta$ if context permits. Plugging all these definitions into the formula for $\mathcal{Z}_{j+1}$ gives
\begin{equation}
    \mathcal{Z}_{j+1}(\phi) = e^{\mathcal{E}_j}\sum_{X_1,X_2,Y_1,\mathfrak{B}\rightarrow Y_2}^*e^{V_{j+1}(X_1,\phi)+\delta\mathcal{E}_j(LX_1)}\mathbb{E}\big[P_j^{X_2}(\phi,\Psi_{\hspace{-1pt}j})R_j^{Y_1}(\phi,\Psi_{\hspace{-1pt}j})\big]J_j^{\mathfrak{B}\rightarrow Y_2}.
\end{equation}
The $*$-sum is over quadruples in $\mathcal{P}_{j+1}^{\times 3}\times(\mathfrak{B}\rightarrow\mathcal{P}_{j+1})$ for which $X_1\cup X_2\cup Y_1 \cup Y_2=\mathbb{T}_{j+1}^2$, all of these sets have disjoint interiors and $Y_1,Y_2$ are completely disjoint. We make one last adjustment to exhibit linear cancellations. Define $X = X_2\cup Y_1\cup\bigcup\mathfrak{B}^*$. Then $X_1 = X^c\cup(X - X_2\cup Y_1\cup Y_2)$ and we find the final form
\begin{align}
    \mathcal{Z}_{j+1}(\phi) = e&^{\mathcal{E}_j + \delta\mathcal{E}_j(\mathbb{T}_j^2)}\sum_{X\in\mathcal{P}_{j+1}}e^{V_{j+1}(X^c,\phi)}\cdot\\
    &\prod_{X'\in\mathcal{C}(X)}\sum_{(Z,Y_1,\mathfrak{B}\rightarrow Y_2)\rightarrow X'}e^{V_{j+1}(X'-Z\cup Y_1\cup Y_2,\phi)-\delta\mathcal{E}_j(L(Z\cup Y_1\cup Y_2))}\mathbb{E}\big[P_j^{Z}(\phi,\Psi_{\hspace{-1pt}j})R_j^{Y_1}(\phi,\Psi_{\hspace{-1pt}j})\big]J_j^{\mathfrak{B}\rightarrow Y_2}.\nonumber
\end{align}
The expectations have partially factorized due to the finite-range property. The sum in the second line is over triples in $\mathcal{P}_{j+1}^{\times 2}\times(\mathfrak{B}\rightarrow\mathcal{P}_{j+1})$ with $Z,Y_1,Y_2$ disjoint and $X = Z\cup Y_1\cup\bigcup\mathfrak{B}^*$. We remark that $Y_2 = X$ is impossible since it would force $X\subsetneq \mathfrak{B}_W^*$ for at least one $W\in\mathcal{C}(Y_2)$. Anyway, we now pin down the definition of the RG flow.

\begin{defn}
   Let $(\varepsilon,\zeta)\in(0,\varepsilon_0]\times\mathrm{B}_{\mathrm{cs}}(\zeta_*)$. The RG flow is the dynamical system trajectory
   \begin{equation}
       (\zeta_j,K_j,\mathcal{E}_j)_{j=-\mathsf{N}-1}^0
   \end{equation}
   with evolution ($X$ connected)
   \begin{align}
       \zeta_{j+1} &= L^2e^{-\frac{1}{2}\beta\Gamma_{\hspace{-1pt}j}(0)}\mathrm{S}_{L^{-1}}\zeta_j,\\
       K_{j+1}(X,\phi) &= \sum_{(Z,Y_1,\mathfrak{B}\rightarrow Y_2)\rightarrow X}e^{V_{j+1}(X-Z\cup Y_1\cup Y_2,\phi)-\delta\mathcal{E}_j(L(Z\cup Y_1\cup Y_2))}\mathbb{E}\big[P_j^{Z}(\phi,\Psi_{\hspace{-1pt}j})R_j^{Y_1}(\phi,\Psi_{\hspace{-1pt}j})\big]J_j^{\mathfrak{B}\rightarrow Y_2},\label{Kflow}\\
       \mathcal{E}_{j+1} &= \mathcal{E}_{j} +\delta\mathcal{E}_j(\mathbb{T}_j^2).
   \end{align}
\end{defn}

We remark that, because $Q_j$ does not depend on $\phi$, $K_j$ depends only on $\phi$ within $X$. Also, the initial step simplifies drastically:
\begin{align}
       K_{-\mathsf{N}}(X,\phi) &= \mathbb{E}\big[P_{-\mathsf{N}-1}^{X}(\phi,\Psi_{\hspace{-1pt}-\mathsf{N}-1})\big],\\
       \mathcal{E}_{-\mathsf{N}} &= 0.
   \end{align}
Anyway, to better understand what this map is doing, we look at the linearization around the trivial RG flow when $\zeta = 0$ (this is the flow of the GFF). We assume that $\delta\mathcal{E}_j$ and $Q_j$ will be linear in $K_j$. Then the linearization is
\begin{equation}
    \mathcal{L}K_j(X,\phi) = \sum_{Y\in\mathcal{P}_{\mathrm{c};j}}^{\bar{Y}^L=X}\big(\mathbb{E}\big[K_j(Y,\mathrm{S}_L\phi+\Psi_{\hspace{-1pt}j}\big]-Q_j(Y)\big) + \mathbb{I}_{\vert X\vert = 1}\sum_{B\in\mathcal{B}_j(L X)}\bigg(\sum_{Y\in\mathcal{S}_j}^{B\subset Y}\frac{1}{\vert Y\vert}Q_j(Y)-\delta\mathcal{E}_j(B)\bigg).
\end{equation}
To obtain this, we have used (\ref{Jcancellation}) and also the identity $\mathbb{E}\big[V_j(L X,\mathrm{S}_L+\Psi_{\hspace{-1pt}j})\big] = V_{j+1}(X,\phi)$, which is the reason for our choice of $\zeta_j$ in the first place. We can define $\delta\mathcal{E}_j$ to remove the last term and then we may choose $Q_j$ to cancel some terms in $K_j$.

To identify these problematic terms, we first do a Fourier analysis in the zero mode of the field. That is, note that $K_j(X,\phi + 2\pi\beta^{-1/2}) = K_j(X,\phi)$ (this is clearly preserved by the flow); then put
\begin{equation}
    k_j^{(q)}(X,\phi) = \frac{1}{2\pi}\int_{-\pi}^\pi e^{-i q x}K_j(X,\phi + \beta^{-1/2}x)\,\dd x.
\end{equation}
We will see that the kinds of norms preserved by the flow in fact imply that $K_j(X)$ is analytic on a strip (refer to Appendix \ref{AppComplexification}), so $k_j^{(q)}$ are well-decaying in $\vert q\vert$ and we also have
\begin{align}
    K_j(X,\phi) &= \sum_{q\in\mathbb{Z}}k_j^{(q)}(X,\phi),\\
    k_j^{(q)}(X,\phi+x) &= e^{i\sqrtbetae qx}k_j^{(q)}(X,\phi)
\end{align}
with the second making sense even for complex $x$ with $\mathrm{Im}(x)$ not too large. In any case, the extraction $Q_j$ is taken from the neutral ($q = 0$) terms:
\begin{equation}
    Q_j(X) = \mathbb{I}_{X\in\mathcal{S}}\mathbb{E}\big[k_j^{(0)}(X,\Psi_{\hspace{-1pt}j})\big].
\end{equation}
Having thus retained control of the RG up to the last step, we are left with handling $\Gamma_{\hspace{-1pt}0}$ which is slightly different but ultimately easier.

This concludes the basic exposition of the setup. In the next section we start working up to analytic estimates on the RG flow.

\subsection{Analytic control}

To study the RG flow, we wish to control
\begin{equation}
    \Vert K_j[\varepsilon,\zeta]-K_j[\dot\varepsilon,\dot\zeta]\Vert_{\mathfrak{h},\kappa,A}.
\end{equation}
That is, we are comparing the fluctuation step where we allow the fluctuation covariance to change; this is important because we want to compare different flows to obtain convergence. 

We will always assume $\dot\varepsilon\leqslant\varepsilon$ with corresponding $\dot{\mathsf{N}}\geqslant\mathsf{N}$, and we will mostly deal with $j\geqslant \mathsf{-N}$. Since the case $j = -\mathsf{N}-1$ has slightly worse estimates, we need to treat it separately in principle, but we have tried to make the presentation as seamless in this regard as we could manage; there are only some qualifications in the propositions below to deal with this border case at the specific parameter values with which it actually arises.

Since there will be a lot of estimates, we make some notational abbreviations. We write
\begin{align}
    V_j = V_j[\varepsilon,\zeta],&\quad\dot{V}_j = V_j[\dot\varepsilon,\dot\zeta],\nonumber\\
    K_j = K_j[\varepsilon,\zeta],&\quad\dot{K}_j = K_j[\dot\varepsilon,\dot\zeta],
\end{align}
and similarly for $J_j[\varepsilon,\zeta]$, $P^X_j[\varepsilon,\zeta]$, etc. In case a mixed argument appears, such as $K_j[\varepsilon,\dot\zeta]$, we will not use an abbreviation. Regarding norms, if $x_1,\ldots,x_n$ belong to the same normed space, write
\begin{equation}
    \Vert x_1,\ldots,x_n\Vert = \max_{i=1,\ldots,n}\Vert x_i\Vert.
\end{equation}
Regarding the various parameters, we assume $\mathfrak{h}\in[\mathfrak{h}_*,2\mathfrak{h}_*]$ and  $\kappa \in[\kappa_*,2\kappa_*]$.

As a general rule of thumb, an estimate on a quantity such as $\Vert K_j-\dot{K}_j\Vert_{\ldots}$ should reduce to an estimate of the form $\Vert K_j\Vert_{\ldots}$ simply by taking $\dot\zeta = 0,\dot\varepsilon = \varepsilon$, which kills $\dot K_j$ and any possible dependence on $\varepsilon$. This is not quite correct because in the general difference estimates, the right-hand side will sometimes involve $2\mathfrak{h}$ whereas the left involves only $\mathfrak{h}$, while this is not the case in the non-difference estimate, so speaking precisely, one has to derive the latter separately; but we will not write this out, since the derivation is always a simpler version of the difference estimate. As a sanity check, the reader can consider the next proposition and reason whether the difference estimates do reduce to their non-difference versions in this way.

\begin{prop}\label{MainRgEstimates}
Take $-\mathsf{N}\leqslant j<j_\bullet$. The potential satisfies
\begin{align}
    \Vert V_j-\dot{V}_j\Vert_{\mathfrak{h};\langle X,\phi\rangle}&\lesssim_1\vert X\vert\big(\Vert\zeta_{0;j}-\dot{\zeta}_{0;j}\Vert_{\wp} + L^{-j}\vert \varepsilon-\dot{\varepsilon}\vert \Vert\zeta_{0;j},\dot{\zeta}_{0;j}\Vert_{\wp}\big),\\
    \Vert e^{V_j}-e^{\dot{V}_j}\Vert_{\mathfrak{h};\langle X,\phi\rangle}&\lesssim_1 2^{\vert X\vert}\big(\Vert\zeta_{0;j}-\dot{\zeta}_{0;j}\Vert_{\wp} + L^{-j}\vert \varepsilon-\dot{\varepsilon}\vert \Vert\zeta_{0;j},\dot{\zeta}_{0;j}\Vert_{\wp}\big).
\end{align}
Assuming furthermore that we have inductively established $K_j,\dot{K}_j\in\mathbb{B}^{j,1}_{2\mathfrak{h},2\kappa,A}$,
\begin{align}
    \vert Q_j(X)-\dot{Q}_j(X)\vert &\lesssim_1 \Vert K_j - \dot{K}_j\Vert_{\mathfrak{h},\kappa;\langle X\rangle} + L^{-j}\vert\varepsilon - \dot\varepsilon\vert\Vert K_j,\dot{K}_j\Vert_{\mathfrak{h},\kappa;\langle X\rangle},\\
    \vert\delta\mathcal{E}_j(B)-\delta\dot{\mathcal{E}}_j(B)\vert&\lesssim_1 \max_{B\subset X\in\mathcal{S}}\big(\Vert K_j - \dot{K}_j\Vert_{\mathfrak{h},\kappa;\langle X\rangle} + L^{-j}\vert\varepsilon - \dot\varepsilon\vert\Vert K_j,\dot{K}_j\Vert_{\mathfrak{h},\kappa;\langle X\rangle}\big),\\
    \vert J_j^{\mathfrak{B}\rightarrow X}-\dot{J}_j^{\mathfrak{B}\rightarrow X}\vert&\leqslant \mathcal{O}_A^{\vert\mathcal{C}(X)\vert}A^{-2\vert X^*\vert}\Vert K_j,\dot{K}_j\Vert^{\vert\mathcal{C}(X)\vert - 1}_{\mathfrak{h},\kappa,A}\cdot\nonumber\\
    &\hspace{20pt}\big(\Vert K_j - \dot{K}_j\Vert_{\mathfrak{h},\kappa, A} + L^{-j}\vert\varepsilon - \dot\varepsilon\vert\Vert K_j,\dot{K}_j\Vert_{\mathfrak{h},\kappa,A}\big),\\
    \Vert P_j^{(\cdot)}(\cdot,\psi)\Vert_{\mathfrak{h};\langle X,\phi\rangle} &\leqslant \mathcal{O}_L^{\vert X\vert}\big(\Vert \zeta_{0;j}\Vert_{\wp} + \Vert K_j\Vert_{\mathfrak{h},\kappa,A}\big)^{\!\vert X\vert},\\
    \Vert P^{(\cdot)}_j(\cdot,\psi)-\dot{P}^{(\cdot)}_j(\cdot,\dot{\psi})\Vert_{\mathfrak{h};\langle X,\phi\rangle}&\leqslant\mathcal{O}_L^{\vert X\vert}\big(\Vert \zeta_{0;j},\dot\zeta_{0;j}\Vert_{\wp} + \Vert K_j,\dot K_j\Vert_{\mathfrak{h},\kappa,A}\big)^{\!\vert X\vert-1}\cdot\nonumber\\
    &\hspace{20pt}\Big(\Vert\zeta_{0;j}-\dot{\zeta}_{0;j}\Vert_{\wp} + \Vert K_j-\dot{K}_j\Vert_{\mathfrak{h},\kappa;\langle(LX)^*\rangle}+\nonumber\\
    &\hspace{30pt} L^{-j}\vert\varepsilon-\dot\varepsilon\vert\big(\Vert \zeta_{0;j},\dot\zeta_{0;j}\Vert_{\wp} + \Vert K_j,\dot K_j\Vert_{\mathfrak{h},\kappa,A}\big)+\nonumber\\
    &\hspace{40pt}\Vert\psi-\dot\psi\Vert_{\mathcal{C}^2(L X)}\Vert\zeta_{0;j},\dot{\zeta}_{0;j}\Vert_{\wp}\Big),\\
    \Vert R_j^{(\cdot)}(\cdot,\psi)\Vert_{\mathfrak{h};\langle X,\phi\rangle}&\leqslant\mathcal{O}_A^{\vert\mathcal
    {C}(X)\vert}A^{-(1+\eta/3)\vert X\vert}G_\kappa(LX,\mathrm{S}_L\phi+\psi)\Vert K_j\Vert_{\mathfrak{h},\kappa,A}^{\vert\mathcal{C}(X)\vert},\\
    \Vert R_j^{(\cdot)}(\cdot,\psi)-\dot{R}_j^{(\cdot)}(\cdot,\dot\psi)\Vert_{\mathfrak{h};\langle X,\phi\rangle}&\leqslant\mathcal{O}_A^{\vert\mathcal{C}(X)\vert}A^{-(1+\eta/3)\vert X\vert}\sup_{t\in[0,1]}\!G_\kappa(LX,\mathrm{S}_L\phi+(1-t)\psi+t\dot\psi)\cdot\nonumber\\
    &\hspace{-50pt}\Vert K_j,\dot{K}_j\Vert_{\mathfrak{h},\kappa,A}^{\vert\mathcal{C}(X)\vert-1}\Big(\big(\Vert K_j,\dot{K}_j\Vert_{\mathfrak{h},\kappa,A}\Vert\zeta_{0;j}-\dot{\zeta}_{0;j}\Vert_{\wp}+\Vert K_j-\dot{K}_j\Vert_{\mathfrak{h},\kappa,A}\big)+\nonumber\\
    &\hspace{70pt}\Vert K_j,\dot{K}_j\Vert_{2\mathfrak{h},\kappa,A}\big(L^{-j}\vert\varepsilon-\dot{\varepsilon}\vert+\Vert\psi-\dot\psi\Vert_{\mathcal{C}^2(L X)}\big)\Big).
\end{align}
The bounds for $P^{(\cdot)}_j$ also hold in the case $\varepsilon = \dot\varepsilon$, $\psi = \dot\psi$, and $j = -\mathsf{N}-1$.  
\end{prop}

Note that we have merely assumed $K_j,\dot{K}_j\in\mathbb{B}^{j,1}_{2\mathfrak{h},2\kappa,A}$, though we will in fact be able to get $K_j,\dot{K}_j\in\mathbb{B}^{j,2-\delta}_{2\mathfrak{h},2\kappa,A}$. Since the proofs involved are somewhat lengthy, we move them to Appendix \ref{AppProofs}.

Now we assemble all these estimates into a bound on $\Vert K_j-\dot{K}_j\Vert_{\mathfrak{h},\kappa,A}$. We also move this proof into Appendix \ref{AppProofs}.

\begin{prop}\label{MainFlowEstimates}
Take $-\mathsf{N}\leqslant j<j_\bullet$. Suppose $K_j\in\mathbb{B}^{j,1}_{\mathfrak{h},\kappa,A}$. Then
\begin{equation}
    \Vert K_{j+1}\Vert_{\mathfrak{h},\kappa,A}\leqslant \mathcal{O}_1L^{2\sigma}\Vert K_j\Vert_{\mathfrak{h},\kappa,A} + \mathcal{O}_A\Vert\zeta_{0;j}\Vert^2_{\wp}.
\end{equation}
Now suppose $K_j\in\mathbb{B}^{j,1}_{2\mathfrak{h},2\kappa,A}$. Then
\begin{align}
    \Vert K_{j+1}&-\dot{K}_{j+1}\Vert_{\mathfrak{h},\kappa,A} \leqslant\nonumber\\
    &\mathcal{O}_1L^{2\sigma}\Vert K_j[\varepsilon,\zeta]-K_j[\varepsilon,\dot{\zeta}],K_j[\dot\varepsilon,\zeta]-K_j[\dot\varepsilon,\dot\zeta]\Vert_{\mathfrak{h},\kappa,A}+\nonumber\\
    &\mathcal{O}_A\Vert\zeta_{0;j},\dot{\zeta}_{0;j}\Vert_{\wp}\big(\Vert\zeta_{0;j}-\dot{\zeta}_{0;j}\Vert_{\wp}+\Vert K_j[\varepsilon,\zeta]-K_j[\varepsilon,\dot{\zeta}],K_j[\dot\varepsilon,\zeta]-K_j[\dot\varepsilon,\dot\zeta]\Vert_{\mathfrak{h},\kappa,A}\big)+\nonumber\\
    &\mathcal{O}_AL^{-j}\vert\varepsilon-\dot\varepsilon\vert\big(\Vert\zeta_{0;j},\dot\zeta_{0;j}\Vert_{\wp}+\Vert K_j[\varepsilon,\zeta],K_j[\varepsilon,\dot\zeta],K_j[\dot\varepsilon,\zeta],K_j[\dot\varepsilon,\zeta]\Vert_{2\mathfrak{h},2\kappa,A}\big).
\end{align}
Both bounds still hold in the case $\varepsilon = \dot\varepsilon$ and $j = -\mathsf{N}-1$.
\end{prop}

Regarding the last claim, we just recall that $K_{-\mathsf{N}-1} = 0$ so that no integrability issues can arise. Anyway, note the nice structure of the second estimate: no term involves a difference of $K_j$ with both $\varepsilon$ and $\dot{\varepsilon}$ appearing. This enables us to bootstrap the estimates, as we first assume that $\varepsilon = \dot{\varepsilon}$, get some bounds, and then reinsert the bounds into the same estimate but with $\varepsilon\neq\dot{\varepsilon}$.

\begin{cor}
    Suppose $L\gtrsim_{\sigma^{-1},\delta^{-1}}1$ (with other constants selected accordingly). We have $K_j\in\mathbb{B}^{j,2-\delta}_{2\mathfrak{h}_*,2\kappa_*,A}$ for all $j \geqslant-\mathsf{N}-1$, more precisely
    \begin{equation}
        \Vert K_j\Vert_{2\mathfrak{h}_*,2\kappa_*,A}\leqslant\Vert\zeta_{0;j}\Vert^{2-\delta}_{\wp},
    \end{equation}
    and furthermore for $j\geqslant -\mathsf{N}$,
    \begin{equation}
        \Vert K_j-\dot{K}_j\Vert_{\mathfrak{h}_*,\kappa_*,A}\leqslant\Vert\zeta_{0;j},\dot{\zeta}_{0;j}\Vert^{1-\delta}_{\wp}\big(\Vert\zeta_{0;j}-\dot{\zeta}_{0;j}\Vert_{\wp}+\Vert\zeta_{0;j},\dot\zeta_{0;j}\Vert_{\wp}L^{-j}\vert\varepsilon-\dot\varepsilon\vert\big).
    \end{equation}
\end{cor}

\begin{proof}
    The first claim is verified inductively. For $j+1$,
    \begin{equation}
        \Vert{K_{j+1}}\Vert_{2\mathfrak{h}_*,2\kappa_*,A} \leqslant \big(\mathcal{O}_1L^{2\sigma}+\mathcal{O}_A\Vert\zeta_{0;j}\Vert^\delta_{\wp}\big)\Vert\zeta_{0;j}\Vert^{2-\delta}_{\wp,(LX)^{[4]}}\leqslant\Vert\zeta_{0,j+1}\Vert^{2-\delta}_{\wp}
    \end{equation}
    as desired. For the second part, we first assume $\varepsilon = \dot\varepsilon$. Then
    \begin{align}
        \Vert K_{j+1}[\varepsilon,\zeta]-&K_{j+1}[\varepsilon,\dot\zeta]\Vert_{\mathfrak{h}_*,\kappa_*,A} \leqslant\nonumber\\
        &\mathcal{O}_1L^{2\sigma}\Vert K_{j}[\varepsilon,\zeta]-K_{j}[\varepsilon,\dot\zeta]\Vert_{\mathfrak{h}_*,\kappa_*,A}+\nonumber\\
        &\mathcal{O}_A\Vert\zeta_{0;j},\dot\zeta_{0;j}\Vert_{\wp}\big(\Vert\zeta_{0;j}-\dot\zeta_{0;j}\Vert_{\wp} + \Vert K_{j}[\varepsilon,\zeta]-K_{j}[\varepsilon,\dot\zeta]\Vert_{\mathfrak{h}_*,\kappa_*,A}\big).
    \end{align}
    Inductively, in much the same way as above,
    \begin{equation}
        \Vert K_{j+1}[\varepsilon,\zeta]-K_{j+1}[\varepsilon,\dot{\zeta}]\Vert_{\mathfrak{h}_*,\kappa_*,A} \leqslant \Vert\zeta_{0,j+1},\dot\zeta_{0;j+1}\Vert^{1-\delta}_{\wp}\Vert\zeta_{0;j+1}-\dot{\zeta}_{0;j+1}\Vert_{\wp}.
    \end{equation}
    Using this, we can bootstrap the estimates and, essentially repeating the arguments, find the result.
\end{proof}

This already looks like a Lipschitz property. Finishing up is not too complicated because there are only finitely many RG steps remaining. The next section deals with those.

\subsection{Remaining steps}\label{RemSteps}

Recall that
\begin{equation}
    \mathrm{Z}(J) = e^{\frac{1}{2}\langle\mathrm{C}_{\varepsilon}J,J\rangle+E}\,\mathbb{E}\bigg[\mathcal{Z}_{j_\bullet}\!\bigg(\sum_{k=j_\bullet}^0\mathrm{S}_{L^{k-j_\bullet}}\Psi_{\hspace{-1pt}\varepsilon;k}+\mathrm{S}_{L^{-j_\bullet}}\mathrm{C}_\varepsilon J\bigg)\bigg].
\end{equation}
We define ${\mathcal{Z}}^{\mathrm{red}}_j[\varepsilon,\zeta] = e^{-E[\varepsilon,\zeta]-\mathcal{E}_j[\varepsilon,\zeta]}\mathcal{Z}_j[\varepsilon,\zeta]$. Then
\begin{equation}
    \mathrm{Z}(J) = e^{\frac{1}{2}\langle \mathrm{C}_\varepsilon J,J\rangle}\frac{\mathbb{E}\big[{\mathcal{Z}}^{\mathrm{red}}_{j_\bullet}\big(\sum_{k=j_\bullet}^0\mathrm{S}_{L^{k-j_\bullet}}\Psi_{\hspace{-1pt}\varepsilon;k}+\mathrm{S}_{L^{-j_\bullet}}\mathrm{C}_\varepsilon J\big)\big]}{\mathbb{E}\big[{\mathcal{Z}}^{\mathrm{red}}_{j_\bullet}\big(\sum_{k=j_\bullet}^0\mathrm{S}_{L^{k-j_\bullet}}\Psi_{\hspace{-1pt}\varepsilon;k}\big)\big]}.
\end{equation}

The estimates we have should allow us to get the Lipschitz property, with the only catch being that we are dividing by a quantity that we also have to be able to bound away from $0$. This is not trivial because $K_j$ may in general become negative, though of course their sum is still non-negative (even positive before taking the limit). But we first deal with the Gaussian prefactor and the numerator.

So, first compute
\begin{align}
\big\vert e^{\frac{1}{2}\langle \mathrm{C}_\varepsilon J,J\rangle} - e^{\frac{1}{2}\langle \mathrm{C}_{\dot\varepsilon} J,J\rangle} \big\vert &\leqslant \vert\langle(\mathrm{C}_{\dot\varepsilon}-\mathrm{C}_\varepsilon)J,J\rangle\vert\int_0^1 e^{\frac{1}{2}\langle((1-t)\mathrm{C}_\varepsilon+t\mathrm{C}_{\dot\varepsilon})J,J\rangle}\dd t\nonumber\\
&\lesssim_1 \vert\varepsilon-\dot\varepsilon\vert e^{\mathcal{O}_1\Vert J\Vert_*^2}.
\end{align}
For the numerator,
\begin{align}
    \Vert{\mathcal{Z}}^{\mathrm{red}}_{j_\bullet} - \dot{\mathcal{Z}}^{\mathrm{red}}_{j_\bullet}\Vert_{\mathfrak{h}_*;\langle\phi\rangle} &\leqslant\sum_{X\in\mathcal{P}_{j_\bullet}}\Vert e^{V_{j_\bullet}}-e^{\dot{V}_{j_\bullet}}\Vert_{{\mathfrak{h}_*;\langle X^c,\phi\rangle}}\prod_{Y\in\mathcal{C}(X)}\Vert K_{j_\bullet}\Vert_{{\mathfrak{h}_*;\langle Y,\phi\rangle}}+\nonumber\\
    &\hspace{15pt}\sum_{X\in\mathcal{P}_{j_\bullet}}\Vert e^{\dot{V}_{j_\bullet}}\Vert_{{\mathfrak{h}_*;\langle X^c,\phi\rangle}}\prod_{Y\in\mathcal{C}(X)}\Vert \dot{K}_{j_\bullet}-K_{j_\bullet}\Vert_{{\mathfrak{h}_*;\langle Y,\phi\rangle}}\nonumber\\
    &\lesssim_1 4^{\vert\mathbb{T}^2_{j_\bullet}\vert}G_{\kappa_*}\!(\mathbb{T}^2_{j_\bullet},\phi)\big(\Vert\zeta-\dot\zeta\Vert_\wp + \vert\varepsilon-\dot\varepsilon\vert\big).
\end{align}
Then
\begin{align}
    \bigg\Vert& \mathbb{E}\bigg[{\mathcal{Z}}^{\mathrm{red}}_{j_\bullet}\bigg(\sum_{k=j_\bullet}^0\mathrm{S}_{L^{k-j_\bullet}}\Psi_{\hspace{-1pt}\varepsilon;k}+\mathrm{S}_{L^{-j_\bullet}}(\cdot)\bigg)-\dot{\mathcal{Z}}^{\mathrm{red}}_{j_\bullet}\bigg(\sum_{k=j_\bullet}^0\mathrm{S}_{L^{k-j_\bullet}}\Psi_{\hspace{-1pt}\dot\varepsilon;k}+\mathrm{S}_{L^{-j_\bullet}}(\cdot)\bigg)\bigg]\bigg\Vert_{\mathfrak{h}_{*}/2;\langle\phi\rangle}\nonumber\\
    &\leqslant \mathbb{E}\Big[\Vert{\mathcal{Z}}^{\mathrm{red}}_{j_\bullet}-\dot{\mathcal{Z}}^{\mathrm{red}}_{j_\bullet}\Vert_{\mathfrak{h}_*/2;\big\langle\sum_{k=j_\bullet}^0\mathrm{S}_{L^{k-j_\bullet}}\Psi_{\hspace{-1pt}\varepsilon;k}+\mathrm{S}_{L^{-j_\bullet}}\phi\big\rangle}\Big]+\nonumber\\
    &\hspace{20pt} \mathbb{E}\bigg[\sup_{t\in[0,1]}\bigg\Vert\bigg\langle \mathrm{D}\vert_{\sum_{k=j_\bullet}^0\mathrm{S}_{L^{k-j_\bullet}}((1-t)\Psi_{\hspace{-1pt}\varepsilon;k}+t\Psi_{\hspace{-1pt}\dot\varepsilon;k})+\mathrm{S}_{L^{-j_\bullet}}(\cdot)}\dot{\mathcal{Z}}_{j_\bullet}^{\mathrm{red}},\sum_{k=j_\bullet}^0\mathrm{S}_{L^{k-j_\bullet}}(\Psi_{\hspace{-1pt}\dot\varepsilon;k}-\Psi_{\hspace{-1pt}\varepsilon;k})\bigg\rangle\bigg\Vert_{\mathfrak{h}_*/2;\langle\phi\rangle}\bigg]\nonumber\\
    &\lesssim_1 \mathcal{O}_1^{\vert\mathbb{T}^2_{j_\bullet}\vert}\big(\Vert\zeta-\dot\zeta\Vert_\wp + \vert\varepsilon-\dot\varepsilon\vert\big)\mathbb{E}\bigg[G_{\kappa_*}\!\bigg(\mathbb{T}^2_{j_\bullet},\sum_{k=j_\bullet}^0\mathrm{S}_{L^{k-j_\bullet}}\Psi_{\hspace{-1pt}\varepsilon;k}+\mathrm{S}_{L^{-j_\bullet}}\phi\bigg)\bigg] +\nonumber\\
    & \hspace{20pt}\mathcal{O}_1^{\vert\mathbb{T}^2_{j_\bullet}\vert}\sum_{k=j_\bullet}^0\mathbb{E}\big[\Vert\Psi_{\hspace{-1pt}\dot\varepsilon;k}-\Psi_{\hspace{-1pt}\varepsilon;k}\Vert_{\mathcal{C}^2(\mathbb{T}^2_{k})}^2\big]^{ 1/2}\cdot\nonumber\\
    &\hspace{25pt}\mathbb{E}\bigg[G_{\kappa_*/8}\!\bigg(\mathbb{T}^2_{j_\bullet},\sum_{k=j_\bullet}^0\mathrm{S}_{L^{k-j_\bullet}}\Psi_{\hspace{-1pt}\varepsilon;k}+\mathrm{S}_{L^{-j_\bullet}}\phi\bigg)\bigg]^{\!1/4}\mathbb{E}\bigg[G_{\kappa_*/8}\!\bigg(\mathbb{T}^2_{j_\bullet},\sum_{k=j_\bullet}^0\mathrm{S}_{L^{k-j_\bullet}}\Psi_{\hspace{-1pt}\dot\varepsilon;k}+\mathrm{S}_{L^{-j_\bullet}}\phi\bigg)\bigg]^{\!1/4}\nonumber\\
    &\leqslant \mathcal{O}_1^{\vert\mathbb{T}^2_{j_\bullet}\vert}G_{\kappa_*/8}(\mathbb{T}^2,\phi)\big(\Vert\zeta-\dot\zeta\Vert_\wp + \vert\varepsilon-\dot\varepsilon\vert\big).
\end{align}
Finally, if we allow $\phi$ to vary,
\begin{align}
    \bigg\vert\mathbb{E}\bigg[{\mathcal{Z}}^{\mathrm{red}}_{j_\bullet}&\bigg(\sum_{k=j_\bullet}^0\mathrm{S}_{L^{k-j_\bullet}}\Psi_{\hspace{-1pt}\varepsilon;k}+\mathrm{S}_{L^{-j_\bullet}}\mathrm{C}_{\varepsilon}J\bigg)\bigg]-\mathbb{E}\bigg[{\mathcal{Z}}^{\mathrm{red}}_{j_\bullet}\bigg(\sum_{k=j_\bullet}^0\mathrm{S}_{L^{k-j_\bullet}}\Psi_{\hspace{-1pt}\varepsilon;k}+\mathrm{S}_{L^{-j_\bullet}}\mathrm{C}_{\dot\varepsilon}J\bigg)\bigg]\bigg\vert\nonumber\\
    &\leqslant \int_0^1\nonumber\bigg\vert\bigg\langle\mathrm{D}\vert_{((1-t)\mathrm{C}_{\varepsilon}+\mathrm{C}_{\dot\varepsilon}t)J}\mathbb{E}\bigg[{\mathcal{Z}}^{\mathrm{red}}_{j_\bullet}\bigg(\sum_{k=j_\bullet}^0\mathrm{S}_{L^{k-j_\bullet}}\Psi_{\hspace{-1pt}\varepsilon;k}+\mathrm{S}_{L^{-j_\bullet}}(\cdot)\bigg)\bigg],(\mathrm{C}_{\dot\varepsilon}-\mathrm{C}_{\varepsilon})J\bigg\rangle\bigg\vert\dd t\nonumber\\
    &\leqslant \mathcal{O}_1^{\vert\mathbb{T}^2_{j_\bullet}\vert}\vert\varepsilon-\dot\varepsilon\vert e^{\mathcal{O}_1\Vert J\Vert^2_*}.
\end{align}
Suppose now that
\begin{equation}
    m_\bullet \coloneqq \inf_{\varepsilon\leqslant\varepsilon_\bullet,\zeta\in\mathrm{B}_\bullet} \mathbb{E}\bigg[{\mathcal{Z}}^{\mathrm{red}}_{j_\bullet}\bigg(\sum_{k=j_\bullet}^0\mathrm{S}_{L^{k-j_\bullet}}\Psi_{\hspace{-1pt}\varepsilon;k}\bigg)\bigg] > 0.
\end{equation}
This is the object of the next section. It then immediately follows that
\begin{equation}
    \Vert\mathrm{Z}[\varepsilon,\zeta]-\mathrm{Z}[\dot\varepsilon,\dot\zeta]\Vert_\mathfrak{C} \leqslant \mathcal{O}_1^{\vert\mathbb{T}^2_{j_\bullet}\vert}m_\bullet^{-2}\big(\Vert\zeta-\dot\zeta\Vert_\wp + \vert\varepsilon-\dot\varepsilon\vert\big),
\end{equation}
which is exactly the Lipschitz property.

Let's also consider here the small coupling case, i.e. $\zeta_\bullet=\zeta_*$, i.e. $j_\bullet = 0$. Then
\begin{equation}
    \mathcal{Z}_0^{\mathrm{red}}(\phi) = e^{V_0(\mathbb{T}^2,\phi)} + K_0(\mathbb{T}^2,\phi).
\end{equation}
Doing a standard computation gives
\begin{equation}
    \mathrm{Z}[0,\zeta](J) = e^{\frac{1}{2}\langle \mathrm{C}J,J\rangle +V(\mathrm{C} J)} \frac{1+ K^V_1\!(\mathrm{C} J)}{1+ K^V_1\!(0)}
\end{equation}
where $V(\phi) = \int \zeta \big(1-\!\cos\!\big(\sqrtbeta\phi\big)\big)$ and
\begin{equation}
    \Vert K_1^V\Vert_{\mathfrak{h}_*,\kappa_*/2} \lesssim_1 \Vert \zeta\Vert_\wp^{2-\delta}.
\end{equation}
To show non-Gaussianity for small constant coupling, we can consider the connected truncated four-point function. This equals
\begin{align}
    \mathfrak{S}^{\mathrm{c},\mathrm{t}}_4(f_1,f_2,f_3,f_4) &= \frac{\partial^4}{\partial t_1\partial t_2\partial t_3\partial t_4}\bigg\vert_{t_1 =t_2=t_3=t_4= 0}\big(V + \log(1+K^V_1)\big)\bigg(\sum_{i=1}^4 t_if_i\bigg)\nonumber\\
    &=-\beta^2\zeta\int_{\mathbb{T}^2}(f_1f_2f_3f_4)(x)\dd x + 
   \sum_{\pi\in\mathrm{part}(4)}\frac{(-1)^{\vert\pi\vert-1}}{(1+K^V_1\!(0))^{\vert\pi\vert}}\prod_{I\in\pi}\bigg\langle \mathrm{D}^{\vert I\vert}\vert_0K^V_1,\bigotimes_{i\in I}f_i\bigg\rangle
\end{align}
where $\mathrm{part}(4)$ denotes the partitions of $\lcurl 1,2,3,4\rcurl$ into (nonempty) subsets. The term involving $K^V_1$ is bounded by $\mathcal{O}_1\zeta^{2-\delta} \prod_{i=1}^4\Vert f_i\Vert_{\mathcal{C}^2(\mathbb{T}^2)}$, so if we just take $f_1 = f_2=f_3=f_4 = 1$,
\begin{equation}
    \vert\mathfrak{S}^{\mathrm{c},\mathrm{t}}_4(1,1,1,1)\vert \geqslant \beta^2\vert\zeta\vert - \mathcal{O}_1\vert\zeta\vert^{2-\delta} > 0.
\end{equation}

\section{The lover bound}

Showing $m_\bullet > 0$ requires some special theory, which we have sequestered in this section. Suppose we manage to find a neighborhood $U_\bullet$ of $0$ in $\mathcal{C}^\infty(\mathbb{T}_{j_\bullet}^2)$ such that
\begin{equation}
    {\mathcal{Z}}^{\mathrm{red}}_{j_\bullet}\vert_{U_\bullet} \geqslant z_\bullet > 0.
\end{equation}
It will then follow that
\begin{equation}
    \mathbb{E}\bigg[{\mathcal{Z}}^{\mathrm{red}}_{j_\bullet}\bigg(\sum_{k=j_\bullet}^0\mathrm{S}_{L^{k-j_\bullet}}\Psi_{\hspace{-1pt}\varepsilon;k}\bigg)\bigg]\geqslant \mathbb{E}\bigg[\big(\mathbb{I}_{U_\bullet} {\mathcal{Z}}^{\mathrm{red}}_{j_\bullet}\big)\hspace{-1pt}\bigg(\sum_{k=j_\bullet}^0\mathrm{S}_{L^{k-j_\bullet}}\Psi_{\hspace{-1pt}\varepsilon;k}\bigg)\bigg] \geqslant z_\bullet\,\mathbb{P}\bigg[\bigg\lcurl\sum_{k=j_\bullet}^0\mathrm{S}_{L^{k-j_\bullet}}\Psi_{\hspace{-1pt}\varepsilon;k} \in U_\bullet\bigg\rcurl\bigg].
\end{equation}
The probability appearing here is nonzero because a Gaussian measure assigns nonzero weight to every neighborhood of the origin, so $m_\bullet > 0$. Hence we only need to find $U_\bullet$, which will of course be some kind of norm ball.

The idea is to apply polymer systems theory (cf. e.g. \cite{Glimm1987}, Chapter 20) to exponentiate the polymer representation. That is, if we define
\begin{equation}
    K_{j}^{V}\!(X,\phi) = e^{-V_j(X,\phi)}K_j(X,\phi),
\end{equation}
then
\begin{equation}
    \mathcal{Z}_j^{\mathrm{red}}(\phi) = \exp\!\bigg(V_j(\mathbb{T}^2_j,\phi) + \sum_{n=1}^\infty\frac{1}{n!}\sum_{(X_1,\ldots,X_n)\in\mathbf{C}_{n;j}}\mathfrak{n}(X_1,\ldots,X_n)K^{V}_j\!(X_1,\phi)\ldots K^{V}_j\!(X_n,\phi)\bigg).
\end{equation}
$\mathbf{C}_{n;j}$ denotes the class of $n$-tuples $(X_1,\ldots,X_n)\in\mathcal{P}_{\mathrm{c};j}$ for which $\bigcup_{i=1}^nX_i$ is connected, i.e. the graph on $\lcurl 1,\ldots,n\rcurl$ induced by the intersection relation of $X_1,\ldots,X_n$ is connected. $\mathfrak{n}(X_1,\ldots,X_n)$ is a combinatorial factor; all we need is that it is bounded by the number of spanning trees on the connectedness graph of $X_1,\ldots,X_n$. 

The equality holds as long as the right-hand side converges. We cannot show this for all $\phi$ because $K_j^V$ could become large, so we define $U_\bullet$ by the condition $G_{\kappa_*}\!(\mathbb{T}^2_{j_\bullet},\phi) \leqslant 2$. Then
\begin{align}
    \bigg\vert\sum_{n=1}^\infty\frac{1}{n!}\sum_{(X_1,\ldots,X_n)\in\mathbf{C}_{n;j_\bullet}}&\mathfrak{n}(X_1,\ldots,X_n)K^{V}_{j_\bullet}\!(X_1,\phi)\ldots K^{V}_{j_\bullet}\!(X_n,\phi)\bigg\vert\nonumber\\
    &\leqslant \sum_{n=1}^\infty\frac{\big(2\Vert\zeta_{j_\bullet}\Vert_\wp^{2-\delta}\big)^{\! n}}{n!}\sum_{(X_1,\ldots,X_n)\in\mathbf{C}_{n;j_\bullet},\tau\in\mathcal{T}_n}^{*} A^{-\vert X_1\vert}\ldots A^{-\vert X_n\vert}.
\end{align}
$\mathcal{T}_n$ denotes the trees on $\lcurl 1,\ldots, n\rcurl$ and the $*$ denotes the condition that $\tau$ be a subtree of the connectedness graph of $X_1,\ldots,X_n$. To control the sum, we do a standard tree pruning argument (cf. e.g. \cite{Dimock2000}). We pick an index $i_0$ and block $B_0\subset X_{i_0}$ and only keep the intersection condition for the polymers along the bonds of $\tau$:
\begin{align}
    \sum_{n=1}^\infty&\frac{\big(2\Vert\zeta_{j_\bullet}\Vert_\wp^{2-\delta}\big)^{\! n}}{n!}\sum_{X_1,\ldots,X_n\in\mathbf{C}_{n;j_\bullet},\tau\in\mathcal{T}_n}^{*} A^{-\vert X_1\vert}\ldots A^{-\vert X_n\vert}\nonumber\\
    &\leqslant 2\zeta_*^{2-\delta}\vert\mathbb{T}_{j_\bullet}^2\vert+\sum_{n=2}^\infty\frac{\big(2\zeta_*^{2-\delta}\big)^{\! n}}{n\cdot n!}\sum_{i_0=1}^n\frac{1}{\vert X_{i_0}\vert}\sum_{X_1,\ldots,X_n\in\mathcal{P}_{\mathrm{c};j_\bullet},\tau\in\mathcal{T}_n,B_0\in\mathcal{B}_{j_\bullet}(X_{i_0})} \prod_{i=1}^nA^{-\vert X_i\vert}\prod_{\lcurl i,j\rcurl\in\tau}\mathbb{I}_{X_i\cap X_j\neq\varnothing}.
\end{align}
We start pruning the tree at the leaves, working inwards towards $i_0$. For instance, if $i$ is a vertex with leaves $j_1,\ldots,j_{d_i-1}$ ($d_i$ being the degrees of $\tau$), then we first perform the sum ($B$ here is some arbitrary block in $X_i$ just serving to pin $X_{j_1}$)
\begin{equation}
    \sum_{X_{j_1}\in\mathcal{P}_{\mathrm{c};j_\bullet}}A^{-\vert X_{j_1}\vert}\mathbb{I}_{X_i\cap X_{j_1}\neq\varnothing} \leqslant \vert X_i\vert \sum_{X_{j_1}\in \mathcal{P}_{\mathrm{c};j_\bullet}}^{B\cap X_{j_1}\neq\varnothing}A^{-\vert X_{j_1}\vert}\leqslant 9\vert X_i\vert\sum_{k=1}^\infty (\mathcal{O}_1A^{-1})^k\leqslant\vert X_i\vert.
\end{equation}
After doing this for all the leaves, $i$ now becomes a leaf itself and a factor $\vert X_i\vert^{d_i-1}\leqslant (d_i-1) ! \,e^{\vert X_i\vert}$ is produced. The exponential is merged with $A$ to give $(e^{-1}A)^{-\vert X_i\vert}$. These factors of $e^{-1}$ do not accumulate over time, since they are always merged to a different polymer. Anyway, proceeding like this (the last step when pruning the leaves of $i_0$ is slightly different but basically the same) leaves us with
\begin{align}
    \sum_{n=1}^\infty\frac{\big(2\Vert\zeta_{j_\bullet}\Vert_\wp^{2-\delta}\big)^{\! n}}{n!}\sum_{X_1,\ldots,X_n\in\mathbf{C}_{n;j_\bullet},\tau\in\mathcal{T}_n}^{*}& A^{-\vert X_1\vert}\ldots A^{-\vert X_n\vert}\nonumber\\
    &\leqslant 2\zeta_*^{2-\delta}\vert\mathbb{T}^2_{j_\bullet}\vert\bigg(1+\sum_{n=2}^\infty\frac{\big(2\zeta_*^{2-\delta}\big)^{\!n-1}}{ n!}\sum_{\tau\in\mathcal{T}_n}\prod_{i=1}^n (d_i-1)!\bigg).
\end{align}
Using Cayley's formula for counting trees with degree sequence fixed and counting the number of possible degree sequences gives
\begin{equation}
    \sum_{\tau\in\mathcal{T}_n}\prod_{i=1}^n (d_i-1)!\leqslant 4^{n-1}(n-2)!
\end{equation}
and so we conclude that
\begin{equation}
    z_\bullet \geqslant e^{-2\zeta_*\vert\mathbb{T}^2_{j_\bullet}\vert}.
\end{equation}
\appendix

\section{Scale decomposition}\label{AppScale}

\subsection{Mollifier}\label{AppMolly}

We take $\chi = c_{\mathfrak{d}}\mathsf{B}^{*2}\mathrm{C}^{*\mathfrak{d}}$ where $\mathsf{B}$ is a standard bump of support radius $1/4$ on $\mathbb{R}^2$ and $c_{\mathfrak{d}}$ is chosen so that $\chi$ has integral $1$. In Fourier space,
\begin{equation}
    \hat{\chi}(p) \propto_1 \int_{\mathbb{R}^2}\frac{\hat{\mathsf{B}}(p-q)^2}{(1+\vert q\vert^2)^{\mathfrak{d}}}\mathrm{d}q.
\end{equation}
Because $\hat{\mathsf{B}}$ is Schwartz, it is easy to see that this still behaves asymptotically like $(1+\vert p\vert^2)^{\mathfrak{d}}$. For $0<\varepsilon<1$, put $\chi_\varepsilon(x) = \varepsilon^{-2}\chi(\varepsilon^{-1}x)$ (this also makes sense as a function on $\mathbb{T}^2$).

\subsection{Continuous decomposition}

Following \cite{BauerschmidtFiniteRange}, we may realize the massive GFF on $\mathbb{T}^2$ as an integral against an $\mathrm{L}^2$-cylindrical Brownian motion $B_t$:
\begin{equation}
    {\Phi} = \int_0^\infty t^{-1/2}{\mathrm{F}}_{\hspace{-1pt}t}(\cdot\,;\hspace{-.5pt}1)*\dd {B}_t.
\end{equation}
Here ${\mathrm{F}}_{\hspace{-1pt}t}(x;\hspace{-.5pt}1) = 0$ when $\vert x\vert_2 \geqslant t/2$, and its precise form is
\begin{equation}
    {\mathrm{F}}_{\hspace{-1pt}t}(x;\hspace{-.5pt}m) = t\sum_{p\in 2\pi\mathbb{Z}^2}\varphi\big(t\sqrt{m^2+p^2}\big)e^{ipx}
\end{equation}
where $\varphi$ is a real even $1/2$-Gevrey function\footnote{We don't need this much regularity, but we take it since it comes for free.} with $\mathrm{supp}(\hat{\varphi})\subset[-1/2,1/2]$ normalized so that
\begin{equation}
    \int_0^\infty t\,\varphi(t)^2\dd t = 1.
\end{equation}
After mollifying and periodizing, we can decompose
\begin{align}
\Phi_{\hspace{-.5pt}\varepsilon} &= \int t^{-1/2}{\mathrm{F}}_{\hspace{-1pt}\varepsilon;t}(\cdot\,;\hspace{-.5pt}m)*\dd {B}_t, \nonumber\\
{\mathrm{F}}_{\hspace{-1pt}\varepsilon;t}(x;\hspace{-.5pt}m)&= (\chi_\varepsilon*{\mathrm{F}}_{\hspace{-1pt}t})(x)=t\hspace{1pt}\sum_{p\in 2\pi\mathbb{Z}^2}\hat\chi(\varepsilon p)\varphi\big(t\sqrt{m^2+p^2}\big)e^{ipx}.
\end{align}
The covariance decomposes as
\begin{equation}
    \mathrm{C}_\varepsilon(x) = (\chi_\varepsilon*\mathrm{C}*\chi_\varepsilon)(x) = \int_0^\infty {\mathrm{F}}^{*2}_{\hspace{-1pt}\varepsilon;t}(x;\hspace{-.5pt}m)\frac{\dd t}{t}.
\end{equation}
The following estimate is immediate for any multiindex $\alpha\leqslant4$ and $\varepsilon\geqslant 0$:
\begin{equation}
\vert\partial^\alpha_x{\mathrm{F}}^{* 2}_{\hspace{-1pt}\varepsilon;t}(x;\hspace{-.5pt}m)\vert\lesssim_{1}t^{-\vert\alpha\vert}e^{-\sqrt{t m}}.
\end{equation}
When $t < \varepsilon$, we expect better estimates since we have mollified. Indeed, in this case
\begin{equation}
\vert\partial^\alpha_x{\mathrm{F}}^{*2}_{\hspace{-1pt}\varepsilon;t}(x;\hspace{-.5pt}m)\vert\lesssim_{1}(\varepsilon^{-1}t)^2\varepsilon^{-\vert\alpha\vert}.
\end{equation}
Finally, we wish to know the value at $x = 0$ more accurately since it enters directly into the RG flow. The sum over $2\pi\mathbb{Z}^2$ is a bit awkward, so we first work on all of $\mathbb{R}^2$: if $\bar{\mathrm{F}}_{\hspace{-1pt}\varepsilon;t}$ is the analogue of ${\mathrm{F}}_{\hspace{-1pt}\varepsilon;t}$ on the full plane, straightforward computations with the Fourier transform show
\begin{equation}
\bar{\mathrm{F}}^{*2}_{\hspace{-1pt}\varepsilon;t}(0;\hspace{-.5pt}m) = \frac{1}{2\pi} + r^{(1)}_{\varepsilon;t}(m);\quad\quad \vert r^{(1)}_{\varepsilon;t}(m)\vert \lesssim_1\varepsilon t^{-1}+t^2m^2.
\end{equation}
Because of the Poisson summation formula
\begin{equation}
    \mathrm{F}^{*2}_{\hspace{-1pt}\varepsilon;t}(x;\hspace{-.5pt}m) = \sum_{y\in\mathbb{Z}^2}\bar{\mathrm{F}}^{*2}_{\hspace{-1pt}\varepsilon;t}(x-y;\hspace{-.5pt}m)
\end{equation}
and the finite range property, $\mathrm{F}^{*2}_{\hspace{-1pt}\varepsilon;t}(x;\hspace{-.5pt}m) = \bar{\mathrm{F}}^{*2}_{\hspace{-1pt}\varepsilon;t}(x;\hspace{-.5pt}m)$ if $t < 1/4$, so for these values the above result transfers immediately.

\subsection{Discrete decomposition \texorpdfstring{\&}{and} scaling}

Write $\varepsilon = \epsilon L^{-\mathsf{N}-1}/6$ for $\epsilon\in[1,L)$. Then put ($j = -\mathsf{N},\ldots,-1$)
\begin{equation}
    \mathrm{C}_{\varepsilon;-\mathsf{N-1}}(x) = \int_0^{L^{-\mathsf{N}}/6}\mathrm{F}_{\hspace{-1pt}\varepsilon;t}(x;\hspace{-.5pt}1)\frac{\dd t}{t},\quad
    \mathrm{C}_{\varepsilon;j}(x) = \int_{L^j/6}^{L^{j+1}/6}\mathrm{F}_{\hspace{-1pt}\varepsilon;t}(x;\hspace{-.5pt}1)\frac{\dd t}{t},\quad
    \mathrm{C}_{\varepsilon;0}(x) = \int_{1/6}^{\infty}\mathrm{F}_{\hspace{-1pt}\varepsilon;t}(x;\hspace{-.5pt}1)\frac{\dd t}{t}
\end{equation}
with corresponding fields $\Phi_{\hspace{-.5pt}\varepsilon;j}$. The covariance $\mathrm{C}_{\varepsilon;j}$ for $j\leqslant -1$ has range at most $L^j/2$.

We will also rescale these covariances so that we are always working \enquote{at scale $1$}. Let $\mathrm{S}_\lambda$ be the scaling map $(\mathrm{S}_\lambda\phi)(x) = \phi(\lambda^{-1}x)$. Then
\begin{equation}
    \Gamma_{\hspace{-1pt}\varepsilon;j}=\mathrm{S}_{L^{-j}}\mathrm{C}_{\varepsilon;j}\mathrm{S}_{L^{-j}}^*
\end{equation}
with the adjoint taken in $\mathrm{L}^2$. The corresponding fields are denoted $\Psi_{\hspace{-1pt}\varepsilon;j} = \mathrm{S}_{L^{-j}}\Phi_{\hspace{-.5pt}\varepsilon;j}$. It is easy to check that
\begin{equation}
    \Gamma_{\hspace{-1pt}\varepsilon;j}(x) = \mathrm{C}_{\hspace{-1pt}\varepsilon;j}(L^jx).
\end{equation}
If $j\leqslant -1$,
\begin{equation}
    \vert\partial^\alpha\Gamma_{\hspace{-1pt}\varepsilon;j}(x)\vert\lesssim_{1}\!\bigg\lcurl\hspace{-5pt}
    \begin{array}{ll}
        \log(L), & \vert\alpha\vert = 0\\
        1, & \vert\alpha\vert\geqslant 0
    \end{array}\,\,.
\end{equation}
For the diagonal values $x = 0$, we find that ($j\neq-\mathsf{N}-1$)
\begin{equation}
    \Gamma_{\hspace{-1pt}\varepsilon;j}(0) = \frac{\log(L)}{2\pi} + r^{(2)}_{\varepsilon;j};\quad\quad \vert r^{(2)}_{\varepsilon;j}\vert\lesssim_1 1.
\end{equation}
This also implies that with $\zeta_{\varepsilon;j}$ as in (\ref{zetajdef}),
\begin{equation}
    \zeta_{\varepsilon;j+1} = \varrho^{(1)}_{\varepsilon;j}L^{2\sigma}\mathrm{S}_{L^{-1}}\zeta_{\varepsilon;j}
\end{equation}
where $\varrho^{(1)}_{\varepsilon;j}$ is bounded uniformly away from $0,\infty$ if $j\neq-\mathsf{N}-1$.
For the $j = 0$ case, similar reasoning gives
\begin{equation}
     \vert\partial^\alpha\Gamma_{\hspace{-1pt}\varepsilon;0}(x)\vert\lesssim_{1} 1.
\end{equation}

Finally, we will also need to compare the fields at the same scales but with different covariances. Take $\dot{\varepsilon}\leqslant\varepsilon$ and restrict $j\neq-\mathsf{N}$. We have
\begin{align}
    \vert\partial^\alpha(\Gamma_{\hspace{-1pt}\varepsilon;j}-\Gamma_{\hspace{-1pt}\dot{\varepsilon};j})\vert(x) &\lesssim_{1} L^{-j}\vert\varepsilon-\dot{\varepsilon}\vert,\label{CovChangeBoundGamma}\\
    \mathbb{E}\big[\Vert\Psi_{\hspace{-1pt}\varepsilon;j}-\Psi_{\hspace{-1pt}\dot\varepsilon;j}\Vert^2_{\mathcal{C}^r(X)}\big]&\lesssim_r L^{-2j}\vert\varepsilon-\dot{\varepsilon}\vert^2\vert X\vert,\label{CovChangeBound}\\
    \Vert\zeta_{\varepsilon;j}-\zeta_{\dot\varepsilon;j}\Vert_{\wp}&\lesssim_1 L^{-j}\vert\varepsilon-\dot{\varepsilon}\vert\Vert\zeta_{0;j}\Vert_{\wp},\label{couplingControl}\\
    \Vert\zeta_{\varepsilon;j}-\dot\zeta_{\varepsilon;j}\Vert_{\wp}&=\varrho^{(2)}_{\varepsilon;j}\Vert\zeta_{0;j}-\dot\zeta_{0;j}\Vert_{\wp}.\label{couplingControl2}
\end{align}
Here $X\in\mathcal{P}_j$ and $\varrho^{(2)}_{\varepsilon;j}$ is bounded uniformly away from $0,\infty$. The last equality also holds for $j = -\mathsf{N}-1$, but then $\varrho^{(2)}_{\varepsilon;-\mathsf{N}-1}$ is bounded uniformly away from $0,\infty$ modulo a dependence on $L$. Showing all these relations is straightforward, using only Gaussian random variable theory, manipulations with the mollifier $\chi_{\varepsilon}$ in Fourier space, and the already stated properties of $\Gamma_{\hspace{-1pt}\varepsilon;j}$.

\subsection{Finer scale decomposition}\label{AppFinerScale}

We will need a refined scale decomposition to prove some estimates below. Assume $L=\ell^\mathfrak{m}$ with $\mathfrak{m}\geqslant\ell^2$. As above, we have covariances (defined for $j = -\mathsf{N},\ldots,-1$, $s = -\mathfrak{m},\ldots,-1$):
\begin{equation}
    \mathrm{C}_{\varepsilon;j;s}(x) = \int_{L^{j+1}\ell^s/6}^{L^{j+1}\ell^{s+1}/6}\mathrm{F}_{\hspace{-1pt}\varepsilon;t}(x;\hspace{-.5pt}1)\frac{\dd t}{t}
\end{equation}
with attendant rescaled versions $\Gamma_{\hspace{-1pt}\varepsilon;j;s}(x) = \mathrm{C}_{\varepsilon;j;s}(L^{j+1}\ell^s x)$ and fields $\Psi_{\hspace{-1pt}\varepsilon;j;s}$. These again satisfy
\begin{align}
    \vert\partial^\alpha\Gamma_{\hspace{-1pt}\varepsilon;j;s}(x)\vert&\lesssim_{1}\!\bigg\lcurl\hspace{-5pt}
    \begin{array}{ll}
        \log(L), & \vert\alpha\vert = 0\\
        1, & \vert\alpha\vert\geqslant 0
    \end{array}\,\,,\\
    \vert\partial^\alpha(\Gamma_{\hspace{-1pt}\varepsilon;j;s}-\Gamma_{\hspace{-1pt}\dot\varepsilon;j;s})\vert(x)&\lesssim_1 L^{-(j+1)}\ell^{-s}\vert\varepsilon-\dot\varepsilon\vert,\\
    \mathbb{E}\big[\Vert\Psi_{\hspace{-1pt}\varepsilon;j}-\Psi_{\hspace{-1pt}\dot\varepsilon;j}\Vert^2_{\mathcal{C}^r(X)}\big]&\lesssim_r  L^{-2(j+1)}\ell^{-2s}\vert\varepsilon-\dot\varepsilon\vert^2\vert X\vert.
\end{align}

\section{Polymers} \label{AppPoly}

\subsection{Basic definitions} \label{AppPolyBasic}

Let $\mathcal{B}_{j;s}$ denote all closed squares of the tesselation $\mathbb{Z}^2/L^{-j-1}\ell^{-s}\mathbb{Z}$ of $\mathbb{T}^2_{j;s}=L^{-j-1}\ell^{-s}\mathbb{T}^2$. $\mathcal{P}_{j;s}$ is the set of polymers which are finite unions (including $\varnothing$) of elements of $\mathcal{B}_{j;s}$. For $X\in\mathcal{P}_{j;s}$, we write $\vert X\vert$ for the area of $X$. $\mathcal{P}_{\mathrm{c};j;s}\subset\mathcal{P}_{j;s}$ denotes nonempty polymers which are connected as sets. $\mathcal{S}_{j;s}\subset\mathcal{P}_{\mathrm{c};j;s}$ denotes \enquote{small} polymers which are defined as being connected with $\vert X\vert\leqslant 4$. Collars of $X$, denoted $X^{[r]}$, are the sets of points which are at an $\infty$-distance of at most $r$ from $X$. $X^* = X^{[3]}$ is especially important as it equals the unions of all the small polymers intersecting $\mathrm{Int}(X)$. $\bar{X}^{(k)}$ denotes the smallest member of $\mathcal{P}_{j;s+k}$ (of course $s$ will \enquote{carry into} $j$ if $s+k \geqslant 0$) containing $\ell^{-k}X$. We use the special notation $\bar{X}^{L} = \bar{X}^{(\mathfrak{m})}$. $\mathcal{C}(X)$ denotes the connected components of $X$. Finally, we will most often deal with polymers in $\mathbb{T}^2_{j;-\mathfrak{m}}$, so in this case we drop the index $-\mathfrak{m}$.

\subsection{Geometric estimates} \label{AppGeometricEst}

A general insight of the Wilsonian RG is that the effective potential stays roughly local. That is, we begin with a potential which is given by a local expression in the field and the RG flow is supposed to approximately preserve this property; this is why we have introduced the distinction between small and large polymers. Note, for instance, that if $X$ is not small, then $\vert\bar{X}^{(1)}\vert<\vert X\vert$, at least if $\ell$ is large enough. In fact, this inequality always gains us a factor uniform in $\vert X\vert$, which enables us to completely control the behavior on large sets and which is the content of the next lemma \cite{ParkCity}.

\begin{lem}
    If $X\in\mathcal{P}_{\mathrm{c};j;s}$ is not small and $k\geqslant1$, then with $\eta = \frac{1}{101}$,
    \begin{equation}
        \vert \bar{X}^{(k)}\vert \leqslant \frac{1}{1+\eta}\vert X\vert.\label{largeShrinkage}
    \end{equation}
\end{lem}

\begin{proof}
    It's enough to consider $k=1$. For clarity, we put $2 = d$ (the proof works in all dimensions with $\eta$ depending on $d$, but we won't bother thinking through the geometry there). We will prove (\ref{largeShrinkage}) inductively with $\eta = (2^d+1+2^d\Delta)^{-1}$ where $\Delta = 6^d-4^d+3^d-2^d-1$ is the largest number of blocks that can touch a connected polymer of $2^d+1$ blocks.

    The induction will be on $\vert\bar{X}^{(1)}\vert$, starting with $\vert \bar{X}^{(1)}\vert = 2^d+1$. Because $\ell$ is large, the only way we can have $\vert X\vert = \vert\bar{X}^{(1)}\vert$ is if $X$ is either disconnected or small, so we are forced to conclude that in our case, $\vert X\vert \geqslant 2^d+2$. Hence
    \begin{equation}
        \bigg(1+\frac{1}{2^d+1}\bigg)\vert\bar{X}^{(1)}\vert\leqslant\vert X\vert,
    \end{equation}
    which witnesses the inductive base case. Now suppose $\vert\bar{X}^{(1)}\vert = n+1$. Remove from $\bar{X}$ a connected subset $Y$ comprising $2^d+1$ blocks. Then $\bar{X}-Y$ breaks up into at most $\Delta$ connected components (since every connected component certainly intersects the $1$-block thickening of $\bar{X}$), denoted $Z_1,\ldots,Z_\Delta$ and ordered by decreasing size (we pad with $\varnothing$ if necessary). Let $k$ be the index such that $Z_i\geqslant  2^d+1$ when $i\leqslant k$ and $Z_i\leqslant 2^d$ for $i > k$. Let $M = \sum_{i=1}^k\vert Z_i\vert$ and $m = \sum_{i=k+1}^\Delta\vert Z_i\vert$. Then
    \begin{equation}
        \frac{\vert X\vert}{\vert \bar{X}^{(1)}\vert}\geqslant\frac{2^d+2+(1+\eta)M+m}{2^d+1+M+m}\geqslant 1+\frac{1+\eta M}{2^d+1+M+2^d\Delta} = 1+\eta.
    \end{equation}
\end{proof}

Due to the way we set up the RG, disconnected polymers sometimes appear in the formulas. To control these, we also have the following proposition \cite{Falco1}.

\begin{lem}
    Let $X$ be any polymer. Then
    \begin{equation}
        \vert\bar{X}^{(k)}\vert\leqslant\max\!\bigg\lcurl\frac{1}{1+\eta/2}\vert X\vert, 8\vert\mathcal{C}(X)\vert\bigg\rcurl.\label{disconnectedGrowth}
    \end{equation}
\end{lem}

\begin{proof}
Break $X$ up into $Y$ the small and $Z$ the large components. Since $\vert \bar{X}^{(j)}\vert \leqslant\vert\bar{Y}^{(j)}\vert + \vert \bar{Z}^{(j)}\vert$, we must have either $\frac{1}{2}\vert\bar{X}^{(j)}\vert \leqslant \vert\bar{Y}^{(j)}\vert$ or $\frac{1}{2}\vert\bar{X}^{(j)}\vert\leqslant\vert \bar{Z}^{(j)}\vert$. In the former case we simply get
\begin{equation}
    \vert \bar{X}^{(j)}\vert\leqslant 2\vert \bar{Y}^{(j)}\vert \leqslant 2\vert Y\vert \leqslant 8\vert\mathcal{C}(X)\vert,
\end{equation}
while in the latter we apply the previous lemma:
\begin{equation}
    \vert X\vert = \vert Y\vert + \vert Z\vert \geqslant \vert \bar{Y}^{(j)}\vert + (1+\eta)\vert \bar{Z}^{(j)}\vert\geqslant\vert \bar{X}^{(j)}\vert + \eta \vert \bar{Z}^{(j)}\vert\geqslant(1+\eta/2)\vert \bar{X}^{(j)}\vert.
\end{equation}
\end{proof}

\subsection{\texorpdfstring{$\mathcal{C}^r$}{Cr} spaces \texorpdfstring{\&}{and} regulators} \label{AppCrReg}

The basic function norms for $X\in\mathcal{P}_{j;s}$ are
\begin{align}
    \Vert\nabla^n\hspace{-1pt}f\Vert_{\mathrm{L}^p(X)} &= \big\Vert(\Vert\partial^\mu \!f\Vert_{\mathrm{L}^p(X)})_{\vert\mu\vert = n}\big\Vert_{\ell^p},\\
    \Vert\nabla^n\hspace{-1pt}f\Vert_{\mathrm{L}^p(\partial X)} &= \big\Vert(\Vert\partial^\mu\! f\Vert_{\mathrm{L}^p(\partial X)})_{\vert\mu\vert = n}\big\Vert_{\ell^p},\\
    \Vert\nabla^n\hspace{-1pt}f\Vert_{\mathcal{B}^q_p(X)} &= \big\Vert(\Vert\nabla^{n}\hspace{-1pt} f\Vert_{\mathrm{L}^p(B)})_{B\in\mathcal{B}_j(X)}\big\Vert_{\ell^q},\\
    \Vert\nabla^n\hspace{-1pt}f\Vert_{\mathcal{C}^r(X)} &= \max_{0\leqslant m\leqslant r}\Vert\nabla^{n+m}\hspace{-1pt}f\Vert_{\mathrm{L}^\infty(X)},\\
    \Vert f\Vert_{\wp} &= \Vert f\Vert_{\mathcal{B}^\wp_{\infty}(\mathbb{T}^2_j)}.
\end{align}
The last norm will only be used to control the coupling; note that when $j=0$ it reduces to the $\mathrm{L}^\wp$ norm. We can also see that $\Vert\!\cdot\!\Vert_{\mathcal{B}^p_p} = \Vert\!\cdot\!\Vert_{\mathrm{L}^p}$.

The notion of differentiable function on $X$ requires a bit of care since this is generically not a submanifold with corners of $\mathbb{T}^2$. However, it is the closure of a relatively compact open set and it is rather easy to convince oneself that it is a Whitney domain,\footnote{The only problematic points could be where two squares touch at only one vertex, but the angle at which they meet is $90^\circ\neq 0$, so the Whitney condition holds. This is also true of polymers in higher dimensions.} so the extension theorem applies (for instance see \cite{Whitney}, though this only treats the case of subsets in $\mathbb{R}^2$; the extension theorem for the torus can easily be obtained by periodization arguments). For instance, all reasonable definitions of $\mathcal{C}^r(X)$ coincide and have good properties: we may define this space as either those functions in $\mathcal{C}^r(\mathrm{Int}(X))$ whose $r$ derivatives all have continuous extensions to the boundary, or equivalently as those functions on $X$ which are restrictions of $\mathcal{C}^r$ functions on $\mathbb{T}^2$.

We will also need functionals of the field values. To control their growth in the field, we introduce regulators \cite{ParkCity,Falco1} depending on $\kappa > 0$ (for $\phi\in\mathcal{C}^2(X)$)
\begin{align}
    \kappa \log G_\kappa(X,\phi) &= \mathfrak{c}^{(1)}\big(\Vert\nabla\phi\Vert^2_{\mathrm{L}^2(X)} + \Vert\nabla^2\phi\Vert^2_{\mathcal{B}^{2}_\infty(X)}\big) + \mathfrak{c}^{(2)}\Vert\nabla\phi\Vert^2_{\mathrm{L}^2(\partial X)},\\
    \kappa \log g_\kappa(X,\phi) &= \mathfrak{c}^{(3)}\sum_{n=0}^2\Vert\nabla^n\phi\Vert^2_{\mathcal{B}^{2}_\infty(X)}.
\end{align}
$G_\kappa$ is what appears directly in norms, while $g_\kappa$ is an important intermediate quantity in computations.

The actual functionals are smooth functions $F$ on $\mathcal{C}^2(X)$, but for now we consider a general Banach space $\mathcal{X}$. Then if $\mathfrak{h} > 0$ and $x\in\mathcal{X}$, we put
\begin{align}
    \Vert F\Vert_{n;\langle x\rangle;\mathcal{X}} &= \sup_{y_1,\ldots,y_n\in\mathrm{Ball}(\mathcal{X})}\big\vert\langle\mathrm{D}^n\vert_x F,y_1,\ldots,y_n\rangle\big\vert,\\
    \Vert F\Vert_{\mathfrak{h};\langle x\rangle;\mathcal{X}} &= \sum_{n=0}^\infty\frac{\mathfrak{h}^n}{n!}\Vert F\Vert_{n;\langle x\rangle;\mathcal{X}}.
\end{align}
We also denote by $\mathcal{N}_\mathfrak{h}(\mathcal{X})$ the set of functionals for which $\Vert F\Vert_{\mathfrak{h};\langle x\rangle}<\infty$ for all $x$.

Finally, the norms satisfy some identities that we will use without further comment. If $T:\mathcal{X}\rightarrow\mathcal{Y}$ is bounded,
\begin{equation}
    \Vert F\circ T\Vert_{\mathfrak{h};\langle x\rangle;\mathcal{X}} \leqslant\Vert F\Vert_{\Vert T\Vert\mathfrak{h};\langle Tx\rangle;\mathcal Y}.
\end{equation}
For instance, $\mathrm{S}_{\lambda}:\mathcal{C}^r(X)\rightarrow\mathcal{C}^r(\lambda X)$ has norm $1$ if $\lambda \geqslant 1$. Furthermore, if $q_{1,2}:\mathcal{X}\rightarrow\mathcal{Y}_{1,2}$ are two quotients and $F_{1,2}\in\mathcal{N}_{\mathfrak{h}}(\mathcal{Y}_{1,2})$,
\begin{equation}
    \Vert (F_1\circ q_1)(F_2\circ q_2)\Vert_{\mathfrak{h};\langle\phi\rangle;\mathcal{X}}\leqslant \Vert F_1\Vert_{\Vert q_1\Vert\mathfrak{h};\langle q_1\phi\rangle;\mathcal{Y}_1}\Vert F_2\Vert_{\Vert q_2\Vert\mathfrak{h};\langle q_2\phi\rangle;\mathcal{Y}_2}.
\end{equation}
For instance, $\mathcal{C}^r(X\cup Y)$ quotients with norm $1$ to $\mathcal{C}^r(X)$ and $\mathcal{C}^r(Y)$, so we have a kind of product inequality for functionals on $\mathcal{C}^r$-spaces where the underlying sets can change.

\subsection{Polymer activities} \label{AppPolymerAct}

A polymer activity is a function $K$ on $\mathcal{P}_j-\lcurl\varnothing\rcurl$. The codomain can for instance be $\mathbb{R}$, but for our purposes it will be more interesting to consider cases where $K(X)$ is a smooth functional on a Banach space $\mathcal{X}_X$ depending on $X$. For instance, $\mathcal{X}_X$ will mostly equal $\mathcal{C}^2(X)$ in these notes. We make the following notational conventions: if $x\in\mathcal{X}_X$, $\mathfrak{h},A > 0$ and $w_X$ is a choice of positive function on each $\mathcal{X}_X$,
\begin{align}
    \Vert K\Vert_{n;\langle X,x\rangle;\mathcal{X}_X} &= \Vert K(X)\Vert_{n;\langle x\rangle;\mathcal{X}_X},\\
    \Vert K\Vert_{\mathfrak{h};\langle X,x\rangle;\mathcal{X}_X} &= \Vert K(X)\Vert_{\mathfrak{h};\langle x\rangle;\mathcal{X}_X},\\
    \Vert K\Vert_{\mathfrak{h},w;\langle X\rangle;\mathcal{X}_X} &= \sup_{x\in\mathcal{X}_X}w_X(x)^{-1}\Vert K\Vert_{\mathfrak{h};\langle X,x\rangle;\mathcal{X}_X},\label{hwnorm},\\
    \Vert K\Vert_{\mathfrak{h},w,A;\mathcal{X}_X} &= \sup_{X\in\mathcal{P}_j}A^{\vert X\vert}\Vert K\Vert_{\mathfrak{h},w;\langle X\rangle;\mathcal{X}_X}.\label{hwanorm}
\end{align}
The subscript $\mathcal{X}_X$ on the left-hand sides is retained only if we feel the need to emphasize what the underlying Banach spaces are, but we will always drop it when $\mathcal{X}_X = \mathcal{C}^2(X)$ which is most often the case. If $w$ depends on a further parameter (for us, $w_X = G_\kappa(X)$), we write that parameter in place of $w$ in (\ref{hwnorm}), (\ref{hwanorm}). If $\alpha > 0$, write $\mathbb{B}^{j,\alpha}_{\mathfrak{h},w,A}$ for the set of polymer activities which satisfy
\begin{equation}
    \Vert K\Vert_{\mathfrak{h},w,A;\mathcal{X}_X} \leqslant \Vert(\zeta_\bullet)_{0,j}\Vert^{\alpha}_{\wp} = \Big(L^{2j}e^{\frac{1}{2}\beta\sum_{k=j}^0\Gamma_{\hspace{-1pt}0;j}}\zeta_\bullet\Big)^{\!\alpha}.
\end{equation}

\subsubsection{Complexification of functionals}\label{AppComplexification}

In one proof, we will need a \enquote{complexification} map to get contractive estimates on charged terms. The whole development of this section is rather elementary, comprising mostly Taylor expansions and basic multilinear algebra.

For a Banach space $\mathcal{X}$, consider the strip $\mathcal{X}^{\mathbb{C},\mathfrak{h}} = \big\lcurl z\in\mathcal{X}^{\mathbb{C}}:\Vert\mathrm{Im}(z)\Vert_\mathcal{X} < \frac{1}{2}\mathfrak{h}\big\rcurl$. If $F \in\mathcal{N}_\mathfrak{h}(\mathcal{X})$, we define $^\mathbb{C\hspace{-1pt}}F$ for $w\in\mathcal{X}^{\mathbb{C},\mathfrak{h}}$ as follows: if $w = x + z$ with $x \in \mathcal{X}$ and $z\in\mathcal{X}^{\mathbb{C}}$ with $\Vert z\Vert_{\mathcal{X}^\mathbb{C}} = \max\lcurl \Vert\mathrm{Re}(z)\Vert_{\mathcal{X}},\Vert\mathrm{Im}(z)\Vert_{\mathcal{X}}\rcurl < \frac{1}{2}\mathfrak{h}$, put
\begin{equation}
    ^\mathbb{C\hspace{-1pt}}F(w) = \sum_{n=0}^\infty\frac{1}{n!}\langle\mathrm{D}^n\vert_x F,z,\ldots,z\rangle
\end{equation}
where $\mathrm{D}^n\vert_x F$ is of course extended to complex arguments by multilinearity. The series converges: if $z = z_0 + iz_1$ are the real and imaginary parts,
\begin{align}
    \big\vert\langle\mathrm{D}^n\vert_x F,z,\ldots,z\rangle\big\vert &\leqslant \sum_{\sigma_1,\ldots,\sigma_n\in\lcurl 0,1\rcurl^n}\big\vert\langle\mathrm{D}^n\vert_x F,z_{\sigma_1},\ldots,z_{\sigma_n}\rangle\big\vert\nonumber\\
    &\leqslant\sum_{\sigma_1,\ldots,\sigma_n\in\lcurl 0,1\rcurl^n}\Vert F\Vert_{n;\langle X\rangle}\Vert z_{\sigma_1}\Vert_{\mathcal{X}}\ldots\Vert z_{\sigma_n}\Vert_{\mathcal{X}}.\nonumber\\
    &\leqslant\mathfrak{h}^n\Vert F\Vert_{n;\langle X\rangle}.
\end{align}
We still need to check that this is independent of the decomposition of $w$, so suppose $w = x + z = x' + z'$. Then $\Vert x - x'\Vert_\mathcal{X} < \mathfrak{h}$ and by Taylor expansion,
\begin{equation}
    \langle\mathrm{D}^n\vert_x F,u_1,\ldots,u_n\rangle = \sum_{m=0}^\infty\frac{1}{m!}\langle\mathrm{D}^{n+m}\vert_{x'}F, x-x',\ldots,x-x',u_1,\ldots,u_n\rangle.\label{FDerivativeTaylor}
\end{equation}
Hence (the series converge absolutely)
\begin{align}
    ^\mathbb{C\hspace{-1pt}}F(x+z) &= \sum_{n,m=0}^\infty\frac{1}{n! m!}\langle \mathrm{D}^{n+m}\vert_{x'}F,x-x',\ldots,x-x',z,\ldots,z\rangle\nonumber\\
    &=\sum_{k=0}^\infty\frac{1}{k!}\sum_{n=0}^k\binom{k}{n}\langle \mathrm{D}^{k}\vert_{x'}F,x-x',\ldots,x-x',z,\ldots,z\rangle\nonumber\\
    &=\sum_{k=0}^\infty\frac{1}{k!}\langle \mathrm{D}^{k}\vert_{x'}F,x-x'+z,\ldots,x-x'+z\rangle\nonumber\\
    &= {^\mathbb{C\hspace{-1pt}}F}(x'+z').
\end{align}
This shows well-definedness.

Some further properties: if $w$ is already real, $^\mathbb{C\hspace{-1pt}}F(w) = F(w)$. $^\mathbb{C\hspace{-1pt}}F$ is also a smooth functional on $\mathcal{X}^{\mathbb{C},\mathfrak{h}}$: if $u_1,\ldots,u_n\in\mathcal{X}^\mathbb{C}$,
\begin{align}
    \frac{\partial^m}{\partial t_1\ldots\partial t_m}\bigg\vert_{t_1=\ldots=t_m = 0}&{^\mathbb{C\hspace{-1pt}}F}\bigg(w + \sum_{i=1}^m t_iu_i\bigg) =\nonumber\\ 
    =&\frac{\partial^m}{\partial t_1\ldots\partial t_m}\bigg\vert_{t_1=\ldots=t_m = 0}\sum_{n=0}^\infty\frac{1}{n!}\bigg\langle \mathrm{D}^n\vert_x F,\psi + \sum_{i=1}^mt_iu_i,\ldots,\psi+\sum_{i=1}^mt_iu_i\bigg\rangle\nonumber\\
    =&\sum_{n=0}^\infty\frac{1}{n!}\langle \mathrm{D}^{n+m}\vert_xF,\psi,\ldots,\psi,u_1,\ldots,u_m\rangle
\end{align}
as we could have guessed from (\ref{FDerivativeTaylor}). Furthermore, with $z$ pure imaginary we can derive
\begin{equation}
    \big\Vert{^\mathbb{C\hspace{-1pt}}F}\big\Vert_{m;\langle w\rangle} \leqslant 2^m\sum_{n=0}^\infty\frac{\Vert z\Vert^n_{\mathcal{X}}}{n!}\Vert F\Vert_{n+m;\langle x\rangle}
\end{equation}
and so
\begin{equation}
    \big\Vert{^\mathbb{C\hspace{-1pt}}F}\big\Vert_{\mathfrak{k};\langle w\rangle} \leqslant \Vert F\Vert_{2\mathfrak{k}+\Vert\mathrm{Im}(w)\Vert_\mathcal{X};\langle \mathrm{Re}(w)\rangle},
\end{equation}
which is finite as long as $2\mathfrak{k} + \Vert\mathrm{Im}(w)\Vert_\mathcal{X}\leqslant\mathfrak{h}$.

\section{Norm \texorpdfstring{\&}{and} regulator estimates}

\subsection{Algebraic inequalities} \label{AppAlgIneq}

The next proposition comprises a menagerie of norm inequalities of various flavors which serve as the technical basis of the paper. This is all standard \cite{ParkCity}; (\ref{n2absorb}) is probably the most interesting: it allows us to absorb multiplicative factors into norms with improved scaling compared to plain $\mathrm{L}^2(X)$ which ultimately enables us to prove (\ref{Gg2}), a strong estimate used to kill polynomials in the field that arise when Taylor expanding polymer activities as in (\ref{smallXthing}).

\begin{lem}
    Let $f,g\in\mathcal{C}^2(\mathbb{T}^2)$ and $X,Y\in\mathcal{P}_{j;s},B\in\mathcal{B}_{j;s}$, $\alpha > 0$. $\mathrm{diam}(X)$ is taken relative to paths entirely within $X$. The following inequalities hold:
    \begin{align}
        \Vert f\Vert_{\mathrm{L}^\infty(X)}^2 &\lesssim_1 \sum_{n=0}^2\Vert\nabla^2\hspace{-1pt}f\Vert^2_{\mathrm{L}^2(X)}\label{sobolevLemma},\\
        \Vert f\Vert_{\mathrm{L}^2(X)},\Vert f\Vert_{\mathrm{L}^2(X)}&\lesssim_1\Vert f\Vert_{\mathcal{B}^2_\infty(X)}\label{L2Bbound},\\
        Y\subset X\in\mathcal{P}_{\mathrm{c};j;s}:\Vert f\Vert^2_{\mathrm{L}^\infty(X)}&\leqslant\frac{2}{\vert Y\vert}\Vert f\Vert^2_{\mathrm{L}^2(Y)} + 2\,\mathrm{diam}(X)^2\Vert\nabla^2\hspace{-1pt}f\Vert_{\mathcal{B}^2_\infty(X)}^2\label{diamBound},\\
        X\in\mathcal{S}:\Vert\nabla f\Vert^2_{\mathcal{C}^1(X)}&\lesssim_1\Vert\nabla f\Vert^2_{\mathrm{L}^2(X)} + \Vert\nabla^2\hspace{-1pt}f\Vert^2_{\mathcal{B}^2_\infty(X)},\label{C1regAbsorb}\\
        \Vert f\Vert_{\mathrm{L}^2(\partial B)}^2&\lesssim_1 \Vert f\Vert^2_{\mathrm{L}^2(B)} + \Vert \nabla f\Vert^2_{\mathrm{L}^\infty(B)},\label{boundaryDerivativeEstimate}\\
        \Vert f\Vert^2_{\mathrm{L}^2(\partial X)}&\leqslant\Vert f\Vert^2_{\mathrm{L}^2(\partial\ell\bar{X}^{(1)})}+\mathcal{O}_1\big(\Vert f\Vert^2_{\mathrm{L}^2(\ell\bar{X}^{(1)}-X)} + \Vert\nabla f\Vert^2_{\mathcal{B}^2_\infty(\ell\bar{X}^{(1)}-X)}\big),\label{closureBoundaryEstimate}\\
        \Vert\nabla(f+g)\Vert^2_{\mathrm{L}^2(X)}&\lesssim \Vert\nabla f\Vert^2_{\mathrm{L}^2(X)} + \mathcal{O}_1\alpha\big(\Vert\nabla f\Vert^2_{\mathrm{L}^2(\partial X)}+\Vert\nabla^2\hspace{-1pt}f\Vert^2_{\mathrm{L}^2(X)}\big)\nonumber\\
        &\hspace{15pt}+\Vert\nabla g\Vert^2_{\mathrm{L}^2(X)} + \mathcal{O}_1\alpha^{-1}\big(\Vert g\Vert^2_{\mathrm{L}^2(\partial X)}+\Vert g\Vert^2_{\mathrm{L}^2(X)}\big)\label{alphaReweight}.
    \end{align}
    Furthermore, we have the scaling relations
    \begin{align}
        \Vert\nabla^n\mathrm{S}_{\ell}f\Vert_{\mathcal{B}^q_p(\ell X)} &\leqslant \ell^{\frac{2}{p\wedge q}-n}\Vert\nabla^n\hspace{-1pt}f\Vert_{\mathcal{B}^q_p(X)},\label{scalingB}\\
        \Vert\nabla^n\mathrm{S}_{\ell}f\Vert_{\mathrm{L}^p(\partial\ell X)} &= \ell^{\frac{1}{p}-n}\Vert\nabla^n\hspace{-1pt}f\Vert_{\mathrm{L}^p(\partial X)},\label{scalingL2B}\\
        \Vert\mathrm{S}_\ell f\Vert_{\wp}&\leqslant \ell^{2/\wp}\Vert f\Vert_{\wp},\label{scalingP}\\
        \Vert\mathrm{S}_{\ell^{-1}} f\Vert_{\wp}&\leqslant \ell^{2/\wp}\Vert f\Vert_{\wp}.\label{scalingP2}
    \end{align}
    Analogous estimates hold with $L$ in place of $\ell$.

    Finally, take $X = X_1$ small and fix a $\lambda > 1$. If $\ell > \mathcal{O}_\lambda$, there are disjoint translates $X_2,\ldots,X_{32\lceil\lambda\rceil} \subset\ell\bar{X}^{(1)}$ of $X_1$ with $\mathfrak{X} = \bigsqcup_i X_i$ such that
    \begin{equation}
        \lambda\Vert f\Vert_{\mathrm{L}^2(X)}^2\leqslant\Vert f\Vert^2_{\mathrm{L}^2(\mathfrak{X})} + \mathcal{O}_\lambda\Vert \nabla f\Vert^2_{\mathcal{B}^2_\infty(\ell\bar{X}^{(1)})}.\label{n2absorb}
    \end{equation}
\end{lem}

\begin{proof}
    The scaling relations are all straightforward. Now consider a block $B\in\mathcal{B}_j$. We may assume it lies in the first quadrant and touches the origin $(0,0)$. Define $h(x) = x_1x_2f(x)$ so that
    \begin{equation}
        h(x) = \int_0^{x_2}\int_{0}^{x_1}\partial_1\partial_2h(y_1,y_2)\dd y_1 \dd y_2
    \end{equation}
    and
    \begin{equation}
        x_1x_2\vert f(x)\vert\leqslant\Vert\partial_1\partial_2h\Vert_{\mathrm{L}^1(B)}\leqslant\sum_{n=0}^2\Vert\nabla^2\hspace{-1pt}f\Vert_{\mathrm{L}^1(B)}.
    \end{equation}
    if $x\in[1/2,1]^2$ we divide by the factor $x_1x_2$, otherwise we appeal to symmetry. So by Hölder,\begin{equation}
        \Vert f\Vert^2_{\mathrm{L}^\infty(B)} \leqslant 3\sum_{n=0}^2\Vert\nabla^n\hspace{-1pt}f\Vert^2_{\mathrm{L}^2(B)},
    \end{equation}
    proving (\ref{sobolevLemma}) after summing over $B\in\mathcal{B}_j(X)$. For (\ref{L2Bbound}), we simply compute
    \begin{equation}
        \Vert f\Vert^2_{\mathrm{L}^2(\partial X)}\leqslant\sum_{B\in\mathcal{B}_j(X)}\Vert f\Vert_{\mathrm{L}^2(\partial B)} \leqslant 4\sum_{B\in\mathcal{B}_j(X)}\Vert{f}\Vert_{\mathrm{L}^\infty(B)}^2 = 4\Vert f\Vert^2_{\mathcal{B}^2_\infty(X)}
    \end{equation}
    and analogously for $\Vert f\Vert_{\mathrm{L}^2(X)}$. Next, it is obvious that for $x,y\in X$, $X$ connected,
    \begin{equation}
        \vert f(x)\vert \leqslant \vert f(y)\vert + \mathrm{diam}(X)\Vert\nabla f\Vert_{\mathrm{L}^\infty(X)}.
    \end{equation}
    If $y\in Y$, by squaring and taking the average over $y$, we get (\ref{diamBound}). Alternatively, fixing a point $x$ in the boundary of a block $B$ and a point $y\in B$,
    \begin{equation}
        \vert f(x)\vert^2\lesssim_1\vert f(y)\vert^2+\Vert\nabla f\Vert^2_{\mathrm{L}^\infty(B)}
    \end{equation}
    and again averaging, first over $y$ and then over $x\in\partial B$, we obtain (\ref{boundaryDerivativeEstimate}). For (\ref{C1regAbsorb}), by (\ref{diamBound}),
    \begin{equation}
        \Vert\nabla f\Vert^2_{\mathcal{C}^1(X)}\leqslant \Vert\nabla f\Vert^2_{\mathrm{L}^\infty(X)} + \Vert\nabla^2\hspace{-1pt}f\Vert^2_{\mathrm{L}^\infty(X)}\lesssim_1\Vert\nabla f\Vert_{\mathrm{L}^2(X)}^2+\Vert\nabla^2\hspace{-1pt}f\Vert^2_{\mathcal{B}^2_\infty(X)}.
    \end{equation}
    Now observe that (\ref{closureBoundaryEstimate}) follows from (\ref{boundaryDerivativeEstimate}),
    \begin{align}
        \Vert f\Vert^2_{\mathrm{L}^2(\partial X)} &\leqslant \Vert f\Vert^2_{\mathrm{L}^2(\partial\ell\bar{X}^{(1)})}+\Vert f\Vert^2_{\mathrm{L}^2(\partial(\ell\bar{X}^{(1)}-X))}\nonumber\\
        &\leqslant\Vert f\Vert^2_{\mathrm{L}^2(\partial\ell\bar{X}^{(1)})} + \mathcal{O}_1\big(\Vert f\Vert^2_{\mathrm{L}^2(\ell\bar{X}^{(1)}-X)}+\Vert\nabla f\Vert^2_{\mathcal{B}^2_\infty(\ell\bar{X}^{(1)}-X)}\big).
    \end{align}
    For (\ref{alphaReweight}), we apply integration by parts to find
    \begin{align}
        \Vert \nabla(f+g)\Vert^2_{\mathrm{L}^2(X)} &= \Vert\nabla f\Vert^2_{\mathrm{L}^2(X)}+\Vert\nabla g\Vert^2_{\mathrm{L}^2(X)}+2\sum_{\vert\mu\vert = 1}\langle\partial^\mu\! f,\partial^\mu\! g\rangle_{\mathrm{L}^2(X)}\nonumber\\
        &\leqslant\Vert\nabla f\Vert^2_{\mathrm{L}^2(X)}+\Vert\nabla g\Vert^2_{\mathrm{L}^2(X)}+2\sum_{\vert\mu\vert = 1}\!\big(\vert\langle \alpha\partial^{2\mu}\!f,\alpha^{-1}g\rangle_{\mathrm{L}^2(X)}\vert + \vert\langle \alpha\partial^{\mu}\!f,\alpha^{-1}g\rangle_{\mathrm{L}^2(\partial X)}\vert\big).
    \end{align}
    By Cauchy-Schwarz and the arithmetic-geometric inequality,
    \begin{align}
        \Vert\nabla(f+g)\Vert^2_{\mathrm{L}^2(X)} &\leqslant \Vert\nabla f\Vert^2_{\mathrm{L}^2(X)} + \Vert\nabla g\Vert^2_{\mathrm{L}^2(X)}\nonumber\\
        &\hspace{15pt}+4\alpha\big(\Vert\nabla^2\hspace{-1pt} f\Vert^2_{\mathrm{L}^2(X)}+\Vert\nabla f\Vert^2_{\mathrm{L}^2(\partial X)}\big) + 8\alpha^{-1}\big(\Vert g\Vert^2_{\mathrm{L}^2(X)}+\Vert g\Vert^2_{\mathrm{L}^2(\partial X)}\big).
    \end{align}
    For (\ref{n2absorb}), we first have to find the polymer translates. This is easy: as long as $\ell\geqslant\mathcal{O}_\lambda$, we can just pick some sidelength-$\ell$ square of $\ell\bar{X}^{(1)}$ which $X$ intersects and make copies of $X$ inside that square. By keeping the copies close together, the radius of the resulting pattern will be of order $\lambda$ (even $\sqrt{\lambda}$ since we are in 2d). Then by (\ref{diamBound}),
    \begin{align}
        \lambda\Vert f\Vert^2_{\mathrm{L}^2(X)} &\leqslant \Vert f\Vert^2_{\mathrm{L}^2(X_1)} + \frac{\lambda-1}{32\lceil\lambda\rceil-1}\sum_{k=2}^{32\lceil\lambda\rceil}\Vert f\Vert^2_{\mathrm{L}^2(X_1)}\nonumber\\
        &\leqslant\Vert f\Vert^2_{\mathrm{L}^2(X_1)} + \frac{\vert X_1\vert^2}{32}\sum_{k=2}^{32\lceil\lambda\rceil}\Vert f\Vert^2_{\mathrm{L}^\infty(X_1)}\nonumber\\
        &\leqslant\Vert f\Vert^2_{\mathrm{L}^2(X_1)} + \sum_{k=2}^{32\lceil\lambda\rceil}\Vert f\Vert^2_{\mathrm{L}^2(X_k)}\nonumber + \mathcal{O}_\lambda\Vert f\Vert^2_{\mathrm{L}^\infty(\ell\bar{X}^{(1)})}\nonumber\\
        &\leqslant\Vert f\Vert^2_{\mathrm{L}^2(\mathfrak{X})} + \mathcal{O}_\lambda\Vert\nabla f\Vert^2_{\mathcal{B}^2_\infty(\ell\bar{X}^{(1)})},
    \end{align}
    concluding the proof.
\end{proof}

As stated, we will use (\ref{n2absorb}) in the next proposition to remove an otherwise problematic constant. In that case $\lambda = 2$ so all the $\mathcal{O}_\lambda\!$'s become $\mathcal{O}_1\!$'s.
The following estimates on regulators are still of a purely algebraic nature, but are necessary to have good control of the fluctuation step.

\begin{lem}
    Because $\mathfrak{c}^{(2)}\lesssim_1 1\lesssim_1\mathfrak{c}^{(1)}\wedge\mathfrak{c}^{(3)}\wedge\ell$, we have the estimates
    \begin{align}
        X\subset Y: G_\kappa(X,\phi)&\leqslant G_\kappa(Y,\phi),\label{strongAbsorb}\\
        X\in\mathcal{P}_{\mathrm{c};j;s}: G_\kappa(X,\mathrm{S}_\ell\phi+\psi)&\leqslant G_\kappa(\bar{X}^{(1)},\phi)g_\kappa(X,\psi),\label{Gg}\\
        X\in\mathcal{S}_{j;s}:G_\kappa(X,\mathrm{S}_\ell\phi+\psi)&\leqslant G_\kappa(\bar{X}^{(1)},\phi)^{1/2}g_\kappa(X,\psi).\label{Gg2}
    \end{align}
\end{lem}

\begin{proof}
    We can take $\kappa = 1$. For (\ref{strongAbsorb}), by induction take $Y = X\cup B$ with $B$ a single square. Since
    \begin{equation}
        \Vert\nabla\phi\Vert^2_{\mathrm{L}^2(\partial X)}\leqslant\Vert\nabla\phi\Vert^2_{\mathrm{L}^2(\partial Y)} + \Vert\nabla\phi\Vert^2_{\mathrm{L}^2(\partial B)},
    \end{equation}
    by (\ref{boundaryDerivativeEstimate}),
    \begin{equation}
        \Vert\nabla\phi\Vert^2_{\mathrm{L}^2(\partial X)}\leqslant\Vert\nabla\phi\Vert^2_{\mathrm{L}^2(\partial Y)} + \mathcal{O}_1\big(\Vert\nabla\phi\Vert^2_{\mathrm{L}^2(B)}+\Vert\nabla^2\phi\Vert^2_{\mathrm{L}^\infty(B)}\big),
    \end{equation}
    so by (\ref{diamBound}),
    \begin{align}
        \log G_1(X,\phi) &\leqslant\mathfrak{c}^{(1)}\big(\Vert\nabla\phi\Vert^2_{\mathrm{L}^2(X)}+\Vert\nabla^2\phi\Vert^2_{\mathcal{B}^2_\infty(X)}\big) + \mathfrak{c}^{(2)}\Vert\nabla\phi\Vert^2_{\mathrm{L}^2(\partial Y)} + \mathcal{O}_1\mathfrak{c}^{(2)}\big(\Vert\nabla\phi\Vert^2_{\mathrm{L}^2(B)}+\Vert\nabla^2\phi\Vert^2_{\mathrm{L}^\infty(B)}\big)\nonumber\\
        &\leqslant\max\!\big\lcurl\mathfrak{c}^{(1)},\mathcal{O}_1\mathfrak{c}^{(2)}\big\rcurl\big(\Vert\nabla\phi\Vert^2_{\mathrm{L}^2(Y)}+\Vert\nabla^2\phi\Vert^2_{\mathcal{B}^2_\infty(Y)}\big) + \mathfrak{c}^{(2)}\Vert\nabla\phi\Vert^2_{\mathrm{L}^2(\partial Y)}\nonumber\\
        &\leqslant\log G_1(Y,\phi)
    \end{align}
    where we have had to require $\mathfrak{c}^{(1)}\gtrsim_1\mathfrak{c}^{(2)}$. Moving on, by (\ref{L2Bbound},\ref{alphaReweight}) and the triangle inequality,
    \begin{align}
        \log G_1(X,\mathrm{S}_\ell\phi+\psi) \leqslant\, &\mathfrak{c}^{(1)}\Vert\nabla\mathrm{S}_\ell\phi\Vert^2_{\mathrm{L}^2(X)} + \mathcal{O}_1\alpha\mathfrak{c}^{(1)}\big(\Vert\nabla\mathrm{S}_\ell\phi\Vert^2_{\mathrm{L}^2(\partial X)} + \Vert\nabla^2\mathrm{S}_\ell\phi\Vert^2_{\mathrm{L}^2(X)}\big)+\nonumber\\
        &\mathfrak{c}^{(1)}\Vert\nabla\psi\Vert^2_{\mathrm{L}^2(X)} + \mathcal{O}_1\alpha^{-1}\mathfrak{c}^{(1)}\big(\Vert\psi\Vert^2_{\mathrm{L}^2(\partial X)} + \Vert\psi\Vert^2_{\mathrm{L}^2(X)}\big)+\nonumber\\
        & 2\mathfrak{c}^{(1)}\Vert\nabla^2\mathrm{S}_\ell\phi\Vert^2_{\mathcal{B}^{2}_\infty(X)}+2\mathfrak{c}^{(1)}\Vert\nabla^2\psi\Vert^2_{\mathcal{B}^{2}_\infty(X)}+2\mathfrak{c}^{(2)}\Vert\nabla\mathrm{S}_\ell\phi\Vert^2_{\mathrm{L}^2(\partial X)} + 2\mathfrak{c}^{(2)}\Vert\nabla\psi\Vert^2_{\mathrm{L}^2(\partial X)}\nonumber\\
        \leqslant\,&\mathfrak{c}^{(1)}\Vert\nabla\mathrm{S}_\ell\phi\Vert^2_{\mathrm{L}^2(X)} + (\mathcal{O}_1\alpha+2)\mathfrak{c}^{(1)}\Vert\nabla^2\mathrm{S}_\ell\phi\Vert^2_{\mathcal{B}^2_\infty(X)} + \mathcal{O}_1(\alpha\mathfrak{c}^{(1)}+\mathfrak{c}^{(2)})\Vert\nabla\mathrm{S}_\ell\phi\Vert^2_{\mathrm{L}^2(\partial X)}+\nonumber\\
        &\mathcal{O}_1\max\!\big\lcurl \alpha^{-1}\mathfrak{c}^{(1)},\mathfrak{c}^{(1)}+\mathfrak{c}^{(2)}\big\rcurl\sum_{k=0}^2\Vert\nabla^k\psi\Vert^2_{\mathcal{B}^2_\infty(X)}.\label{Gg1}
    \end{align}
    We estimate the third term by (\ref{closureBoundaryEstimate}):
    \begin{equation}
        \Vert\nabla\mathrm{S}_\ell\phi\Vert^2_{\mathrm{L}^2(\partial X)}\leqslant \Vert\nabla\mathrm{S}_\ell\phi\Vert^2_{\mathrm{L}^2(\partial\ell\bar{X}^{(1)})} + \mathcal{O}_1\big(\Vert\nabla\mathrm{S}_\ell\phi\Vert^2_{\mathrm{L}^2(\ell\bar{X}^{(1)}-X)}+\Vert\nabla^2\mathrm{S}_\ell\phi\Vert^2_{\mathcal{B}^2_\infty(\ell\bar{X}^{(1)}-1)}\big).
    \end{equation}
    Plugging this in and rescaling,
    \begin{align}
        \log G_1(X,\mathrm{S}_\ell\phi+\psi) \leqslant\,&\max\!\big\lcurl\mathfrak{c}^{(1)},\mathcal{O}_1(\alpha\mathfrak{c}^{(1)}+\mathfrak{c}^{(2)})\big\rcurl \Vert\nabla\phi\Vert^2_{\mathrm{L}^2(\bar{X}^{(1)})} +\nonumber\\
        &\ell^{-2}\max\!\big\lcurl(\mathcal{O}_1\alpha+2)\mathfrak{c}^{(1)},\mathcal{O}_1(\alpha\mathfrak{c}^{(1)}+\mathfrak{c}^{(2)})\big\rcurl\Vert\nabla^2\phi\Vert^2_{\mathcal{B}^2_\infty(\bar{X}^{(1)})}+\nonumber\\
        &\ell^{-1}\mathcal{O}_1(\alpha\mathfrak{c}^{(1)}+\mathfrak{c}^{(2)})\Vert\nabla\phi\Vert^2_{\mathrm{L}^2(\partial\bar{X}^{(1)})}+\nonumber\\
        &\mathcal{O}_1\max\!\big\lcurl\alpha^{-1}\mathfrak{c}^{(1)},\mathfrak{c}^{(1)}+\mathfrak{c}^{(2)}\big\rcurl\sum_{k=0}^2\Vert\nabla^k\psi\Vert^2_{\mathcal{B}^2_\infty(X)}.
    \end{align}
    To restore the old constants $\mathfrak{c}^{(1)},\mathfrak{c}^{(2)}$, we pick $\alpha$ such that $\alpha\mathfrak{c}^{(1)}\leqslant\mathfrak{c}^{(2)}$; then $\mathcal{O}_1(\alpha\mathfrak{c}^{(1)}+\mathfrak{c}^{(2)})\leqslant\mathcal{O}_1\mathfrak{c}^{(2)}\leqslant\mathfrak{c}^{(1)}$. Pick $\ell$ so that $\mathcal{O}_1\ell^{-1}(\alpha\mathfrak{c}^{(1)}+\mathfrak{c}^{(2)})\leqslant\mathfrak{c}^{(2)}$ and $\ell^{-2}(\mathcal{O}_1\alpha+2)\leqslant 1$. Finally, take $\mathfrak{c}^{(3)}\geqslant\mathcal{O}_1\max\!\big\lcurl\alpha^{-1}\mathfrak{c}^{(1)},\mathfrak{c}^{(1)}+\mathfrak{c}^{(2)}\big\rcurl$. Then (\ref{Gg}) follows.

    Now if $X$ is small and the conditions on constants as above, take (\ref{Gg1}) and combine it with (\ref{closureBoundaryEstimate},\ref{n2absorb}):
    \begin{align}
        2\log G_1(X,\mathrm{S}_\ell\phi+\psi) \leqslant\,&2\mathfrak{c}^{(1)}\Vert\nabla\mathrm{S}_\ell\phi\Vert^2_{\mathrm{L}^2(X)}+\mathcal{O}_1(\alpha+1)\mathfrak{c}^{(1)}\Vert\nabla^2\mathrm{S}_\ell\phi\Vert^2_{\mathcal{B}^2_\infty(X)}+\nonumber\\
        &\mathcal{O}_1(\alpha\mathfrak{c}^{(1)}+\mathfrak{c}^{(2)})\Vert\nabla\mathrm{S}_\ell\phi\Vert^2_{\mathrm{L}^2(\partial X)}+\mathcal{O}_1\max\!\big\lcurl\alpha^{-1}\mathfrak{c}^{(1)},\mathfrak{c}^{(1)}+\mathfrak{c}^{(2)}\big\rcurl\sum_{k=0}^2\Vert\nabla^k\psi\Vert^2_{\mathcal{B}^2_\infty(X)}\nonumber\\
        \leqslant\,&\max\!\big\lcurl\mathfrak{c}^{(1)},\mathcal{O}_1(\alpha\mathfrak{c}^{(1)}+\mathfrak{c}^{(2)})\big\rcurl\Vert\nabla^2\phi\Vert^2_{\mathrm{L}^2(\bar{X}^{(1)})}+\nonumber\\
        &\ell^{-2}\mathcal{O}_1\max\!\big\lcurl(\alpha+1)\mathfrak{c}^{(1)},\alpha\mathfrak{c}^{(1)}+\mathfrak{c}^{(2)}\big\rcurl\Vert\nabla^2\phi\Vert^2_{\mathcal{B}^2_\infty(\bar{X}^{(1)})}+\\
        &\ell^{-1}\mathcal{O}_1(\alpha\mathfrak{c}^{(1)}+\mathfrak{c}^{(2)})\Vert\nabla\phi\Vert^2_{\mathrm{L}^2(\partial\bar{X}^{(1)})}+\mathcal{O}_1\max\!\big\lcurl\alpha^{-1}\mathfrak{c}^{(1)},\mathfrak{c}^{(1)}+\mathfrak{c}^{(2)}\big\rcurl\sum_{k=0}^2\Vert\nabla^k\psi\Vert^2_{\mathcal{B}^2_\infty(X)}.\nonumber
    \end{align}
    To restore the constants, proceed as before.
\end{proof}

\subsection{Fluctuation step bounds} \label{AppFluctStepBounds}

We combine these estimates from the previous section with properties of the scale decomposition. These are a mix of \cite{Falco1} and \cite{Dimock2000}; because it considers the critical phase, the former also does an analysis of the second order, which is much more complicated than the first-order analysis we have to perform.

\begin{lem}
    Take $X$ connected and let $\Psi_{\hspace{-1pt}\varepsilon;j},\Psi_{\hspace{-1pt}\varepsilon;j;s}$ be the various fluctuation fields. First take $j \leqslant -1$.
    \begin{equation}
        \mathbb{E}\big[g_{\kappa_*/8}(X,\Psi_{\varepsilon;j;s})\big] \leqslant 2^{\log_L(\ell)\vert X\vert}.\label{Eg}
    \end{equation}
    Furthermore, let $\kappa \geqslant \kappa_*/2$. Then
    \begin{align}
        \mathbb{E}\Big[\sup_{t\in[0,1]}G_\kappa(X,\mathrm{S}_L\phi+(1-t)\Psi_{\hspace{-1pt}\varepsilon;j}+\Psi_{\hspace{-1pt}\dot\varepsilon;j})\Big]&\leqslant 2^{\vert X\vert}G_\kappa(\bar{X}^L,\phi),\label{supermpr}\\
        X\in\mathcal{S}:\big(1+\Vert\nabla\phi\Vert_{\mathcal{C}^1(\bar{X}^L)}\big)^{\!2\,}\mathbb{E}\Big[\sup_{t\in[0,1]}G_\kappa(X,t\mathrm{S}_L\phi+\Psi_{\hspace{-1pt}\varepsilon;j})\Big]&\lesssim_1 \kappa G_\kappa(\bar{X}^L,\phi).\label{smallXthing}
    \end{align}
    In case $j = 0$, if $\kappa \geqslant \kappa_*/8$,
    \begin{equation}
        \mathbb{E}\Big[\sup_{t\in[0,1]}G_\kappa(X,\phi+(1-t)\Psi_{\hspace{-1pt}\varepsilon;0}+\Psi_{\hspace{-1pt}\dot\varepsilon;0})\Big]\leqslant 2G_{\kappa/2}(\mathbb{T}^2,\phi).
    \end{equation}
\end{lem}
We note that (\ref{supermpr}) immediately implies the weaker
\begin{equation}
    \mathbb{E}\big[G_\kappa(X,\mathrm{S}_L\phi+\Psi_{\hspace{-1pt}\varepsilon;j})\big]\leqslant 2^{\vert X\vert}G_\kappa(\bar{X}^L,\phi).\label{supermpr0}
\end{equation}

\begin{proof}

For (\ref{Eg}), first take $\vert\mu\vert\leqslant4$ and estimate
\begin{equation}
    \mathbb{E}\big[\!\exp(\mathcal{O}_1\kappa_*^{-1}\mathfrak{c}^{(3)}\Vert\partial^\mu\Psi_{\hspace{-1pt}\varepsilon;j;s}\Vert^2_{\mathrm{L}^2(X)})\big] \leqslant \exp(\mathcal{O}_1\kappa_*^{-1}\mathfrak{c}^{(3)}\mathrm{Tr}\,\mathrm{Cov}(\partial^\mu\Psi_{\hspace{-1pt}\varepsilon;j;s})).
\end{equation}
This holds as long as the $\mathrm{L}^2$ operator norm of $\mathcal{O}_1\kappa_*^{-1}\mathfrak{c}^{(3)}\mathrm{Cov}(\partial^\mu\Psi_{\hspace{-1pt}\varepsilon;j;s})$ is bounded by $1/2$, that is, we have to take $\kappa_*$ such that $ \mathfrak{c}^{(3)}\Vert\mathrm{Cov}(\partial^\mu\Psi_{\hspace{-1pt}\varepsilon;j;s})\Vert_{\mathcal{B}(\mathrm{L}^2(X))}\lesssim_1 \kappa_*$. Since
\begin{equation}
    \mathrm{Cov}(\partial^\mu\Psi_{\hspace{-1pt}\varepsilon;j;s})(f)(x) = \int_{\mathbb{T}_{j;s}^2}\mathbb{I}_{x,y\in X}(-\partial^{2\mu}\Gamma_{\hspace{-1pt}\varepsilon;j;s})(x-y)f(y)\dd y,
\end{equation}
we apply Young's inequality and the finite-range property to find
\begin{equation}
        \Vert\mathrm{Cov}(\partial^\mu\Psi_{\hspace{-1pt}\varepsilon;j;s})\Vert_{\mathcal{B}(\mathrm{L}^2(X))}\leqslant\Vert\partial^{2\mu}\Gamma_{\hspace{-1pt}\varepsilon;j;s}\Vert_{\mathrm{L}^1(\mathbb{T}^2)}\lesssim_1 \ell^2\Vert\partial^{2\mu}\Gamma_{\hspace{-1pt}\varepsilon;j;s}\Vert_{\mathrm{L}^\infty(\mathbb{T}^2)}\lesssim_1 \mathfrak{m}\log(\ell) = \log(L).
\end{equation}
Hence we need $\kappa_* \gtrsim_1 \mathfrak{c}^{(3)}\log(L)$ which we do have since we took $\mathfrak{c}^{(4)}\gtrsim_1\mathfrak{c}^{(3)}$.

Now, because $\mathrm{Tr}\,\mathrm{Cov}(\partial^\mu\Psi_{\hspace{-1pt}\varepsilon;j;s}) = \vert X\vert(-\partial^{2\mu}\Gamma_{\hspace{-1pt}\varepsilon;j;s})(0)$,
\begin{equation}
    \mathbb{E}\big[\!\exp(\mathcal{O}_1\kappa_*^{-1}\mathfrak{c}^{(3)}\Vert\partial^\mu\Psi_{\hspace{-1pt}\varepsilon;j;s}\Vert^2_{\mathrm{L}^2(X)}\big]\leqslant\exp(\mathcal{O}_1\kappa_*^{-1}\mathfrak{c}^{(3)}\log(\ell)\vert X\vert).
\end{equation}
By (\ref{sobolevLemma}),
\begin{equation}
    \log g_{\kappa_*/8}(X,\Psi_{\hspace{-1pt}\varepsilon;j;s})\lesssim_1\kappa_*^{-1}\mathfrak{c}^{(3)}\sum_{n=0}^4\sum_{B\in\mathcal{B}_j(X)}\Vert\nabla^n\Psi_{\hspace{-1pt}\varepsilon;j;s}\Vert^2_{\mathrm{L}^2(B)}\lesssim_1\kappa_*^{-1}\mathfrak{c}^{(3)}\sum_{\vert\mu\vert\leqslant4}\Vert\partial^\mu\Psi_{\hspace{-1pt}\varepsilon;j;s}\Vert^2_{\mathrm{L}^2(X)}.
\end{equation}
There are $31$ terms in this sum, so by Hölder and $\mathfrak{c}^{(4)}\gtrsim_1\mathfrak{c}^{(3)}$, we get (\ref{Eg}):
\begin{equation}
    \mathbb{E}\big[g_{\kappa_*/8}(X,\Psi_{\hspace{-1pt}\varepsilon;j;s})\big]\leqslant\prod_{\vert\mu\vert\leqslant 4}\mathbb{E}\big[\!\exp(\mathcal{O}_1\kappa_{*}^{-1}\mathfrak{c}^{(3)}\Vert\partial^\mu\Psi_{\hspace{-1pt}\varepsilon;j;s}\Vert^2_{\mathrm{L}^2(X)})\big]^{\hspace{-1pt}1/31}\leqslant 2^{\log_L(\ell)\vert X\vert}.
\end{equation}
For (\ref{supermpr}) we argue inductively:
\begin{align}
    \mathbb{E}\Big[\sup_{t\in[0,1]}G_\kappa(X,\mathrm{S}_L\phi,(1-t)\Psi_{\hspace{-1pt}\varepsilon;j}+\Psi_{\hspace{-1pt}\dot\varepsilon;j})\Big]&= \mathbb{E}\bigg[\sup_{t\in[0,1]}G_\kappa\bigg(X,\mathrm{S}_{L}\phi+\sum_{s=0}^{\mathfrak{m}-1}\mathrm{S}_{\ell^s}\big((1-t)\Psi_{\hspace{-1pt}\varepsilon;j;s}+t\Psi_{\hspace{-1pt}\dot\varepsilon;j;s}\big)\bigg)\bigg]\nonumber\\
    &\leqslant G_\kappa(\bar{X}^L,\phi)\prod_{s=0}^{\mathfrak{m}-1}\mathbb{E}\Big[\sup_{t\in[0,1]}g_\kappa(\bar{X}^{(s)},(1-t)\Psi_{\hspace{-1pt}\varepsilon;j;s}+t\Psi_{\hspace{-1pt}\dot\varepsilon;j;s})\Big]\nonumber\\
    &\leqslant G_\kappa(\bar{X}^L,\phi)\prod_{s=0}^{\mathfrak{m}-1}\mathbb{E}\big[g_{\kappa_*/8}(\bar{X}^{(s)},\Psi_{\hspace{-1pt}\varepsilon;j;s})\big]^{1/2}\mathbb{E}\big[g_{\kappa_*/8}(\bar{X}^{(s)},\Psi_{\hspace{-1pt}\dot\varepsilon;j;s})\big]^{1/2}\nonumber\\
    &\leqslant 2^{\log_L(\ell)\sum_{s=0}^{\mathfrak{m}-1}\vert\bar{X}^{(s)}\vert}G_\kappa(\bar{X}^L,\phi)\nonumber\\
    &\leqslant 2^{\log_L(\ell)\mathfrak{m}\vert X\vert}G_\kappa(\bar{X}^L,\phi)\nonumber\\
    &=2^{\vert X\vert}G_\kappa(\bar{X}^L,\phi).
\end{align}

For (\ref{smallXthing}), first note that ($\kappa \geqslant 1$)
\begin{equation}
    (1+x)^{a}\lesssim_a \kappa^{a/2}e^{\kappa^{-1}x^2/2},
\end{equation}
so (taking $\mathfrak{c}^{(1)}\gtrsim_1 1$ since we wish to apply (\ref{C1regAbsorb}))
\begin{equation}
    \big(1+\Vert\nabla\phi\Vert_{\mathcal{C}^1(\bar{X}^L)}\big)^{\!2\,}\lesssim_1\kappa G_\kappa(\bar{X}^L,\phi)^{1/2}.\label{polyEstimate}
\end{equation}
Proceeding inductively as before but applying (\ref{Gg2}) at the last step,
\begin{align}
    \mathbb{E}\Big[\sup_{t\in[0,1]}G_\kappa(X,t\mathrm{S}_L\phi+\Psi_{\hspace{-1pt}\varepsilon;j})\Big] &\leqslant \prod_{s=0}^{\mathfrak{m}-1}\mathbb{E}\big[g_{\kappa_*/2}(\bar{X}^{(s)},\Psi_{\hspace{-1pt}\varepsilon;j;s})\big]\sup_{t\in[0,1]}G_\kappa(\bar{X}^L,t\phi)^{1/2}\nonumber\\
    &\lesssim_1 G_\kappa(\bar{X}^L,t\phi)^{1/2}.
\end{align}
By (\ref{polyEstimate}), the proof is finished, except the $j = 0$ case. Here similar but simpler reasoning works (since $\kappa$ is of order $\log(L)$ and the last covariance is of order $\mathcal{O}_1$, there are no issues).
\end{proof}

Having established the basic estimates, we can obtain the crucial scaling properties of the polymer activities. The first proposition is very general and so very suboptimal; we have to use more model-specific ideas to beat the scaling factors coming up here.

\begin{lem}
    Let $X\in\mathcal{P}_{\mathrm{c};j}$, $\kappa\geqslant\kappa_*/2$. Then
    \begin{equation}
    \big\Vert\mathbb{E}\big[K_j(\cdot,\mathrm{S}_L(\cdot)+\Psi_{\hspace{-1pt}\varepsilon;j})\big]\big\Vert_{\mathfrak{h},\kappa;\langle X\rangle;\mathcal{C}^2(\bar{X}^L)} \leqslant 2^{\vert X\vert}\Vert K_j\Vert_{\mathfrak{h},\kappa;\langle X\rangle}.
    \end{equation}
\end{lem}
\begin{proof}
    We just unravel the definitions.
    \begin{align}
        \big\vert\mathrm{D}^n\vert_\phi\mathbb{E}\big[K_j(X,\mathrm{S}_L(\cdot)+\Psi_{\hspace{-1pt}\varepsilon;j})\big]\big\vert &= \Bigg\vert\frac{\partial^n}{\partial t_1\ldots\partial t_n}\bigg\vert_{t_1=\ldots=t_n=0}\mathbb{E}\bigg[K_j\bigg(X,\mathrm{S}_L\phi+\Psi_{\hspace{-1pt}\varepsilon;j}+\sum_{i=1}^n t_i\mathrm{S}_Lf_i\bigg)\bigg]\Bigg\vert\nonumber\\
        &\leqslant \mathbb{E}\big[\big\vert\langle\mathrm{D}^n\vert_{\mathrm{S}_L\phi+\Psi_{\hspace{-1pt}\varepsilon;j}}K_j(X),\mathrm{S}_Lf_1,\ldots,\mathrm{S}_Lf_n\rangle\big\vert\big]\nonumber\\
        &\leqslant\mathbb{E}\big[\Vert K\Vert_{n;\langle X,\mathrm{S}_L\phi+\Psi_{\hspace{-1pt}\varepsilon;j}\rangle}\big]\Vert f_1\Vert_{\mathcal{C}^2(\bar{X}^L)}\ldots \Vert f_n\Vert_{\mathcal{C}^2(\bar{X}^L)}.
    \end{align}
So $\big\Vert\mathbb{E}\big[K_j(\cdot,\mathrm{S}_L(\cdot)+\Psi_{\hspace{-1pt}\varepsilon;j})\big]\big\Vert_{\mathfrak{h};\langle X,\phi\rangle;\mathcal{C}^2(\bar{X}^L)}\leqslant\mathbb{E}\big[\Vert K\Vert_{\mathfrak{h};\langle X,\mathrm{S}_L\phi+\Psi_{\hspace{-1pt}\varepsilon;j}\rangle}\big]$. By introducing $G_\kappa$ and using (\ref{supermpr0}),
\begin{equation}
    \big\Vert\mathbb{E}\big[K_j(\cdot,\mathrm{S}_L(\cdot)+\Psi_{\hspace{-1pt}\varepsilon;j})\big]\big\Vert_{\mathfrak{h};\langle X,\phi\rangle;\mathcal{C}^2(\bar{X}^L)}\leqslant \mathbb{E}\big[G_\kappa(X,\mathrm{S}_L\phi+\Psi_{\hspace{-1pt}\varepsilon;j})\big]\Vert K\Vert_{\mathfrak{h},\kappa;\langle X\rangle}\leqslant 2^{\vert X\vert}G_\kappa(\bar{X}^L,\phi)\Vert K\Vert_{\mathfrak{h},\kappa;\langle X\rangle}.
\end{equation}
\end{proof}

The stronger scaling, as we've said, requires model-specific ideas. We will control small sets by using a strong contraction property for charged terms and by normalizing neutral terms by extracting the vacuum energy contribution. In the following propositions, $\mathrm{Rem}$ denotes the Taylor remainder.

\begin{lem}
    Let $X\in\mathcal{P}_{\mathrm{c};j}$.
    \begin{equation}
        \Vert\mathrm{Rem}_m K_j\Vert_{\mathfrak{h};\langle X,\phi\rangle}\leqslant\big(1+\mathfrak{h}^{-1}\Vert\phi\Vert_{\mathcal{C}^2(X)}\big)^{\! m+1\,}\sup_{t\in[0,1]}\sum_{n=m+1}^\infty\frac{\mathfrak{h}^n}{n!}\Vert K_j\Vert_{n;\langle X,t\phi\rangle}.\label{Rem0}
    \end{equation}
    Furthermore if $X\in\mathcal{S}$ and $x_0\in X$, denote $\delta\phi(x) = \phi(x)-\phi(x_0)$. Then
    \begin{equation}
        \Vert K(\cdot,\mathrm{S}_L\delta(\cdot)+\psi)\Vert_{\mathfrak{h};\langle X,\phi\rangle;\mathcal{C}^2(\bar{X}^L)}\leqslant\Vert K\Vert_{8\mathfrak{h}/L;\langle X,\mathrm{S}_L\delta\phi+\psi\rangle}.\label{vanishingAtPoint}
    \end{equation}
\end{lem}

\begin{proof}
    For (\ref{vanishingAtPoint}), we can take $x_0\in X$ which has a path joining it to any other point of lenght at most $2\sqrt{2}$. If $f\in\mathcal{C}^2(X)$,
    \begin{align}
        \Vert\mathrm{S}_L\delta f\Vert_{\mathcal{C}^2(X)} &= \max\!\bigg\lcurl\bigg\Vert\int_{x_0}^x\nabla\mathrm{S}_Lf(y)\cdot \dd y\bigg\Vert_{\mathrm{L}^\infty(X)},\Vert\nabla\mathrm{S}_Lf\Vert_{\mathrm{L}^\infty(X)},\Vert\nabla^2\mathrm{S}_Lf\Vert_{\mathrm{L}^\infty(X)}\bigg\rcurl\nonumber\\
        &\leqslant L^{-1}\mathrm{diam}(X)\max_{m=1,2}\Vert\mathrm{S}_L\nabla^m\hspace{-1pt}f\Vert_{\mathrm{L}^\infty(X)}\nonumber\\
        &\leqslant 8L^{-1}\Vert f\Vert_{\mathcal{C}^2(\bar{X}^L)}.
    \end{align}
    Therefore
    \begin{align}
        \big\vert\langle\mathrm{D}^n\vert_\phi K_j(X,\mathrm{S}_L\delta(\cdot)+\psi),f_1,\ldots,f_n\rangle\big\vert &= \Bigg\vert\frac{\partial^n}{\partial t_1\ldots\partial t_n}\bigg\vert_{t_1=\ldots=t_n=0}K_j\bigg(X,\mathrm{S}_L\delta\phi+\psi+\sum_{i=1}^nt_i\mathrm{S}_L\delta f_i\bigg)\Bigg\vert\nonumber\\
        &\leqslant (8L^{-1})^n\Vert K_j\Vert_{n;\langle X,\mathrm{S}_L\phi+\psi\rangle}\Vert f_1\Vert_{\mathcal{C}^2(\bar{X}^L)}\ldots\Vert f_n\Vert_{\mathcal{C}^2(\bar{X}^L)}.
    \end{align}
    This indeed implies the desired inequality. For (\ref{Rem0}), we break up derivatives into two cases. For terms with $n\geqslant m+1$, $\mathrm{D}^n\mathrm{Rem}_m F = \mathrm{D}^n F$ since the remainder is obtained by subtracting a polynomial. On the other hand for $n\leqslant m$,
    \begin{align}
        \big\vert\langle\mathrm{D}^n\vert_{\phi}\mathrm{Rem}_mK_j(X),f_1,\ldots,f_n)\rangle\big\vert&\leqslant\int_0^1\frac{(1-t)^{m-n}}{(m-n)!}\big\vert\langle\mathrm{D}^{m+1}\vert_{t\phi}K_j(X),\phi,\ldots,\phi,f_1,\ldots,f_n\rangle\big\vert\mathrm{d}t\nonumber\\
        &\leqslant\frac{1}{(m+1-n)!}\sup_{t\in[0,1]}\Vert K_j\Vert_{m+1;\langle X,\phi\rangle}\Vert\phi\Vert_{\mathcal{C}^2(X)}^{m+1-n}\prod_{i=1}^n\Vert f_i\Vert_{\mathcal{C}^2(X)}.
    \end{align}
    Then
    \begin{align}
        \sum_{n=0}^{m+1}\frac{\mathfrak{h}^n}{n!}\Vert\mathrm{Rem}_mK_j\Vert_{n;\langle X,\phi\rangle} &\leqslant \mathfrak{h}^{m+1}\sup_{t\in[0,1]}\Vert K_j\Vert_{m+1;\langle X,t\phi\rangle}\sum_{n=0}^{m+1}\frac{(\mathfrak{h}^{-1}\Vert\phi\Vert_{\mathcal{C}^2(X)})^{m+1-n}}{n!(m+1-n)!}\nonumber\\
        &\leqslant\frac{\mathfrak{h}^{m+1}}{(m+1)!}\big(1+\mathfrak{h}^{-1}\Vert\phi\Vert_{\mathcal{C}^2(X)}\big)^{\! m+1\,}\sup_{t\in[0,1]}\Vert K_j\Vert_{m+1;\langle X,t\phi\rangle}
    \end{align}
    and the result follows.
\end{proof}

\begin{lem}
    Let $\vert q\vert\geqslant 1$, $X\in\mathcal{S}$. Also fix $\mathfrak{h}_* \gtrsim_1 1$ and take $\mathfrak{h} \geqslant\mathfrak{h}_*$. Then
    \begin{equation}
        \big\Vert\mathbb{E}\big[k_j^{(q)}(\cdot,\mathrm{S}_L(\cdot)+\Psi_{\hspace{-1pt}\varepsilon;j})\big]\big\Vert_{\mathfrak{h},\kappa;\langle X\rangle;\mathcal{C}^2(\bar{X}^L)}\leqslant e^{\sqrtbetae\vert q\vert\mathfrak{h}-(\vert q\vert-1/2)\beta\Gamma_{\hspace{-1pt}\varepsilon;j}}2^{\vert X\vert}\Vert K_j\Vert_{\mathfrak{h},\kappa;\langle X\rangle}.\label{chargedScaling}
    \end{equation}
\end{lem}

\begin{proof}
    We only consider $q > 0$. Take $f\in\mathcal{S}$ positive with integral $1$. We do an analytic continuation argument: if $z\in\mathbb{R}$ and $x_0\in X$,
    \begin{align}
        \mathbb{E}\big[k_j^{(q)}(X,\phi+\Psi_{\hspace{-1pt}\varepsilon;j})\big] &= \int k_j^{(q)}(X,\phi+\psi)\,\gamma[\Gamma_{\hspace{-1pt}\varepsilon;j}](\dd\psi)\nonumber\\
        &= e^{-\frac{1}{2}z^2\langle \Gamma_{\hspace{-1pt}\varepsilon;j}f,f\rangle}\int e^{-z\langle f,\psi\rangle}k_j^{(q)}(X,\phi+\psi+z\Gamma_{\hspace{-1pt}\varepsilon;j}f)\,\gamma[\Gamma_{\hspace{-1pt}\varepsilon;j}](\dd\psi)\\
        &=e^{-\frac{1}{2}z^2\langle \Gamma_{\hspace{-1pt}\varepsilon;j}f,f\rangle+i\beta qz(\Gamma_{\hspace{-1pt}\varepsilon;j}f)(x_0)}\int e^{-z\langle f,\psi\rangle}\,{^\mathbb{C\hspace{-1pt}}}k^{(q)}_j(X,\phi+\psi+z\delta(\Gamma_{\hspace{-1pt}\varepsilon;j}f))\,\gamma[\Gamma_{\hspace{-1pt}\varepsilon;j}](\dd\psi).\nonumber
    \end{align}
    By the identity principle, equality holds for all $z$ on some strip around $\mathbb{R}$. We note that \begin{equation}
        \sqrtbeta\Vert\delta(\Gamma_{\hspace{-1pt}\varepsilon;j}f)\Vert_{\mathcal{C}^2(X)}\lesssim_1\Vert\nabla\Gamma_{\hspace{-1pt}\varepsilon;j}f\Vert_{\mathcal{C}^1(X)}\lesssim_1\max_{\vert\mu\vert=1,2}\Vert\partial^\mu\Gamma_{\hspace{-1pt}\varepsilon;j}\Vert_{\mathrm{L}^\infty(\mathbb{T}^2)}\lesssim_1 1,
    \end{equation}
    so because $\mathfrak{h}_*\gtrsim_1 1$, the strip is big enough that we can take $z = i\sqrtbeta$. In the limit as $f$ approaches $\delta_{x_0}$,
    \begin{equation}
        \mathbb{E}\big[k_j^{(q)}(X,\phi+\Psi_{\hspace{-1pt}\varepsilon;j})\big] = e^{-(q-1/2)\beta\Gamma_{\hspace{-1pt}\varepsilon;j}(0)}\int e^{-i\sqrtbetae\psi(x_0)}\,{^\mathbb{C\hspace{-1pt}}}k^{(q)}_j(X,\phi+\psi+i\sqrtbeta\delta(\Gamma_{\hspace{-1pt}\varepsilon;j}\delta_{x_0}))\,\gamma[\Gamma_{\hspace{-1pt}\varepsilon;j}](\dd\psi).
    \end{equation}
    It follows that
    \begin{align}
        \big\Vert\mathbb{E}\big[k^{(q)}_j&(\cdot,\mathrm{S}_L(\cdot)+\Psi_{\hspace{-1pt}\varepsilon;j})\big]\big\Vert_{\mathfrak{h};\langle X,\phi\rangle;\mathcal{C}^2(\bar{X}^L)} = \big\Vert e^{i\sqrtbetae q\mathrm{S}_L(\cdot)(x_0)}\mathbb{E}\big[k^{(q)}_j(\cdot,\mathrm{S}_L\delta(\cdot)+\Psi_{\hspace{-1pt}\varepsilon;j})\big]\big\Vert_{\mathfrak{h};\langle X,\phi\rangle;\mathcal{C}^2(\bar{X}^L)}\nonumber\\
        &\leqslant e^{\sqrtbetae q\mathfrak{h}}\big\Vert\mathbb{E}\big[k_j^{(q)}(\cdot,\cdot+\Psi_{\hspace{-1pt}\varepsilon;j})\big]\big\Vert_{8\mathfrak{h}/L;\langle X,\mathrm{S}_L\delta\phi\rangle}\nonumber\\
        &\leqslant e^{\sqrtbetae q\mathfrak{h}-(q-1/2)\beta\Gamma_{\hspace{-1pt}\varepsilon;j}(0)}\big\Vert\mathbb{E}\big[e^{-i\sqrtbetae\Psi_{\hspace{-1pt}\varepsilon;j}(x_0)}\,{^\mathbb{C\hspace{-1pt}}}k^{(q)}_j(\cdot,\cdot+\Psi_{\hspace{-1pt}\varepsilon;j}+i\sqrtbeta\delta(\Gamma_{\hspace{-1pt}\varepsilon;j}\delta_{x_0}))\big]\big\Vert_{8\mathfrak{h}/L;\langle X,\mathrm{S}_L\delta\phi\rangle}\nonumber\\
        &\leqslant e^{\sqrtbetae q\mathfrak{h}-(q-1/2)\beta\Gamma_{\hspace{-1pt}\varepsilon;j}(0)}\mathbb{E}\big[\Vert{^\mathbb{C\hspace{-1pt}}}k^{(q)}_j\big\Vert_{8\mathfrak{h}/L;\langle X,\mathrm{S}_L\delta\phi+\Psi_{\hspace{-1pt}\varepsilon;j}+i\sqrtbetae \delta(\Gamma_{\hspace{-1pt}\varepsilon;j}\delta_{x_0})}\big]\nonumber\\
        &\leqslant e^{\sqrtbetae q\mathfrak{h}-(q-1/2)\beta\Gamma_{\hspace{-1pt}\varepsilon;j}(0)}\mathbb{E}\big[\Vert{^\mathbb{C\hspace{-1pt}}}k^{(q)}_j\big\Vert_{8\mathfrak{h}/L+2\sqrtbetae\Vert\delta(\Gamma_{\hspace{-1pt}\varepsilon;j}\delta_{x_0})\Vert_{\mathcal{C}^2(X)};\langle X,\mathrm{S}_L\delta\phi+\Psi_{\hspace{-1pt}\varepsilon;j}\rangle}\big]\nonumber\\
        &\leqslant e^{\sqrtbetae q\mathfrak{h}-(q-1/2)\beta\Gamma_{\hspace{-1pt}\varepsilon;j}(0)}\mathbb{E}\big[G_\kappa(X,\mathrm{S}_L\phi+\Psi_{\hspace{-1pt}\varepsilon;j})\big]\Vert k^{(q)}_j\Vert_{\mathfrak{h},\kappa;\langle X\rangle}
    \end{align}
    where the last inequality holds due to $L,\mathfrak{h}_*\gtrsim_1 1$. The proof concludes after applying (\ref{supermpr0}).
\end{proof}

\begin{lem}
    Take $\kappa \geqslant \kappa_*/2$. Suppose $K_j$ is real and even in $\phi$. Then
    \begin{equation}
        \big\Vert\mathrm{Rem}_0\mathbb{E}\big[k^{(0)}_j(\cdot,\mathrm{S}_L(\cdot)+\Psi_{\hspace{-1pt}\varepsilon;j})\big]\big\Vert_{\mathfrak{h},\kappa;\langle X\rangle;\mathcal{C}^2(\bar{X}^L)}\lesssim_1 2^{\vert X\vert}\kappa L^{-2}\Vert K\Vert_{\mathfrak{h},\kappa;\langle X\rangle}.\label{Rem0scaling}
    \end{equation}
\end{lem}

\begin{proof}
    By assumption, $\mathrm{Rem}_0\mathbb{E}\big[k^{(0)}_j(X,\mathrm{S}_L(\cdot)+\Psi_{\hspace{-1pt}\varepsilon;j})\big] = \mathrm{Rem}_1\mathbb{E}\big[k^{(0)}_j(X,\mathrm{S}_L(\cdot)+\Psi_{\hspace{-1pt}\varepsilon;j})\big]$. From (\ref{Rem0},\ref{vanishingAtPoint}),
    \begin{align}
        \big\Vert\mathrm{Rem}_1\mathbb{E}\big[k^{(0)}_j(\cdot,\mathrm{S}_L(\cdot)+\psi)\big]\big\Vert_{\mathfrak{h};\langle X,\phi\rangle;\mathcal{C}^2(\bar{X}^L)}&\leqslant \Vert\mathrm{Rem}_1k^{(0)}_j\Vert_{8\mathfrak{h}/L;\langle X,\mathrm{S}_L\delta\phi+\psi\rangle}\nonumber\\
        &\leqslant\big(1+L\Vert\mathrm{S}_L\delta\phi\Vert_{\mathcal{C}^2(X)}\big)^{\!2\,}\sup_{t\in[0,1]}\sum_{n=2}^\infty\frac{(8\mathfrak{h}/L)^n}{n!}\Vert k^{(0)}_j\Vert_{n;\langle X,t\mathrm{S}_L\delta\phi+\psi\rangle}\nonumber\\
        &\lesssim_1\big(1+\Vert\nabla\phi\Vert_{\mathcal{C}^1(X)}\big)^{\!2\,}L^{-2}\sup_{t\in[0,1]}\Vert k^{(0)}_j\Vert_{\mathfrak{h};\langle X,t\mathrm{S}_L\delta\phi+\psi\rangle}\\
        &\lesssim_1 L^{-2}\Vert K_j\Vert_{\mathfrak{h},\kappa;\langle X\rangle}\big(1+\Vert\nabla\phi\Vert_{\mathcal{C}^1(X)}\big)^{\!2\,}\sup_{t\in[0,1]}G_\kappa(X,t\mathrm{S}_L\phi+\psi).\nonumber
    \end{align}
    The result is immediate by (\ref{smallXthing}).
\end{proof}

\section{Proofs of some propositions}\label{AppProofs}

These proofs closely follow those in \cite{Falco1}. In some respects they are simpler because we are in the subcritical regime, while in others they are more complicated because we have to deal with the fluctuation fields changing a bit at all scales when the UV cutoff is lowered.

\subsection{Proof of Proposition \ref{MainRgEstimates}} 

First,
\begin{equation}
    \langle\mathrm{D}^n\vert_\phi V_j(X),f_1,\ldots,f_n\rangle = \beta^{n/2}\int_X\zeta_j(x)\cos^{(n)}\!\!\big(\sqrtbeta\phi(x)\big)f_1(x)\ldots f_n(x)\dd x,
\end{equation}
so
\begin{align}
    \Vert V_j-\dot{V}_j\Vert_{\mathfrak{h};\langle X,\phi\rangle} &\lesssim_1 \vert X\vert\Vert\zeta_{\varepsilon;j}-\dot\zeta_{\dot\varepsilon,j}\Vert_{\wp}\nonumber\\
    &\lesssim_1\vert X\vert\big(\Vert\zeta_{0;j}-\dot\zeta_{0,j}\Vert_{\wp}+L^{-j}\vert\varepsilon-\dot{\varepsilon}\vert \Vert\zeta_{0;j},\dot\zeta_{0,j}\Vert_{\wp}\big).
\end{align}
Furthermore,
\begin{align}
    \Vert e^{V_j}-e^{\dot{V}_j}\Vert_{\mathfrak{h};\langle X,\phi\rangle}&\leqslant\sum_{n=1}^\infty\frac{1}{n!}\Vert V_j^n-\dot{V}^n_j\Vert_{\mathfrak{h};\langle X,\phi\rangle}\nonumber\\
    &\leqslant\Vert V_j-\dot{V}_j\Vert_{\mathfrak{h};\langle X,\phi\rangle}\sum_{n=1}^\infty\frac{1}{n!}\sum_{k=0}^{n-1}\Vert V_j\Vert^{n-1-k}_{\mathfrak{h};\langle X,\phi\rangle}\Vert\dot{V}_j\Vert^k_{\mathfrak{h};\langle X,\phi\rangle}\nonumber\\
    &\leqslant e^{\Vert V_j,\dot{V}_j\Vert^n_{\mathfrak{h};\langle X,\phi\rangle}}\Vert V_j-\dot{V}_j\Vert_{\mathfrak{h};\langle X,\phi\rangle}\nonumber\\
    &\lesssim_1 2^{\vert X\vert}\big(\Vert\zeta_{0;j}-\dot\zeta_{0,j}\Vert_{\wp}+L^{-j}\vert\varepsilon-\dot{\varepsilon}\vert \Vert\zeta_{0;j},\dot\zeta_{0,j}\Vert_{\wp}\big).
\end{align}
For the extraction,
\begin{align}
    \vert Q_j&(X)-\dot{Q}_j(X)\vert\nonumber\\
    &=\big\vert\mathbb{E}\big[k^{(0)}_j(X,\Psi_{\hspace{-1pt}j})-\dot{k}^{(0)}_j(X,\dot{\Psi}_{\hspace{-1pt}j})\big]\big\vert\nonumber\\
    &\leqslant\mathbb{E}\big[\big\vert k^{(0)}_j-\dot{k}^{(0)}_j\big\vert(X,\Psi_{\hspace{-1pt}j})\big]+\mathbb{E}\big[\big\vert\dot{k}^{(0)}_j(X,\Psi_{\hspace{-1pt}j})-\dot{k}^{(0)}_j(X,\dot{\Psi}_{\hspace{-1pt}j})\big\vert\big]\nonumber\\
    &\leqslant\Vert K_j-\dot{K}_j\Vert_{\mathfrak{h},\kappa;\langle X\rangle}\mathbb{E}\big[G_\kappa(X,\Psi_{\hspace{-1pt}j})\big]+\mathbb{E}\bigg[\int_0^1\big\vert\langle \mathrm{D}\vert_{(1-t){\Psi}_{\hspace{-1pt}j}+t\dot{\Psi}_{\hspace{-1pt}j}}\dot{k}^{(0)}_j(X),{\Psi}_{\hspace{-1pt}j}-\dot{\Psi}_{\hspace{-1pt}j}\rangle\big\vert\dd t\bigg]\nonumber\\
    &\lesssim_1 \Vert K_j-\dot{K}_j\Vert_{\mathfrak{h},\kappa;\langle X\rangle} + \Vert\dot{K}_j\Vert_{\mathfrak{h},\kappa;\langle X\rangle}\mathbb{E}\Big[\Vert{\Psi}_{\hspace{-1pt}j}-\dot{\Psi}_{\hspace{-1pt}j}\Vert_{\mathcal{C}^2(X)}\sup_{t\in[0,1]}G_\kappa(X,(1-t){\Psi}_{\hspace{-1pt}j}+t\dot{\Psi}_{\hspace{-1pt}j})\Big]\nonumber\\
    &\leqslant\Vert K_j-\dot{K}_j\Vert_{\mathfrak{h},\kappa;\langle X\rangle} + \Vert K_j,\dot{K}_j\Vert_{\mathfrak{h},\kappa;\langle X\rangle}\mathbb{E}\big[\Vert{\Psi}_{\hspace{-1pt}j}-\dot{\Psi}_{\hspace{-1pt}j}\Vert_{\mathcal{C}^2(X)}\big]^{\!1/2}\mathbb{E}\Big[\sup_{t\in[0,1]}G_{\kappa/2}(X,(1-t){\Psi}_{\hspace{-1pt}j}+t\dot{\Psi}_{\hspace{-1pt}j})\Big]^{\!1/2}\nonumber\\
    &\lesssim_1 \Vert K_j-\dot{K}_j\Vert_{\mathfrak{h},\kappa;\langle X\rangle} + L^{-j}\vert\varepsilon-\dot\varepsilon\vert\Vert K_j,\dot{K}_j\Vert_{\mathfrak{h},\kappa;\langle X\rangle}.
\end{align}
This also immediately implies the estimate for $\delta\mathcal{E}_j$. For $J_j^{\mathfrak{B}\rightarrow X}-\dot{J}_j^{\mathfrak{B}\rightarrow X}$,
\begin{equation}
    \vert J_j^{\mathfrak{B}\rightarrow X}-\dot{J}_j^{\mathfrak{B}\rightarrow X}\vert \leqslant \sum_{\varnothing\neq\mathcal{A}\subset\mathcal{C}(X)}\prod_{Y\in\mathcal{A}}\vert J_j(Y,\mathfrak{B}_Y)-\dot{J}_j(Y,\mathfrak{B}_Y)\vert\prod_{Z\in\mathcal{A}^c}\vert\dot{J}_j(Z,\mathfrak{B}_Z)\vert
\end{equation}
and so it will be enough to control $J_j(Y,B)-\dot{J}_j(Y,B)$, which is easily done from its definition:
\begin{align}
    \vert J_j(Y,B)-\dot{J}_j(Y,B)\vert&\leqslant \bigg(\sum_{D\in\mathcal{B}_j(LB),X\in\mathcal{S}_j}^{D\subset X,\bar{X} = Y}+\,\mathbb{I}_{Y=B}\sum_{Y'\in\mathcal{S}_{j+1}}\sum_{D\in\mathcal{B}_j(LB),X\in\mathcal{S}_j}^{D\subset X,\bar{X} = Y'}\bigg)\vert Q_j(X)-\dot{Q}_j(X)\vert\nonumber\\
    &\lesssim_{1} A^{-1}L^2\big(\Vert K_j-\dot{K}_j\Vert_{\mathfrak{h},\kappa,A} + L^{-j}\vert\varepsilon-\dot\varepsilon\vert\Vert K_j,\dot{K}_j\Vert_{\mathfrak{h},\kappa,A}\big)\nonumber\\
    &\leqslant \mathcal{O}_A A^{-2\vert Y^*\vert}\big(\Vert K_j-\dot{K}_j\Vert_{\mathfrak{h},\kappa,A} + L^{-j}\vert\varepsilon-\dot\varepsilon\vert\Vert K_j,\dot{K}_j\Vert_{\mathfrak{h},\kappa,A}\big).
\end{align}
In going to the last line, we used that $\vert Y^*\vert\leqslant 49\vert Y\vert< 200$ because $Y\in\mathcal{S}$. Writing the inequality in this specific form comes in handy later.

Now let's consider $R_j^{(\cdot)}$. Expanding the definition,
\begin{align}
    \Vert R&^{(\cdot)}_j(\cdot,\psi)-\dot{R}_j^{(\cdot)}(\cdot,\dot{\psi})\Vert_{\mathfrak{h};\langle X,\phi\rangle}\nonumber\\
    &=\bigg\Vert\prod_{Y\in\mathcal{C}(X)}\bigg(\tilde{K}_j(Y,\cdot,\psi)-\sum_{B\in\mathcal{B}_{j+1}(Y)}J_j(Y,B)\bigg)-\prod_{Y\in\mathcal{C}(X)}\bigg(\dot{\tilde{K}}_j(Y,\cdot,\dot\psi)-\sum_{B\in\mathcal{B}_{j+1}(Y)}\dot J_j(Y,B)\bigg)\bigg\Vert_{\mathfrak{h};\langle\phi\rangle;\mathcal{C}^2(X)}\nonumber\\
    &\leqslant \sum_{\mathcal{A}\subset\mathcal{C}(X)}\bigg\Vert\prod_{Y\in\mathcal{A}}\tilde{K}_j(Y,\cdot,\psi)\sum_{\mathfrak{B}\rightarrow\bigcup\mathcal{A}^c}J_j^{\mathfrak{B}\rightarrow\bigcup\mathcal{A}^c}-\prod_{Y\in\mathcal{A}}\dot{\tilde{K}}_j(Y,\cdot,\dot\psi)\sum_{\mathfrak{B}\rightarrow\bigcup\mathcal{A}^c}\dot{J}_j^{\mathfrak{B}\rightarrow\bigcup\mathcal{A}^c}\bigg\Vert_{\mathfrak{h};\langle\phi\rangle;\mathcal{C}^2(X)}\nonumber\\
    &\leqslant\sum_{\varnothing\neq\mathcal{A}\subset\mathcal{C}(X)}\big\Vert\hat{K}_j^{\bigcup\mathcal{A}}(\cdot,\psi)-\dot{\hat{K}}_j^{\bigcup\mathcal{A}}(\cdot,\dot\psi)\big\Vert_{\mathfrak{h};\langle\phi\rangle;\mathcal{C}^2(\bigcup\mathcal{A})}\sum_{\mathfrak{B}\rightarrow\bigcup\mathcal{A}^c}\vert J_j^{\mathfrak{B}\rightarrow\bigcup\mathcal{A}^c}\vert+\nonumber\\
    &\hspace{45pt}\sum_{\mathcal{A}\subsetneq\mathcal{C}(X)}\big\Vert\dot{\hat{K}}_j^{\bigcup\mathcal{A}}(\cdot,\dot\psi)\big\Vert_{\mathfrak{h};\langle\phi\rangle;\mathcal{C}^2(\bigcup\mathcal{A})}\sum_{\mathfrak{B}\rightarrow\bigcup\mathcal{A}^c}\vert J_j^{\mathfrak{B}\rightarrow\bigcup\mathcal{A}^c}-\dot J_j^{\mathfrak{B}\rightarrow\bigcup\mathcal{A}^c}\vert.
\end{align}
Here we have provisionally defined
\begin{equation}
    \hat{K}_j^{X}(\phi,\psi) = \prod_{Y\in\mathcal{C}(X)}\tilde{K}_j(Y,\phi,\psi).
\end{equation}
The main issue is controlling the norm of this object. We need to get a handle on $\tilde{K}_j$ itself (first with $\psi = \dot\psi$):
\begin{align}
    \Vert\tilde{K}_j(\cdot&,\cdot,\psi)-\dot{\tilde{K}}_j(\cdot,\cdot,\psi)\Vert_{\mathfrak{h};\langle X,\phi\rangle}\nonumber\\
    &\leqslant \sum_{Y\in\mathcal{P}_j}^{\bar{Y}^L = X}\bigg\Vert e^{V_j(L X-Y,\cdot)}\prod_{X'\in\mathcal{C}(Y)}K_j(X',\cdot)-e^{\dot V_j(L X-Y,\cdot)}\prod_{X'\in\mathcal{C}(Y)}\dot K_j(X',\cdot)\bigg\Vert_{\mathfrak{h};\langle\mathrm{S}_L\phi+\psi\rangle;\mathcal
    C^2(L X)}\nonumber\\
    &\leqslant\sum_{Y\in\mathcal{P}_j}^{\bar{Y}^L = X}\Vert e^{V_j}-e^{\dot{V}_j}\Vert_{\mathfrak{h};\langle LX - Y,\mathrm{S}_L\phi+\psi\rangle}\prod_{X'\in\mathcal{C}(Y)}\Vert K_j\Vert_{\mathfrak{h};\langle X',\mathrm{S}_L\phi+\psi\rangle}+\nonumber\\
    &\hspace{20pt}\sum_{Y\in\mathcal{P}_j}^{\bar{Y}^L = X}\Vert e^{\dot{V}_j}\Vert_{\mathfrak{h};\langle LX - Y,\mathrm{S}_L\phi+\psi\rangle}\bigg\Vert \prod_{X'\in\mathcal{C}(Y)}K_j(X',\cdot)-\prod_{X'\in\mathcal{C}(Y)}\dot K_j(X',\cdot)\bigg\Vert_{\mathfrak{h};\langle\mathrm{S}_L\phi+\psi\rangle;\mathcal
    C^2(Y)}\nonumber\\
    &\leqslant\mathcal{O}_L^{\vert X\vert}G_\kappa(L X,\mathrm{S}_L\phi+\psi)\sum_{Y\in\mathcal{P}_j}^{\bar{Y}^L=X}A^{-\vert Y\vert}\big(\mathcal O_1\Vert K_j,\dot K_j\Vert_{\mathfrak{h},\kappa,A}\big)^{\!\vert\mathcal{C}(Y)\vert-1}\cdot\nonumber\\
    &\hspace{40pt}\big(\Vert K_j,\dot K_j\Vert_{\mathfrak{h},\kappa,A}\big(\Vert\zeta_{0;j}-\dot\zeta_{0,j}\Vert_{\wp}+L^{-j}\vert\varepsilon-\dot{\varepsilon}\vert \Vert\zeta_{0;j},\dot\zeta_{0,j}\Vert_{\wp}\big) + \Vert K_j-\dot K_j\Vert_{\mathfrak{h},\kappa,A}\big)\nonumber\\
    &\leqslant\mathcal{O}_A A^{-(1+\eta/3)\vert X\vert}G_\kappa(L X,\mathrm{S}_L\phi+\psi)\cdot\\
    &\hspace{40pt}\big(\Vert K_j,\dot K_j\Vert_{\mathfrak{h},\kappa,A}\big(\Vert\zeta_{0;j}-\dot\zeta_{0,j}\Vert_{\wp}+L^{-j}\vert\varepsilon-\dot{\varepsilon}\vert \Vert\zeta_{0;j},\dot\zeta_{0,j}\Vert_{\wp}\big) + \Vert K_j-\dot K_j\Vert_{\mathfrak{h},\kappa,A}\big)\nonumber.
\end{align}
In going to the last line, we used (\ref{disconnectedGrowth}) to get $A^{-\vert Y\vert}\leqslant A^{-(1+\eta/2)\vert X\vert + 9\vert\mathcal{C}(Y)\vert}$. Now for $\hat{K}_j$,
\begin{equation}
    \Vert \hat{K}_j^{(\cdot)}(\cdot,\psi) - \dot{\hat{K}}_j^{(\cdot)}(\cdot,\dot\psi)\Vert_{\mathfrak{h};\langle X,\phi\rangle} \leqslant \Vert \hat{K}_j^{(\cdot)}(\cdot,\psi) - {\hat{K}}_j^{(\cdot)}(\cdot,\dot\psi)\Vert_{\mathfrak{h};\langle X,\phi\rangle}+\Vert {\hat{K}}_j^{(\cdot)}(\cdot,\dot\psi) - \dot{\hat{K}}_j^{(\cdot)}(\cdot,\dot\psi)\Vert_{\mathfrak{h};\langle X,\phi\rangle}.
\end{equation}
The first term is more delicate. We use the fact that $\hat{K}_j$ is really only a function of $\mathrm{S}_L\phi+\psi$ to get
\begin{align}
   \hat{K}_j^X&(\phi,\psi)-\hat{K}_j^X(\phi,\dot\psi)\\
   &=\sum_{Y_0\in\mathcal{C}(X)}\int_0^1 \prod_{Y\in\mathcal{C}(X)-{Y_0}}\tilde{K}_j(Y,\phi,(1-t)\psi+t\dot\psi)\langle \mathrm{D}\vert_{\phi + t\mathrm{S}_{L^{-1}}(\dot\psi-\psi)}\tilde{K}_j(Y_0,\cdot,\dot\psi),\mathrm{S}_{L^{-1}}(\dot\psi-\psi)\rangle\,\dd t,\nonumber
\end{align}
so
\begin{align}
    &\Vert\hat{K}_j^{(\cdot)}(\cdot,\psi)-\hat{K}_j^{(\cdot)}(\cdot,\dot\psi)\Vert_{\mathfrak{h};\langle X,\phi\rangle}\\
    &\leqslant L^2\Vert\psi-\dot\psi\Vert_{\mathcal{C}^2(LX)}\sum_{Y_0\in\mathcal{C}(X)}\sup_{t\in[0,1]}\prod_{Y\in\mathcal{C}(X)-\lcurl Y_0\rcurl}\Vert \tilde{K}_j(\cdot,\cdot,(1-t)\psi+t\dot\psi)\Vert_{\mathfrak{h};\langle Y,\phi\rangle}\Vert \tilde{K}_j(\cdot,\cdot,\dot\psi)\Vert_{2\mathfrak{h};\langle Y_0,\phi + t\mathrm{S}_{L^{-1}}(\psi-\dot\psi)\rangle}\nonumber\\
    &\leqslant\mathcal{O}_A^{\vert\mathcal{C}(X)\vert} A^{-(1+\eta/3)\vert X\vert}\Vert\psi-\dot\psi\Vert_{\mathcal{C}^2(LX)}G_\kappa(L X,\mathrm{S}_L\phi+(1-t)\psi+t\dot \psi)\Vert K_j,\dot K_j\Vert_{\mathfrak{h},\kappa,A}^{\vert\mathcal{C}(X)\vert-1}\Vert K_j,\dot K_j\Vert_{2\mathfrak{h},\kappa,A}\nonumber.
\end{align}
For the second term,
\begin{align}
    \Vert {\hat{K}}_j^{(\cdot)}(\cdot,\dot\psi)& - \dot{\hat{K}}_j^{(\cdot)}(\cdot,\dot\psi)\Vert_{\mathfrak{h};\langle X,\phi\rangle}\leqslant \sum_{\varnothing\neq\mathcal{A}\subset\mathcal{C}(X)}\prod_{Y\in\mathcal{A}}\Vert \tilde{K}_j(Y,\cdot,\dot\psi)-\dot{\tilde{K}}_j(Y,\cdot,\dot\psi)\Vert_{\mathfrak{h};\langle Y,\phi\rangle}\prod_{Y\in\mathcal{A}^c}\Vert\dot{\tilde{K}}_j(Y,\cdot,\dot\psi)\Vert_{\mathfrak{h};\langle Y,\phi\rangle}\nonumber\\
    &\leqslant\mathcal{O}_A^{\vert\mathcal{C}(X)\vert} A^{-(1+\eta/3)\vert X\vert}G_\kappa(L X,\mathrm{S}_L\phi+\dot\psi)\Vert K_j,\dot K_j\Vert_{\mathfrak{h},\kappa,A}^{\vert\mathcal{C}(X)\vert-1}\cdot\\
    &\hspace{30pt}\big(\Vert K_j,\dot K_j\Vert_{\mathfrak{h},\kappa,A}\big(\Vert\zeta_{0;j}-\dot\zeta_{0,j}\Vert_{\wp}+L^{-j}\vert\varepsilon-\dot{\varepsilon}\vert \Vert\zeta_{0;j},\dot\zeta_{0,j}\Vert_{\wp}\big) + \Vert K_j-\dot K_j\Vert_{\mathfrak{h},\kappa,A}\big).\nonumber
\end{align}
By combining all this (we regroup some terms to isolate the ones which vanish when $\varepsilon = \dot\varepsilon$)
\begin{align}
    \Vert \hat{K}_j^{(\cdot)}(\cdot,\psi)& - \dot{\hat{K}}_j^{(\cdot)}(\cdot,\dot\psi)\Vert_{\mathfrak{h};\langle X,\phi\rangle}\nonumber\\
    &\leqslant\mathcal{O}_A^{\vert\mathcal{C}(X)\vert} A^{-(1+\eta/3)\vert X\vert}\Vert K_j,\dot K_j\Vert_{\mathfrak{h},\kappa,A}^{\vert\mathcal{C}(X)\vert-1}\sup_{t\in[0,1]}G_\kappa(L X,\mathrm{S}_L\phi+(1-t)\psi+t\dot\psi)\cdot\nonumber\\
    &\hspace{20pt}\Big(\big(\Vert K_j,\dot K_j\Vert_{\mathfrak{h},\kappa,A}\Vert\zeta_{0;j}-\dot\zeta_{0,j}\Vert_{\wp} + \Vert K_j-\dot K_j\Vert_{\mathfrak{h},\kappa,A}\big)\nonumber\\
    &\hspace{30pt}+\Vert K_j,\dot K_j\Vert_{2\mathfrak{h},\kappa,A}\big(\Vert\psi-\dot\psi\Vert_{\mathcal{C}^2(LX)}+L^{-j}\vert\varepsilon-\dot{\varepsilon}\vert \Vert\zeta_{0;j},\dot\zeta_{0,j}\Vert_{\wp}\big)\Big).
\end{align}
From here, the estimate on $R^{(\cdot)}_j$ follows immediately. Finally, we consider $P^{(\cdot)}_j$. For the case $\psi\neq\dot\psi$, a simple computation shows
\begin{align}
    P_j^X&(\phi,\psi)-P_j^X(\phi,\dot\psi)\\
    &=\sum_{B\in\mathcal{B}_{j+1}(X)}\int_0^1 P_j^{X-\mathrm{Int}(B)}(\phi,(1-t)\psi+t\dot\psi)e^{V_j(L B,\mathrm{S}_L\phi+(1-t)\psi+t\dot\psi)}\frac{\dd}{\dd t}V_j(L B,\mathrm{S}_L\phi+(1-t)\psi+t\dot\psi)\dd t\nonumber.
\end{align}
Taking norms,
\begin{align}
    \Vert P_j^X&(\cdot,\psi)-P_j^X(\cdot,\dot\psi)\Vert_{\mathfrak{h};\langle\phi\rangle;\mathcal{C}^2(X)}\nonumber\\
    &\leqslant \mathcal{O}_L\Vert\psi-\dot\psi\Vert_{\mathcal{C}^2(L X)}\big(\Vert\zeta_{0;j}-\dot\zeta_{0;j}\Vert_{\wp}+L^{-j}\vert\varepsilon-\dot\varepsilon\vert\Vert\zeta_{0;j},\dot\zeta_{0;j}\Vert_{\wp}\big)\cdot\\
    &\hspace{40pt}\sum_{B\in\mathcal{B}_{j+1}(X)}\Vert P_j^{(\cdot)}(\cdot,(1-t)\psi+t\dot\psi)\Vert_{\mathfrak{h};\langle X-\mathrm{Int}(B), \phi\rangle}.
\end{align}
Then for the difference,
\begin{align}
    &\Vert P_j^{(\cdot)}(\cdot,\psi)-\dot{P}_j^{(\cdot)}(\cdot,\psi)\Vert_{\mathfrak{h};\langle X, \phi\rangle}\nonumber\\
    &\leqslant\bigg\Vert \prod_{B\in\mathcal{B}_{j+1}(X)}\big(e^{V_j(L B,\mathrm{S}_L(\cdot)+\psi)}-e^{V_{j+1}(B,\cdot)+\delta\mathcal{E}_j(LB)}\big)\nonumber\\
    &\hspace{120pt}-\prod_{B\in\mathcal{B}_{j+1}(X)}\big(e^{\dot V_j(L B,\mathrm{S}_L(\cdot)+\psi)}-e^{\dot V_{j+1}(B,\cdot)+\delta\dot{\mathcal{E}}_j(L B)}\big)\bigg\Vert_{\mathfrak{h};\langle\phi\rangle;\mathcal{C}^2(X)}
    \nonumber\\
    &\leqslant\sum_{\varnothing\neq\mathcal{A}\subset\mathcal{B}_{j+1}(X)}\prod_{B\in\mathcal{A}^c}\Vert e^{\dot V_j(L B,\mathrm{S}_L(\cdot)+\psi)}-e^{\dot V_{j+1}(B,\cdot)+\delta\dot{\mathcal{E}}_j(L B)}\Vert_{\mathfrak{h};\langle\phi\rangle;\mathcal{C}^2(B)}\nonumber\cdot\\
    &\hspace{40pt}\prod_{B\in\mathcal{A}}\Vert e^{V_j(L B,\mathrm{S}_L(\cdot)+\psi)}-e^{V_{j+1}(B,\cdot)+\delta\mathcal{E}_j(LB)}-e^{\dot V_j(L B,\mathrm{S}_L(\cdot)+\psi)}+e^{\dot V_{j+1}(B,\cdot)+\delta\dot{\mathcal{E}}_j(L B)}\Vert_{\mathfrak{h};\langle\phi\rangle;\mathcal{C}^2(B)}\nonumber\\
    &\leqslant \mathcal{O}_L^{\vert X\vert}\big(\Vert\zeta_{0;j},\dot{\zeta}_j\Vert_{\wp} + \Vert K_j,\dot K_j\Vert_{\mathfrak{h},\kappa,A}\big)\cdot\\
    &\hspace{20pt}\Big(\Vert\zeta_{0;j}-\dot\zeta_{0;j}\Vert_{\wp}+\Vert K_j-\dot K_j\Vert_{\mathfrak{h},\kappa, A} + L^{-j}\vert\varepsilon-\dot\varepsilon\vert\big(\Vert\zeta_{0;j},\dot\zeta_{0;j}\Vert_{\wp}+\Vert K_j,\dot K_j\Vert_{\mathfrak{h},\kappa,A}\big)\Big).\nonumber
\end{align}

Getting the last inequality requires some computations which should be more or less standard by now. Finally, we remark that the factors $\mathcal{O}_L$ which are generated can be used to absorb $\varrho^{(2)}_{\varepsilon;j}$ in the special case $j = -\mathsf{N}-1, \varepsilon = \dot\varepsilon$.

\subsection{Proof of Proposition \ref{MainFlowEstimates}}
We decompose
\begin{equation}
    K_{j+1} - \dot{K}_{j+1} = (K_{j+1}[\varepsilon,\zeta]-K_{j+1}[\varepsilon,\dot\zeta])+(K_{j+1}[\varepsilon,\dot\zeta]-K_{j+1}[\dot\varepsilon,\dot\zeta]).
\end{equation}
For the first piece, we break the RG into the linear and nonlinear parts: $K_{j+1} = \mathcal{L}K_j + \mathcal{R}K_j$ where
\begin{align}
    \mathcal{L}K_j &= \mathcal{L}^{(1)}K_j+\mathcal{L}^{(2)}K_j+\mathcal{L}^{(3)}K_j,\nonumber\\
    \mathcal{R}K_j &= \mathcal{R}^{(1)}K_j+\mathcal{R}^{(2)}K_j+\mathcal{R}^{(3)}K_j+\mathcal{R}^{(4)}K_j+\mathcal{R}^{(5)}K_j+\mathcal{R}^{(6)}K_j;
\end{align}
and where
\begin{align}
    \mathcal{L}^{(1)}K_j &= \sum_{Y\in\mathcal{S}_j}^{\bar{Y}^L = X}\mathrm{Rem}_0\mathbb{E}\big[k_j^{(0)}(Y,\mathrm{S}_L\phi+\Psi_{\hspace{-1pt}j})\big],\nonumber\\
    \mathcal{L}^{(2)}K_j &= \sum_{q\neq 0}\sum_{Y\in\mathcal{S}_j}^{\bar{Y}^L = X}\mathbb{E}\big[k_j^{(q)}(Y,\mathrm{S}_L\phi+\Psi_{\hspace{-1pt}j})\big],\nonumber\\
    \mathcal{L}^{(3)}K_j &= \sum_{Y\in\mathcal{P}_{\mathrm{c};j}-\mathcal{S}_j}^{\bar{Y}^L=X}\mathbb{E}\big[K_j(Y,\mathrm{S}_L\phi+\Psi_{\hspace{-1pt}j})\big],\nonumber\\
    \mathcal{R}^{(1)}K_j &= \mathbb{I}_{\vert X\vert = 1}\big(\mathbb{E}\big[e^{V_j(L X-Y,\mathrm{S}_L\phi+\Psi_{\hspace{-1pt}j})}\big]-e^{V_{j+1}(X,\phi)+\delta\mathcal{E}_j(LX)}+\delta\mathcal{E}_j(LX)\big),\nonumber\\
    \mathcal{R}^{(2)}K_j &= \sum_{Y\in\mathcal{P}_{\mathrm{c};j}}^{\bar{Y}^L = X}\mathbb{E}\big[(e^{V_j(LX-Y,\mathrm{S}_L\phi+\Psi_{\hspace{-1pt}j})}-1)K_j(Y,\mathrm{S}_L\phi+\Psi_{\hspace{-1pt}j})\big],\nonumber\\
    \mathcal{R}^{(3)}K_j &= \sum_{Y\in\mathcal{P}_j}^{\bar{Y}^L = X,\vert \mathcal{C}(Y)\vert\geqslant 2}\mathbb{E}\bigg[e^{V_j(LX-Y,\mathrm{S}_L\phi+\Psi_{\hspace{-1pt}j})}\prod_{Y'\in\mathcal{C}(Y)}K_j(Y',\mathrm{S}_L\phi+\Psi_{\hspace{-1pt}j})\bigg],\nonumber\\
    \mathcal{R}^{(4)}K_j &= \mathbb{I}_{\vert X\vert\geqslant 2}\,\mathbb{E}\big[P_j^X(\phi,\Psi_{\hspace{-1pt}j})\big] + (e^{-\delta\mathcal{E}_j(LX)}-1)\mathbb{E}\big[P_j^X(\phi,\Psi_{\hspace{-1pt}j})\big],\nonumber\\
    \mathcal{R}^{(5)}K_j &= \sum_{(Z,Y_1,\mathfrak{B}\rightarrow Y_2)\rightarrow X}^{\vert\mathcal{C}(Y_1\cup Y_2)\vert\geqslant 1,\vert Z\vert + \vert\mathcal{C}(Y_1\cup Y_2)\vert\geqslant 2}\mathbb{E}\big[P_j^Z(\phi,\Psi_{\hspace{-1pt}j})R^{Y_1}_j(\phi,\Psi_{\hspace{-1pt}j})\big]J_j^{\mathfrak{B}\rightarrow Y_2},\\
    \mathcal{R}^{(6)}K_j &= \sum_{(Z,Y_1,\mathfrak{B}\rightarrow Y_2)\rightarrow X}^{\vert\mathcal{C}(Y_1\cup Y_2)\vert\geqslant 1}(e^{V_{j+1}(X-Z\cup Y_1\cup Y_2,\phi)-\delta\mathcal{E}_j(L(Z\cup Y_1\cup Y_2))}-1)\mathbb{E}\big[P_j^Z(\phi,\Psi_{\hspace{-1pt}j})R^{Y_1}_j(\phi,\Psi_{\hspace{-1pt}j})\big]J_j^{\mathfrak{B}\rightarrow Y_2}.\nonumber
\end{align}
We first control the linear parts. For $\mathcal{L}^{(1)}K_j$, compute
\begin{align}
    \Vert\mathcal{L}^{(1)}K_j[\varepsilon,\zeta]&-\mathcal{L}^{(1)}K_j[\varepsilon,\dot\zeta]\Vert_{\mathfrak{h};\langle X,\phi\rangle} \leqslant \sum_{Y\in\mathcal{S}}^{\bar{Y}^L = X}\big\Vert \mathrm{Rem}_0\mathbb{E}\big[(k_j^{(0)}[\varepsilon,\zeta]-k_j^{(0)}[\varepsilon,\dot\zeta])(Y,\mathrm{S}_L(\cdot)+\Psi_j)\big]\big\Vert_{\mathfrak{h};\langle\phi\rangle;\mathcal{C}^2(\bar{Y}^L)}\nonumber\\
    &\lesssim_1 L^{-2}\log(L)G_\kappa(X,\phi)\sum_{Y\in\mathcal{S}}^{\bar{Y}^L = X}\Vert K_j[\varepsilon,\zeta]-K_j[\varepsilon,\dot\zeta]\Vert_{\mathfrak{h},\kappa;\langle Y\rangle}\nonumber\\
    &\lesssim_1 A^{-\vert X\vert}\log(L)G_\kappa(X,\phi)\Vert K_j[\varepsilon,\zeta]-K_j[\varepsilon,\dot\zeta]\Vert_{\mathfrak{h},\kappa;\langle LX\rangle}
\end{align}
For $\mathcal{L}^{(2)}K_j$,
\begin{align}
    \Vert\mathcal{L}^{(2)}K_j[\varepsilon,\zeta]&-\mathcal{L}^{(2)}K_j[\varepsilon,\dot\zeta]\Vert_{\mathfrak{h};\langle X,\phi\rangle} \leqslant \sum_{q\neq 0}\sum_{Y\in\mathcal{S}}^{\bar{Y}^L = X}\big\Vert \mathbb{E}\big[(k_j^{(q)}[\varepsilon,\zeta]-k_j^{(q)}[\varepsilon,\dot\zeta])(Y,\mathrm{S}_L(\cdot)+\Psi_j)\big]\big\Vert_{\mathfrak{h};\langle\phi\rangle;\mathcal{C}^2(\bar{Y}^L)}\nonumber\\
    &\lesssim_1 L^2{e^{-\frac{1}{2}\beta\Gamma_{\hspace{-1pt}j}(0)}}A^{-\vert X\vert}G_\kappa(X,\phi)\Vert K_j[\varepsilon,\zeta]-K_j[\varepsilon,\dot\zeta]\Vert_{\mathfrak{h},\kappa;\langle LX\rangle}\sum_{n=0}^\infty e^{-(\sqrtbetae\Gamma_{\hspace{-1pt}j}(0)-\mathfrak{h})\sqrtbetae n}\nonumber\\
    &\lesssim_1 L^{2\sigma}A^{-\vert X\vert}G_\kappa(X,\phi)\Vert K_j[\varepsilon,\zeta]-K_j[\varepsilon,\dot\zeta]\Vert_{\mathfrak{h},\kappa;\langle LX\rangle}.
\end{align}
Finally, for $\mathcal{L}^{(3)}K_j$,

\begin{align}
    \Vert\mathcal{L}^{(3)}K_j[\varepsilon,\zeta]&-\mathcal{L}^{(2)}K_j[\varepsilon,\dot\zeta]\Vert_{\mathfrak{h};\langle X,\phi\rangle} \leqslant \sum_{Y\in\mathcal{P}_{\mathrm{c};j}-\mathcal{S}_j}^{\bar{Y}^L = X}\big\Vert \mathbb{E}\big[(K_j[\varepsilon,\zeta]-K_j[\varepsilon,\dot\zeta])(Y,\mathrm{S}_L(\cdot)+\Psi_j)\big]\big\Vert_{\mathfrak{h};\langle\phi\rangle;\mathcal{C}^2(\bar{Y}^L)}\nonumber\\
    &\leqslant \mathcal{O}_L^{\vert X\vert}A^{-(1+\eta)\vert X\vert}G_\kappa(X,\phi)\Vert K_j[\varepsilon,\zeta]-K_j[\varepsilon,\dot\zeta]\Vert_{\mathfrak{h},\kappa,A}.
\end{align}
For the nonlinear terms, there is a bit of computation. We only treat one of the more complicated terms, since the rest are very similar or even easier. Precisely, we focus on $\mathcal{R}^{(6)}$. First we expand ($W = Z \cup Y_1\cup Y_2$ for brevity)
\begin{align}
    (e&^{V_j[\varepsilon,\zeta](X-W,\phi)-\delta\mathcal{E}_j[\varepsilon,\zeta](LW)}-1)\mathbb{E}\big[P_j^Z[\varepsilon,\zeta](\phi,\Psi_{\hspace{-1pt}j})R_j^{Y_1}[\varepsilon,\zeta](\phi,\Psi_{\hspace{-1pt}j})\big]J_j^{\mathfrak{B}\rightarrow Y_2}[\varepsilon,\zeta]\nonumber\\
    &\hspace{5pt}-(e^{V_j[\varepsilon,\dot\zeta](X-W,\phi)-\delta\mathcal{E}_j[\varepsilon,\dot\zeta](LW)}-1)\mathbb{E}\big[P_j^Z[\varepsilon,\dot\zeta](\phi,\Psi_{\hspace{-1pt}j})R_j^{Y_1}[\varepsilon,\dot\zeta](\phi,\Psi_{\hspace{-1pt}j})\big]J_j^{\mathfrak{B}\rightarrow Y_2}[\varepsilon,\dot\zeta]\label{longIdentity}\\
    &=e^{V_j[\varepsilon,\zeta](X-W,\phi)}(e^{-\delta\mathcal{E}_j[\varepsilon,\zeta](LW)}-e^{-\delta\mathcal{E}_j[\varepsilon,\dot\zeta](LW)})\mathbb{E}\big[P_j^Z[\varepsilon,\zeta](\phi,\Psi_{\hspace{-1pt}j})R^{Y_1}_j[\varepsilon,\zeta](\phi,\Psi_{\hspace{-1pt}j})\big]J_j^{\mathfrak{B}\rightarrow Y_2}[\varepsilon,\zeta]+\nonumber\\
    &\hspace{15pt}(e^{V_j[\varepsilon,\zeta]}-e^{V_j[\varepsilon,\dot\zeta]})(X-W,\phi)e^{-\delta\mathcal{E}_j[\varepsilon,\dot\zeta](LW)}\mathbb{E}\big[P_j^Z[\varepsilon,\zeta](\phi,\Psi_{\hspace{-1pt}j})R^{Y_1}_j[\varepsilon,\zeta](\phi,\Psi_{\hspace{-1pt}j})\big]J_j^{\mathfrak{B}\rightarrow Y_2}[\varepsilon,\zeta]+\nonumber\\
    &\hspace{15pt}(e^{V_j[\varepsilon,\dot\zeta](X-W,\phi)-\delta\mathcal{E}_j[\varepsilon,\dot\zeta](LW)}-1)\mathbb{E}\big[(P_j^Z[\varepsilon,\zeta]-P_j^Z[\varepsilon,\dot\zeta])(\phi,\Psi_{\hspace{-1pt}j})R^{Y_1}_j[\varepsilon,\zeta](\phi,\Psi_{\hspace{-1pt}j})\big]J_j^{\mathfrak{B}\rightarrow Y_2}[\varepsilon,\zeta]+\nonumber\\
    &\hspace{15pt}(e^{V_j[\varepsilon,\dot\zeta](X-W,\phi)-\delta\mathcal{E}_j[\varepsilon,\dot\zeta](LW)}-1)\mathbb{E}\big[P_j^Z[\varepsilon,\dot\zeta](\phi,\Psi_{\hspace{-1pt}j})(R^{Y_1}_j[\varepsilon,\zeta]-R^{Y_1}_j[\varepsilon,\dot\zeta])(\phi,\Psi_{\hspace{-1pt}j})\big]J_j^{\mathfrak{B}\rightarrow Y_2}[\varepsilon,\zeta]+\nonumber\\
    &\hspace{15pt}(e^{V_j[\varepsilon,\dot\zeta](X-W,\phi)-\delta\mathcal{E}_j[\varepsilon,\dot\zeta](LW)}-1)\mathbb{E}\big[P_j^Z[\varepsilon,\dot\zeta](\phi,\Psi_{\hspace{-1pt}j})R^{Y_1}_j[\varepsilon,\dot\zeta](\phi,\Psi_{\hspace{-1pt}j})\big](J_j^{\mathfrak{B}\rightarrow Y_2}[\varepsilon,\zeta]-J_j^{\mathfrak{B}\rightarrow Y_2}[\varepsilon,\dot\zeta]).\nonumber
\end{align}
Then taking norms,
\begin{align}
    \Vert\!\cdot\!\Vert_{\mathfrak{h};\langle X,\phi\rangle} \leqslant \mathcal{O}_L^{\vert X\vert}&\mathcal{O}_A A^{-(1+\eta/4)\vert X\vert}G_\kappa(X,\phi)\cdot\nonumber\\
    &\Vert\zeta_{0;j},\dot{\zeta}_{0;j}\Vert_{\wp}\big(\Vert\zeta_{0;j}-\dot{\zeta}_{0;j}\Vert_{\wp} + \Vert K_j[\varepsilon,\zeta]-K_j[\varepsilon,\dot\zeta]\Vert_{\mathfrak{h},\kappa,A}\big).
\end{align}
Summing gives the kind of term required for the first part.

For the second part, we proceed similarly, but now there is no need to isolate the linear terms. Hence we take the basic definition of $K_{j+1}$, write down an identity like (\ref{longIdentity}), and find a bound by the quantity
\begin{equation}
    L^{-j}\vert \varepsilon-\dot\varepsilon\vert \mathcal{O}_L^{\vert X\vert}A^{-(1+\eta/3)\vert Y_1\cup Y_2^*\vert}G_{2\kappa}(Y_1,\phi)\Vert\zeta_{0;j},\dot{\zeta}_{0;j}\Vert_{\wp}^{\vert Z\vert}\big(\mathcal{O}_A\Vert K_j[\varepsilon,\dot\zeta],K_j[\dot\varepsilon,\dot\zeta]\Vert_{2\mathfrak{h},2\kappa,A}\big)^{\!\vert\mathcal{C}(Y_1\cup Y_2)\vert}.
\end{equation}
Ultimately, by summing over sets, the proposition follows.

A final remark as regards $j=-\mathsf{N}-1$: since
\begin{equation}
    K_{-\mathsf{N}}(X,\phi) = \mathbb{E}\big[P^X_{-\mathsf{N}-1}(\phi,\Psi_{\hspace{-1pt}-\mathsf{N}-1})\big],
\end{equation}
much simpler computations than above verify this special case.

\section{Laplace transforms} \label{AppLaplace}

We are considering probability measures $\mu$ on $\mathcal{D}'(\mathbb{T}^d)$ for which the Laplace transform
\begin{equation}
    \mathrm{Z}_\mu(f) = \int e^{\langle\phi,f\rangle}\mu(\dd\phi)
\end{equation}
belongs to $\mathfrak{C}$. In this case the complex Laplace transform
\begin{equation}
    \mathrm{Z}_\mu^{\mathbb{C}}(f) = \int e^{\langle\phi,\mathrm{Re}(f)\rangle + i\langle\phi,\mathrm{Im}(f)\rangle}\mu(\dd\phi)
\end{equation}
is still Fr\'echet-analytic on $\mathcal{C}^\infty(\mathbb{T}^d;\mathbb{C})$. We want to understand how uniform convergence of Laplace transforms as continuous functions gives us information about weak convergence of measures. As far as general notions of infinite-dimensional complex analysis go, \cite{Dineen1999} is a comprehensive reference.

\begin{prop}
    Suppose $\mu_n$ is a sequence of probability measures on $\mathcal{D}'(\mathbb{T}^d)$ with $\mathrm{Z}_{\mu_n}\in\mathfrak{C}$. Suppose further that $\mathrm{Z}_{\mu_n}$ converge in norm to some functional $\mathrm{Z}$. Then $\mathrm{Z}$ is the Laplace transform of a probability measure $\mu$ with Laplace transform in $\mathfrak{C}$, $\mu_n\rightarrow\mu$ weakly, and $\mathrm{Z}_\mu = \mathrm{Z}$.
\end{prop}

\begin{proof}
    Consider some vectors $f_1,\ldots,f_r,g_1,\ldots,g_r \in\mathcal{C}^\infty$. Define
    \begin{equation}
        F_n(z_1,\ldots,z_{2r}) = \mathrm{Z}^\mathbb{C}_{\mu_n}(z_1 f_1+\ldots+z_r f_r + z_{r+1} g_1+\ldots+z_{2r}g_r).
    \end{equation}
    These are holomorphic functions on $\mathbb{C}^{2r}$. On the real slice $\mathbb{R}^{2r}$, we know that $F_n$ converge locally uniformly to some $G$ defined by $\mathrm{Z}$. Now we want to show that every subsequence of $F_n$ contains a subsubsequence which converges to some $F$ (independent of the subsequence) and this will establish convergence.

    Any subsequence $F_{n_k}$ does contain a subsubsequence converging to some holomorphic $F$ on $\mathbb{C}^{2r}$ by compactness (since $F_n$ are locally bounded and hence a normal family). Suppose that some other subsubsequence converges to an $F'$ which might be different from $F$. But $F$ and $F'$ coincide on $\mathbb{R}^{2r}$ and hence by the identity principle $F = F'$.

    With this established, it is immediate that $\mathrm{Z}^{\mathbb{C}}_{\mu_n}$ converge pointwise to some function $\mathrm{Z}^{\mathbb{C}}$ bounded by $G_{2\kappa_*}$ which extends $\mathrm{Z}$. It is also clear that it is G\^ateaux-holomorphic:
    \begin{equation}
        \mathrm{Z}^{\mathbb{C}}(u + iv+z_1(f_1+i g_1)+\ldots+(f_r+ig_r)) = \lim_{n\rightarrow\infty}F_n(1,i,z_1,\ldots,z_r,z_1,\ldots,z_r),
    \end{equation}
    which is a uniform-on-compacts limit of restrictions of holomorphic functions to affine subspaces. Therefore the G\^ateaux derivatives indeed exist. But a locally bounded G\^ateaux-holomorphic function is Fr\'echet-holomorphic. In particular, this means that the characteristic functions $\mathrm{Z}^{\mathbb{C}}_{\mu_n}(i f)$ converge pointwise to a function continuous at $x = 0$, so by the Lévy-Fernique theorem for nuclear spaces (cf. e.g. \cite{LevyFernique}, Theorem 5.1 and references), $\mu_n$ converge weakly to a measure $\mu$ with characteristic function $\mathrm{Z}^\mathbb{C}(i f)$. Since this characteristic function has a holomorphic extension, Cauchy bounds lead to moment estimates which suffice to conclude that $\mu$ indeed has a Laplace transform. Again by an identity principle, $\mathrm{Z}_\mu = \mathrm{Z}$. For completeness, we provide the details of this last step in the next proposition.
\end{proof}

\begin{lem}
    Suppose the characteristic function $\Phi$ of a probability measure $\mu$ on $\mathcal{D}'(\mathbb{T}^d)$ has a G\^ateaux-holomorphic extension to $\mathcal{C}^\infty(\mathbb{T}^d;\mathbb{C})$. Then this is the (rotated) Laplace transform of $\mu$.
\end{lem}

\begin{proof}
    The main point is to establish moment bounds. Let $\mathrm{ev}_{\! f}$ by the evaluation functional at $f\in\mathcal{C}^\infty$, which becomes a random variable under $\mu$. Its characteristic function $\varphi_{\hspace{-1pt}f}(t) = \mu(e^{i t\,\mathrm{ev}_{\! f}})$ satisfies
    \begin{equation}
        \varphi_{\hspace{-1pt}f}(t) = \Phi(t f).
    \end{equation}
    Then $\varphi_{\hspace{-1pt}f}$ admits a holomorphic extension to $\mathbb{C}$. A well-known result (\cite{Klenke2020}, Theorem 15.35) states that since $\varphi_{\hspace{-1pt}f}$ is smooth,
    \begin{equation}
        \mathbb{E}\big[\vert\mathrm{ev}_{\! f}\vert^{2n}\big] \leqslant 2n\vert\varphi_{\hspace{-1pt}f}^{(2n)}(0)\vert.
    \end{equation}
    For odd moments we use
    \begin{equation}
        \mathbb{E}\big[\vert\mathrm{ev}_{\! f}\vert^{2n+1}\big] \leqslant 1+\mathbb{E}\big[\vert\mathrm{ev}_{\! f}\vert^{2n+2}\big] \leqslant 1 + (2n+2)\vert\varphi_{\hspace{-1pt}f}^{(2n+2)}(0)\vert.
    \end{equation}
    Integrating over a circle of radius $r$ gives
    \begin{equation}
        \varphi_{\hspace{-1pt}f}^{(n)}(0) = \frac{n!}{2\pi i}\oint_{\vert z\vert = r} \frac{\varphi_{\hspace{-1pt}f}(z)}{z^{n+1}}\dd z,
    \end{equation}
    that is,
    \begin{equation}
        \vert\varphi_{\hspace{-1pt}f}^{(2n)}(0)\vert\leqslant (2n)! M_{f,r}r^{-2n}
    \end{equation}
    where $M_{f,r}$ is an upper bound on $\varphi_{\hspace{-1pt}f}$ on $\vert z\vert \leqslant r$. It follows that
    \begin{align}
        \mathbb{E}\big[\vert\mathrm{ev}_{\! f}\vert^{2n}\big]&\leqslant 2n(2n)! M_{f,r}r^{-2n},\nonumber\\
        \mathbb{E}\big[\vert\mathrm{ev}_{\! f}\vert^{2n+1}\big]&\leqslant 1+ (2n+2)(2n+2)! M_{f,r}r^{-2n-2}.
    \end{align}
    But we can take $r = 2$ and then the integral
    \begin{equation}
        \int e^{\vert\langle\phi,f\rangle\vert}\mu(\dd\phi) = \sum_{k=0}^\infty \frac{1}{k!}\mathbb{E}\big[\vert \mathrm{ev}_{\! f}\vert^k\big] \leqslant e + M_{f,2}\sum_{k=0}^\infty (k+1)^22^{-k}
    \end{equation}
    converges. Thus $\mu$ has a Laplace transform which has to be the given extension of $\Phi$ by the identity principle (applied, as before, to \enquote{finite-dimensional slices} of $\Phi$).
\end{proof}

Finally, we state a result about the tails of the measure given a Gaussian bound on the Laplace transform. The result below can probably be formulated more abstractly, but we provide a direct proof avoiding heavy machinery.

\begin{prop}
    Suppose $\mu$ is a measure on $\mathcal{D}'(\mathbb{T}^d)$ with Laplace transform bounded as
    \begin{equation}
        \mathrm{Z}_\mu(f) \leqslant c\,e^{\mathsf{p}(f)^\mathfrak{q}}
    \end{equation}
    where $\mathsf{p}$ is some continuous seminorm on $\mathcal{C}^\infty(\mathbb{T}^d)$, $c > 0$, and $\mathfrak{q}\in (1,\infty)$ with Hölder conjugate $\mathfrak{p} = \frac{\mathfrak{q}}{\mathfrak{q}-1}$. Then there are an $\mathrm{L}^\mathfrak{p}$-based Sobolev norm $\Vert\!\cdot\!\Vert_{\natural}$ and $\lambda > 0$ small such that
    \begin{equation}
        \mu\big(e^{\lambda\Vert\cdot\Vert^{\mathfrak{p}}_\natural}\big) < \infty.
    \end{equation}
\end{prop}

\begin{proof}
    Let $\mathfrak{F}$ be the Fourier transform, which induces an isomorphism of the dual pairs $(\mathcal{C}^\infty(\mathbb{T}^d),\mathcal{D}'(\mathbb{T}^d))$ and $(\mathcal{S}_{\mathrm{i}}^{\vphantom{a}}(\mathbb{Z}^d;\mathbb{C}),\mathcal{S}_{\mathrm{i}}'(\mathbb{Z}^d;\mathbb{C}))$. The second couple are just function spaces on $\mathbb{Z}^d$ with the natural polynomial growth conditions, where the subscript indicates that the functions in question satisfy $f(-x) = \overline{\hspace{-1pt}f(x)\hspace{-1pt}}\,$. The duality pairing is explicitly given by
    \begin{equation}
        \langle f,g\rangle_{\mathcal{S}^{\vphantom{a}}_\mathrm{i}\hspace{-0.6pt}(\mathbb{Z}^d),\mathcal{S}'_{\mathrm{i}}\hspace{-0.6pt}(\mathbb{Z}^d)} = \sum_{x\in\mathbb{Z}^d} \overline{\hspace{-1pt}f(x)\hspace{-1pt}}\,g(x).
    \end{equation}
    A natural system of norms is given by ($p\in[1,\infty],s\in\mathbb{R}$)
    \begin{equation}
    \Vert f\Vert_{s,p} = \Vert (1+\vert\cdot\vert_\infty)^s f\Vert_{\ell^p(\mathbb{Z}^d;\mathbb{C})}.
    \end{equation}
    If we consider the pushforward measure $\mu_* =\mathfrak{F}_*\mu$, we find
    \begin{equation}
        \mathrm{Z}_{\mu_*}\!(f) = \int e^{\langle f,\mathfrak{F}\phi\rangle}\mu(\dd\phi) = \int e^{\langle \mathfrak{F}^{-1}\!f,\phi\rangle}\mu(\dd\phi) = \mathrm{Z}_\mu(\mathfrak{F}^{-1}\hspace{-1pt}f).
    \end{equation}
    Thus, there are $s\in\mathbb{R}$, $\alpha > 0$ such that
    \begin{equation}
        \mathrm{Z}_{\mu_*}\!(f) \leqslant c\,e^{\alpha\Vert f\Vert_{s,\mathfrak{q}}^\mathfrak{q}/\mathfrak{q}}.
    \end{equation}
    If we define
    \begin{equation}
        \delta_{x}^{\mathrm{Re}} = \frac{\delta_x+\delta_{-x}}{2},\quad \delta_{x}^{\mathrm{Im}} = \frac{\delta_x-\delta_{-x}}{2i},
    \end{equation}
    then
    \begin{equation}
        \vert f(x)\vert^p \leqslant 2^{p-1}\big(\vert\langle \delta_{x}^{\mathrm{Re}},f\rangle\vert^p + \vert\langle \delta_{x}^{\mathrm{Im}},f\rangle\vert^p\big)
    \end{equation}
    and we can use Hölder's inequality to compute
    \begin{align}
        \mu_*&\hspace{-1pt}\big(e^{\lambda\Vert\cdot\Vert_{-r,\mathfrak{p}}^\mathfrak{p}}\hspace{-1pt}\big) \nonumber\\
        &\leqslant
        \sum_{n=0}^\infty \frac{1}{n!}\big(2^{\mathfrak{p}-1}\lambda\big)^{\! n}\sum_{x_1,\ldots,x_n\in\mathbb{Z}^d,u_1,\ldots,u_n\in\lcurl\mathrm{Re},\mathrm{Im}\rcurl}(1+\vert x_1\vert_\infty)^{-r\mathfrak{p}}\ldots (1+\vert x_n\vert_\infty)^{-r\mathfrak{p}}\mu_*\hspace{-1pt}\big(\vert\langle \delta^{u_1}_{x_1},\cdot\rangle\vert^\mathfrak{p}\ldots\vert\langle \delta^{u_n}_{x_n},\cdot\rangle\vert^\mathfrak{p}\big)\nonumber\\
        &\leqslant\sum_{n=0}^\infty \frac{1}{n!}\bigg(2^{\mathfrak{p}-1}\lambda\sum_{x\in\mathbb{Z}^d,u\in\lcurl\mathrm{Re},\mathrm{Im}\rcurl}(1+\vert x\vert_\infty)^{-r\mathfrak{p}}\mu_*\hspace{-1pt}\big(\vert\langle \delta^u_{x},\cdot\rangle\vert^{n\mathfrak{p}}\big)^{\!1/n}\bigg)^{\! n}.
    \end{align}
    Apply Markov's inequality with $t \geqslant 0$ arbitrary:
    \begin{align}
        \mu_*\hspace{-1pt}\big(\lcurl \vert\langle \delta^u_{x},\cdot\rangle\vert\geqslant t\rcurl\big) &\leqslant \inf_{y\in[0,\infty)}\frac{\mu_*\hspace{-1pt}\big(\!\cosh(y\hspace{1pt}\langle\delta^u_{x},\cdot\rangle)\big)}{\cosh(yt)}\nonumber\\
        &\leqslant \inf_{y\in[0,\infty)}e^{-yt}\big(\mathrm{Z}_{\mu_*}\!(t\delta^u_{x})+\mathrm{Z}_{\mu_*}\!(-t\delta^u_{x})\big)\nonumber\\
        &\leqslant 2c\inf_{y\in[0,\infty)} e^{2^{1-\mathfrak{q}}\alpha y^\mathfrak{q}(1+\vert x\vert_\infty)^{s\mathfrak{q}}/\mathfrak{q}-y t}\nonumber\\
        &= 2c\,e^{-2\mathfrak{p}^{-1} \alpha^{-\mathfrak{p}/{\mathfrak{q}}}(1+\vert x\vert_\infty)^{-s\mathfrak{p}}t^{\mathfrak{p}}},
    \end{align}
    and therefore if $n\geqslant 1$,
    \begin{equation}
        \mu_*\hspace{-1pt}\big(\vert\langle \delta^u_{x},\cdot\rangle\vert^{n\mathfrak{p}}\big) = \int_0^\infty \mu_*\hspace{-1pt}\big(\lcurl \vert\langle \delta^u_{x},\cdot\rangle\vert\geqslant t^{1/n\mathfrak{p}}\rcurl\big)\dd t
        \leqslant2c\,n!\big(2^{-1}\mathfrak{p} \alpha^{\mathfrak{p}/{\mathfrak{q}}}(1+\vert x\vert_\infty)^{s\mathfrak{p}}\big)^{\! n}.
    \end{equation}
    Using this, it follows that if $r > s+d/\mathfrak{p}$ and $\lambda$ small enough,
    \begin{align}
    \mu_*\hspace{-1pt}\big(e^{\lambda\Vert\cdot\Vert_{-r,\mathfrak{p}}^\mathfrak{p}}\hspace{-1pt}\big)
    \leqslant c + 2c\sum_{n=1}^\infty \bigg(2^{d}d\,\mathfrak{p} \alpha^{\mathfrak{p}/{\mathfrak{q}}}\lambda\sum_{k=1}^\infty k^{-(r-s)\mathfrak{p}+d-1}\bigg)^{\! n} <\infty.
    \end{align}
    $\Vert\!\cdot\!\Vert_\natural = \Vert\mathfrak{F}(\cdot)\Vert_{-r,\mathfrak{p}}$ is the Sobolev norm we were after.
\end{proof}

\section*{Acknowledgments}

I thank Roland Bauerschmidt, Jon Dimock, Massimiliano Gubinelli, Peter Paulovics, and Sarah-Jean Meyer for useful input. Special thanks go to David Brydges for his help and many enlightening conversations without which this paper might not have appeared. The work was supported by Oxford's Clarendon Fund.

\printbibliography

\end{document}